\documentclass[a4paper]{article}
\usepackage{enumitem}
\usepackage{calc}

\usepackage{lmodern} %Type1-font for non-english texts and characters
\usepackage{graphicx} %%For loading graphic files

\usepackage{amsmath, accents}
\usepackage{amssymb,tikz}
\usepackage{pict2e}
\usepackage{amsthm}
\usepackage{amscd}
\usepackage{mathrsfs}
\usepackage{todonotes}
\usetikzlibrary{calc} % This is to add two points in tikzpictures

\usepackage{yfonts}

\usepackage{marvosym}
\usepackage[percent]{overpic}

\newtheorem{theorem}{Theorem} [section]
\newtheorem{proposition}[theorem]{Proposition}	
	
\newtheorem{lemma}[theorem]{Lemma}

\newtheorem{remark}[theorem]{Remark}

\theoremstyle{definition}
\newtheorem{definition}[theorem]{Definition}

\newcommand{\C}{\mathbb{C}}

\newcommand{\R}{\mathbb{R}}
\newcommand{\N}{\mathbb{N}}

\newcommand{\supp}{\operatorname{supp}}
\newcommand{\re}{\text{\upshape Re\,}}

\newcommand{\erfc}{\mathrm{erfc}\,}

\tikzset{
	master/.style={
		execute at end picture={
			\coordinate (lower right) at (current bounding box.south east);
			\coordinate (upper left) at (current bounding box.north west);
		}
	},
	slave/.style={
		execute at end picture={
			\pgfresetboundingbox
			\path (upper left) rectangle (lower right);
		}
	}
}

\let\oldbibliography\thebibliography
\renewcommand{\thebibliography}[1]{\oldbibliography{#1}
\setlength{\itemsep}{-0.5pt}}

\def\XXint#1#2#3{{\setbox0=\hbox{$#1{#2#3}{\int}$}
\vcenter{\hbox{$#2#3$}}\kern-.5\wd0}}

\usepackage{tikz}
\usetikzlibrary{arrows}
\usetikzlibrary{decorations.pathmorphing}
\usetikzlibrary{decorations.markings}
\usetikzlibrary{patterns}
\usetikzlibrary{automata}
\usetikzlibrary{positioning}
\usepackage{tikz-cd}
\tikzset{->-/.style={decoration={
				markings,
				mark=at position #1 with {\arrow{latex}}},postaction={decorate}}}
	
	\tikzset{-<-/.style={decoration={
				markings,
				mark=at position #1 with {\arrowreversed{latex}}},postaction={decorate}}}

\usetikzlibrary{shapes.misc}\tikzset{cross/.style={cross out, draw,
         minimum size=2*(#1-\pgflinewidth),
         inner sep=0pt, outer sep=0pt}}

\usepackage{pgfplots}
\allowdisplaybreaks	

\numberwithin{equation}{section}

\def\bigO{{\cal O}}

\usepackage[colorlinks=true]{hyperref}
\hypersetup{urlcolor=blue, citecolor=red, linkcolor=blue}

\begin{document}
\title{\vspace{-1cm} Random normal matrices: eigenvalue \\ correlations near a hard wall\\ }
\author{Yacin Ameur\footnote{Centre for Mathematical Sciences, Lund University, 22100 Lund, Sweden. e-mails: yacin.ameur@math.lu.se,  christophe.charlier@math.lu.se, joakim.cronvall@math.lu.se
}\,, Christophe Charlier$^{*}$\hspace{-0.09cm},  Joakim Cronvall$^{*}$}

\maketitle

\begin{abstract} We study pair correlation functions for planar Coulomb systems in the pushed phase, near a ring-shaped impenetrable wall.
We assume coupling constant $\Gamma=2$ and that the number $n$ of particles is large. We find that the correlation functions decay slowly along the edges of the wall, in a narrow interface
stretching a distance of order $1/n$ from the hard edge. At distances much larger than $1/\sqrt{n}$, the effect of the hard wall is negligible and pair correlation functions decay very quickly, and in between sits an interpolating interface that we call the ``semi-hard edge''.

More precisely, we provide asymptotics for the correlation kernel $K_{n}(z,w)$ as $n\to\infty$ in two microscopic regimes (with either $|z-w| = \bigO (1/\sqrt{n})$ or $|z-w| = \bigO (1/n)$), as well as in three macroscopic regimes (with $|z-w| \asymp 1$). For some of these regimes, the asymptotics
 involve
 oscillatory theta functions and weighted Szeg\H{o} kernels.
\end{abstract}



\section{Introduction: the pushed phase Coulomb gas}\label{section:pushed phase}
Hard edge boundary conditions are well-known in the theory of Hermitian random matrices, where
they are, for instance, associated with the Bessel kernel, see \cite{TW,CD2019,Fo}. In dimension two, the study of one-component plasmas near a hard wall is likewise of interest and variants appear in for example \cite{J,J2,Sm,RB,CPR,AR2017,AKM,AKMW,CMV,Fo,CFLV1,CFLV2,Seo0,Seo2021,CharlierAnnuli,ACCL2022 1,ACCL2022 2,Berezin,JV,CBalayage,CLBalayage}. As in the majority of these works, we take the coupling constant to be $\Gamma=2$, or equivalently, we consider eigenvalues of random normal matrices (see \cite{CZ1998,Z}).

\medskip  In \cite{J}, Jancovici considers two versions of the Ginibre ensemble,\footnote{Ginibre matrices are almost surely non-normal. Nevertheless, their eigenvalues are distributed in the same way as the eigenvalues of random normal matrices with the Gaussian potential $|z|^{2}$.} the first being the usual ``soft edge'' ensemble.
In the second model a wall is placed along the boundary of the droplet (nowadays this is called a ``soft/hard edge'').
One of Jancovici's motivations for studying a Coulomb system near a hard wall is its interest for describing an electrolyte near a colloidal wall or an electrode plate.

\medskip In both cases it was shown that the pair correlation functions along the edge
decay very slowly compared with the situation in the bulk. More precisely, they decay only as an inverse power of the distance along the edge. The heuristic is that the screening cloud that surrounds a particle sitting near the edge is prevented by the external field and/or the wall from being rotationally symmetric, which gives the particle plus cloud system a nonvanishing electrical dipole moment. This dipole moment is not strongly localized and causes long-range correlations along the edge \cite{J}.

\medskip In the soft edge-case, Jancovici's results were first generalized to the elliptic Ginibre ensemble in \cite{FJ} and later to quite general potentials in the paper \cite{AC}, where a leading term for macroscopic correlations near the outer boundary component of the droplet is found and expressed in terms of a Szeg\H{o} kernel. The paper \cite{ACC2022} provides corresponding results near the boundary of a ring-shaped
spectral gap, where some additional oscillations (depending on the number of particles) enter the picture. Microscopically, the oscillations are expressed in terms of the Jacobi theta function, which enters the subleading term of the microscopic density near the soft edge; macroscopically, the pair correlations along the boundary involve
an oscillatory Szeg\H{o} kernel.  The heuristic is that there is an additional uncertainty concerning the number of particles that fall near each of the two boundary components of the spectral gap, and this manifests itself in terms of some oscillations. It is observed in \cite{ACC2022} that the subleading term of the microscopic density near the edge is closely related to fluctuations of linear statistics.

\medskip The goal of the present work is to adapt the above results to the setting where a hard wall is placed inside the bulk of a plasma. This is known as a ``pushed phase'' \cite{CMV,CFLV1,CFLV2} or as a ``hard edge''  \cite{Seo2021,ACCL2022 1}, and is quite different from the ``critical phase'', i.e., the soft/hard edge in Jancovici's original work.
In the pushed phase, the plasma undergoes two transitions, at distances $\bigO(1/\sqrt{n})$ and $\bigO(1/n)$ respectively, from the hard edge.

\medskip

Before we continue, it is expedient to briefly recall the Coulomb gas model.
We consider Coulomb systems $\{z_j\}_{j=1}^n$ in the complex plane $\C$ subjected to an appropriate confining potential $Q$, where the Hamiltonian is
$$H_n=\sum_{i\ne j} \log \frac 1 {|z_i-z_j|}+n\sum_{j=1}^n Q(z_j),$$
i.e. it is the sum of the logarithmic interaction energy and the energy of interaction with the external field. The probability law of the system is given by the Gibbs measure
\begin{equation}\label{2d Coulomb gas}d \mathbb{P}_n(z_1,\ldots,z_n)=\frac 1 {n!Z_n}e^{-H_n(z_1,\ldots z_n)}\, \prod_{j=1}^n \frac {d^2 z_j}\pi\end{equation}
where $Z_n$ is the normalizing constant.

\medskip Classically, one approximates
the random measure $\frac 1 n \sum_{j=1}^n \delta_{z_j}$
by a continuous unit charge distribution $\sigma=\sigma[Q]$. This distribution is precisely given by Frostman's equilibrium measure associated with the potential $Q$, which is the unique minimizer $\sigma$ of the ``$Q$-energy''
\begin{align*}
\nu \mapsto I_Q[\nu]=\iint_{\C^2}\log \frac 1 {|z-w|}\, d\nu(z)\, d\nu(w)+\int_\C Q\, d\nu
\end{align*}
among all compactly supported unit measures $\nu$. The support $S=\supp\sigma$ is called the droplet in potential $Q$. (The measures $\frac{1}{n}\sum_{j=1}^n \delta_{z_j}$ converge, in a probabilistic and weak sense,
to $\sigma$ as $n\to\infty$, see e.g. \cite{Deift,A1}.)

\medskip If $Q$ is smooth in a neighbourhood of the droplet, then by Frostman's theorem (see \cite{SaTo}), $\sigma$ is absolutely continuous with respect to the area measure $d^2 z$ and takes the form $d\sigma(z)=\Delta Q(z)\chi_S(z)\,\frac {d^2 z}\pi$ where $\Delta Q:=\frac 1 4 (Q_{xx}+Q_{yy})$ is the usual Laplacian divided by four and $\chi_S$ is the indicator function of $S$. For example, if $Q(z)=|z|^2$ we obtain the Ginibre ensemble for which $S=\{z\,:\,|z|\le 1\}$ and $d\sigma(z)=\chi_S(z)\, \frac {d^2 z}\pi$.

\medskip We now fix a suitable (smooth) potential. In this paper we assume the rotational symmetry
$Q(z)=Q(|z|)$ and that the droplet is a disc $S$.
We also fix an open subset $G$ of the interior of $S$ which we take to be an annulus $r_1<|z|<r_2$ and modify the potential inside $G$ to $+\infty$, i.e. we put
\begin{align}\label{potential Q tilde}
\tilde{Q}(z)=\begin{cases}Q(z),& z\in \C\setminus G,\cr
+\infty,& z\in G.\cr
\end{cases}
\end{align}

Replacing $Q$ by $\tilde{Q}$ has a drastic effect even at the level of the equilibrium measure. Indeed, the equilibrium measure $\tilde{\sigma}$ associated with the potential $\tilde{Q}$ puts zero mass in the gap $G$,
 and the restriction $\sigma|_G$ is swept to a measure $\nu_G$ supported on the boundary of $G$ via a construction known as balayage \cite{SaTo}.
In short, the measure $\nu_G$ has the same total mass as $\sigma|_G$, is supported on the boundary of $G$ and takes the form $\nu_G=c_1\nu_1+c_2\nu_2$ where $\nu_j$ denotes the arclength measure along the circle $|z|=r_j$ and $c_j$ are certain constants, given in \eqref{def of taustar} below. The balayage measure $\nu_G$ and the restriction $\sigma|_G$ have furthermore identical logarithmic potentials in the complement $\C\setminus \overline{G}$.

\medskip The equilibrium measure $\tilde{\sigma}$ can now be written $$\tilde{\sigma}=\sigma|_{\tilde{S}}+\nu_G,$$ where
$d\sigma=\Delta Q(z)\frac {d^2 z} \pi$ and $\tilde{S}=S\setminus G$ is the droplet associated with the potential $\tilde{Q}$.

\medskip On the level of the Coulomb system $\{z_j\}_{j=1}^n$, the picture is that most of the particles that originally occupied the gap get ``swept'' to a thin interface in $\tilde{S}$ near the boundary $\partial G$, at a distance of $\bigO(1/n)$; following \cite{ACCL2022 1} we call this the ``hard edge regime''. If we denote this interface by $E_n$ then the corresponding  random measure $\frac 1 n\sum_{z_j\in E_n}\delta_{z_j}$ approximates the singular part $\nu_G$ of $\tilde{\sigma}$.

\medskip The 1-particle density of the pushed system is of order of magnitude $n^2$ in the hard edge regime. Further inside the bulk of $\tilde{S}$, at distances much larger than $1/\sqrt{n}$ from $\partial G$, the effect of the hard wall becomes negligible and the $1$-particle density is, to a first order approximation, given by $n$ times the equilibrium density $n \Delta Q(z)\frac {d^2 z} \pi$ of the unconstrained ensemble (associated with the potential $Q$).
In between these regimes sits a transitional regime which we call the ``semi-hard edge'', following \cite{ACCL2022 1}.

\medskip As we already mentioned, a different kind of wall, known as a soft/hard wall or a critical phase, is obtained by placing the hard wall along the boundary of the droplet $S$. This corresponds to
setting $G=\C\setminus S$ in \eqref{potential Q tilde}. In this case the local statistics near the wall, in a $\bigO(1/\sqrt{n})$-interface, is affected, but the
equilibrium measure and the droplet
are unchanged, and no hard-edge regime is present.

\medskip In studying hard walls, we must face the difficulty that the wall might cause long-range correlations merely due to its ``symmetry breaking''. A main insight in the forthcoming work \cite{C} is that, for ``general'' potentials, symmetry breaking is avoided precisely if the
hard wall is placed along the boundary of the droplet associated with the potential $Q/\tau$ where $0<\tau<1$ is a fixed suitable constant. In other words, the hard wall must be placed along the boundary of a mass-$\tau$ droplet associated with $Q$ (cf.~\cite{LM} for more about $\tau$-droplets). If this is done, some first order universality results can be proven; however, we emphasize that much more refined asymptotic results, such as the ones we obtain below, remain currently out of reach in this generality.
Hard walls which break the symmetry are so far studied mostly at the level of the equilibrium measure, see the works \cite{AR2017,CBalayage, CLBalayage}.

\medskip  For rotationally symmetric pushed phase models (including Coulomb gases in $\R^d$ and Yukawa gases) more work has been done.
 In \cite{CFLV2} it is proven that the weighted logarithmic energy $I_Q[\tilde{\sigma}]$ of the equilibrium measure exhibits a third-order phase transition as the wall crosses the critical phase and enters the pushed phase. In the planar case ($d=2$), a large $n$-expansion of the free-energy is proved in \cite{CharlierAnnuli}, allowing for the case of annular spectral gaps. The third order phase transition is reflected in the leading coefficient of that expansion.
In the work \cite{Seo2021}, Seo finds the leading order microscopic one-point density near a hard wall and proves it to be universal for a class of rotationally symmetric ensembles (which may have an additional  logarithmic singularity along the hard edge). Partition functions with hard edges are related with so-called \textit{large gap (or hole) probabilities}; this topic has a long history, see e.g. \cite{ForresterHoleProba, APS2009, SP2024}, and has found important recent applications to physics, see e.g. \cite{LMS2018}. Disc counting statistics are considered in \cite{ACCL2022 1,ACCL2022 2}, and \cite{Berezin} gives a functional limit theorem for certain smooth radially symmetric linear statistics in the hard edge regime, near the outer boundary of the droplet.

\medskip It is worth noting that a different type of hard-edge ensemble, corresponding to the external potential $Q=0$ and a hard wall outside some compact set $\Sigma$ bounded by a Jordan curve
has attracted some recent interest (e.g.~\cite{JV}). In this case the entire system will tend to occupy the hard wall regime, i.e. the portion of $\Sigma$ which has distance
$\bigO(1/n)$ to the boundary of $\Sigma$.
In a way, this case is simpler than the kind of hard walls studied in the present paper, since there is essentially just one regime (no ``bulk'' to interact with). In the case when $\Sigma$ is an elliptic disc (possibly with an added logarithmic singularity along the hard edge) correlations are studied in \cite{ACV, NAKP2020}, following the earlier work \cite{ZS2000} on truncated unitary matrices.

\medskip The present work gives a comprehensive study of correlations near a ring-shaped hard wall in the hard and semi-hard regimes on the microscopic and macroscopic levels. We shall consider the class of underlying rotationally symmetric potentials of the form \begin{equation}\label{originML}Q(z)=|z|^{2b}+\frac {2\alpha} n\log \frac 1 {|z|}\end{equation}
where $b>0$ and $\alpha>-1$. These are sometimes called model Mittag-Leffler potentials, and
they give rise to the droplets $S=\{z\,:\,|z|\le b^{-1/2b}\}$.
The corresponding point-processes \eqref{2d Coulomb gas} are well studied and are sometimes called model Mittag-Leffler ensembles, see \cite{AV2003, AKS2018, CharlierFH, BC2022} and the references there.

\medskip For the potentials \eqref{originML} and an annular spectral gap $G=\{r_1<|z|<r_2\}$ with $0<r_1<r_2<b^{-1/2b}$ we provide a full asymptotic picture of the correlations associated with the potential $\tilde{Q}$, cf.~Figures \ref{fig: ML with hard wall} and \ref{fig:summary} for illustrations.

\medskip In the rest of this paper, we will use the symbol $Q$ to denote the pushed-phase Mittag-Leffler potential (above denoted $\tilde{Q}$) which is $+\infty$ in $G$, and we will write $\mu$ for the equilibrium measure of $Q$ and $S=\supp\mu=\{|z|\le b^{-1/2b}\}\setminus G$ for the pushed-phase droplet, see Figure \ref{fig:summary}.

\section{Description of the model}

The $k$-point correlation functions $\{R_{n,k}:\C^{k}\to [0,+\infty)\}_{k= 1}^{n}$ associated with the point process $\{z_j\}_{j=1}^n$ in \eqref{2d Coulomb gas} are defined such that
\begin{align*}
\mathbb{E}_n[f(z_1,\ldots,z_k)]=\frac {(n-k)!}{n!}\int_{\C^k}f(z_1,\ldots,z_k)\, R_{n,k}(z_1,\ldots,z_k)\, \prod_{j=1}^{k}\frac{d^{2}z_{j}}{\pi}
\end{align*}
holds for all continuous and compactly supported functions $f$ on $\mathbb{C}^{k}$. The point process \eqref{2d Coulomb gas} is determinantal, meaning that all correlation functions exist and that there exists a correlation kernel $K_{n}(z,w)$ such that
\begin{align*}
R_{n,k}(z_{1},\ldots,z_{n}) = \det(K_{n}(z_{i},z_{j}))_{i,j=1}^{k}, \qquad \mbox{for all } n \in \N_{>0}, \; k\in \{1,2,\ldots,n\}, \; z_{1},\ldots,z_{k}\in \C.
\end{align*}

It should be noted that while the correlation functions are uniquely defined, a correlation kernel $K_n(z,w)$ is only defined up to a multiplicative ``cocycle'', i.e., a function of the form $c_n(z,w)=g_n(z)\overline{g_n(w)}$ where $g_n$ is a unimodular function. We fix $K_n$ uniquely by taking $K_n(z,w)$ to be the reproducing kernel of the $n$-dimensional subspace of $L^2$ consisting of weighted polynomials $p(z)e^{-nQ(z)/2}$ where $p$ is a holomorphic polynomial of degree at most $n-1$. This $K_n$ is called the \textit{canonical} correlation kernel and is used without exception in the following; see \eqref{def of Kn} for an explicit formula.

\medskip We shall study large $n$ asymptotics for pair correlations $K_n(z,w)$ when $z,w$ are close to a hard wall. However, before
specializing to our setting, it is convenient to recall a few well known asymptotic results
in other regimes, such as the bulk and soft edge regimes. Our discussion is far from exhaustive and we refer to the recent survey \cite{BF2022} for more background.

\medskip If $z_{0}$ is a regular point in the bulk (i.e. the interior of $S$) then \cite{AHM Duke} the rescaled kernels
\begin{align*}
\mathrm{K}_n(\mathrm{z},\mathrm{w})=\frac{1}{n\Delta Q(z_{0})} K_{n}\Big( z_{0} + \frac{\mathrm{z}}{(n\Delta Q(z_{0}))^{1/2}} , z_{0} + \frac{\mathrm{w}}{(n\Delta Q(z_{0}))^{1/2}} \Big), \qquad \mathrm{z},\mathrm{w}\in \C,
\end{align*}
converge as $n\to + \infty$ (after multiplication by suitable cocycles $c_n(\mathrm{z},\mathrm{w})$) to the Ginibre kernel $e^{\mathrm{z} \overline{\mathrm{w}}-|\mathrm{z}|^{2}/2-|\mathrm{w}|^{2}/2}$, and where we recall that $\Delta Q$ denotes the standard Laplacian divided by four, i.e. $\Delta Q=(Q_{xx}+Q_{yy})/4$.

\medskip

\medskip In the soft edge case, the large $n$ behavior of $K_{n}(z,w)$ is well understood when $z,w$ are close to the ``outer boundary'' of $S$, provided that this is an everywhere regular Jordan curve and that $\Delta Q>0$ along the curve. By the outer boundary, we mean the component $\partial U$ of $\partial S$ where $U$ (throughout) denotes the unbounded component of $\C\setminus S$.
In this case, the leading order asymptotics of $K_{n}(z,w)$ is found in e.g. \cite{FH, K2005, TV2015, AKM, HW} for the case when $z,w$ are at a distance of order $n^{-1/2}$ from the outer boundary $\partial U$ and such that $|z-w| = \bigO(n^{-1/2})$.\footnote{$f(n)=\bigO(g(n))$ means $f(n) \leq C g(n)$ for all large enough $n$ and $C>0$ is independent of $n$.} The papers \cite{LR,ACC2022} provide subleading corrections to the microscopic kernel near an outer boundary, and \cite{ACC2022} also gives such correction terms near the boundary of a ring-shaped spectral gap and relates them to fluctuations of linear statistics.

\medskip As we already noted in Section \ref{section:pushed phase},
two points $z, w$ near a (smooth) outer soft edge $\partial U$ are strongly correlated even if $|z-w|\asymp 1$,\footnote{$f(n)\asymp g(n)$ means $C_{1}g(n) \leq f(n)\leq C_{2} g(n)$ for all large enough $n$ and $C_{2}\geq C_{1}>0$ are independent of $n$.} i.e. even if $z, w$ lie at a macroscopic distance
of each other. (By contrast, if $z, w$ are in the bulk, $K_{n}(z, w)$ gets exponentially small as $n\to+\infty$ with $|z-w|\asymp 1$, and $K_{n}(z, w)$ has also a Gaussian decay in the distance $|z-w|$ for fixed $n$).
A formula for long-range correlations along the outer boundary $\partial U$ of the droplet is computed in \cite{AC}.
As noticed in \cite{ACC2022}, the situation is more subtle when $z,w$ are in a $n^{-1/2}$ neighborhood of the boundary of a (soft) spectral gap that is not simply connected: in this case, the asymptotics of $K_{n}(z,w)$ are oscillatory and described in terms of Jacobi $\theta$-functions (if $|z-w|=\bigO(n^{-1/2})$) or weighted Szeg\H{o} kernels (if $|z-w|\asymp 1$).

\medskip Some further related results for correlations along the boundary are found in \cite{ADM} (higher dimensional elliptic Ginibre ensemble) and \cite{BY2022} (lemniscate ensembles).

\medskip Scaling limits have also been investigated in other situations than the bulk and soft edge regimes, see e.g. \cite{AB2012, AKS2018} near root-type bulk singularities, \cite{AKM, AKMW, Seo0, BLY2021, KLM2023} near singular boundary points (such as cusp-like singularities and so-called local droplets), \cite{FKS1997, ACV, AB2021, BS2021} for bandlimited point processes, and \cite{J,CPR,RB,AKM,AKMW,AKS20} near soft/hard edges and other sorts of soft/hard boundary conditions.

\medskip Much less is known about pair correlations near the hard edge for pushed-phase ensembles (i.e. with a non-constant potential).
Indeed, to our knowledge, the only hitherto recorded results appear in the paper \cite{Seo2021}, where large $n$ asymptotics for
$K_{n}(z,w)$ is studied the microscopic regime where $z,w$ satisfy $|z-w|=\bigO(n^{-1})$ and are in a $n^{-1}$ neighborhood of $\partial U$ (with $\partial U$ being a hard edge), in the case of rotation-invariant $Q$ (which may have a weak logarithmic singularity along the hard edge).

\medskip Recall that the droplet associated with the Mittag-Leffler potential \eqref{originML} is given by $\{|z|\le b^{-1/2b}\}$. We now fix numbers $r_1,r_2$ with $0<r_1<r_2<b^{-1/2b}$
and redefine that potential to be $+\infty$ in the gap (or hard wall) $G:=\{z:|z|\in (r_{1},r_{2})\}$, i.e. we set
\begin{align}\label{def of Q hard edge}
Q(z) = \begin{cases}
|z|^{2b} - \frac{2\alpha}{n}\log |z|, & \mbox{if } |z| \notin (r_{1},r_{2}), \\
+\infty, & \mbox{if } |z| \in (r_{1},r_{2}),
\end{cases} \qquad b>0, \; \alpha > -1.
\end{align}
In this work, we focus on the corresponding point process:
\begin{align}\label{def of point process}
\frac{1}{n!\mathcal{Z}_{n}} \prod_{1 \leq j < k \leq n} |z_{k} -z_{j}|^{2} \prod_{j=1}^{n}e^{-n Q(z_{j})}\prod_{j=1}^n \frac {d^2 z_j} \pi.
\end{align}
\begin{figure}
\begin{center}
\begin{tikzpicture}[master]
\node at (0,0) {\includegraphics[width=4.2cm]{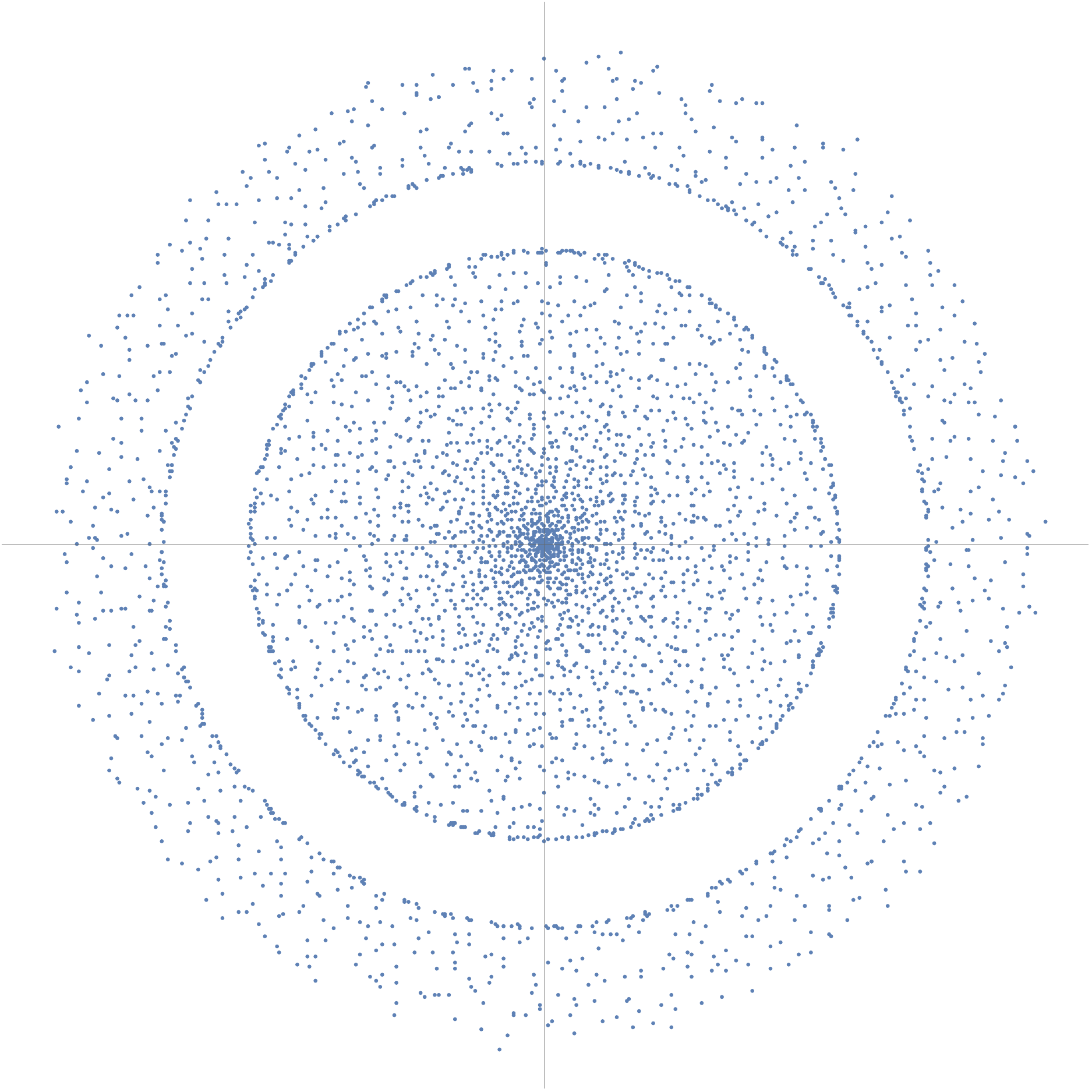}};
\node at (0,-2.5) {$b=\frac{1}{2}$};
\end{tikzpicture}
\begin{tikzpicture}[slave]
\node at (0,0) {\includegraphics[width=4.2cm]{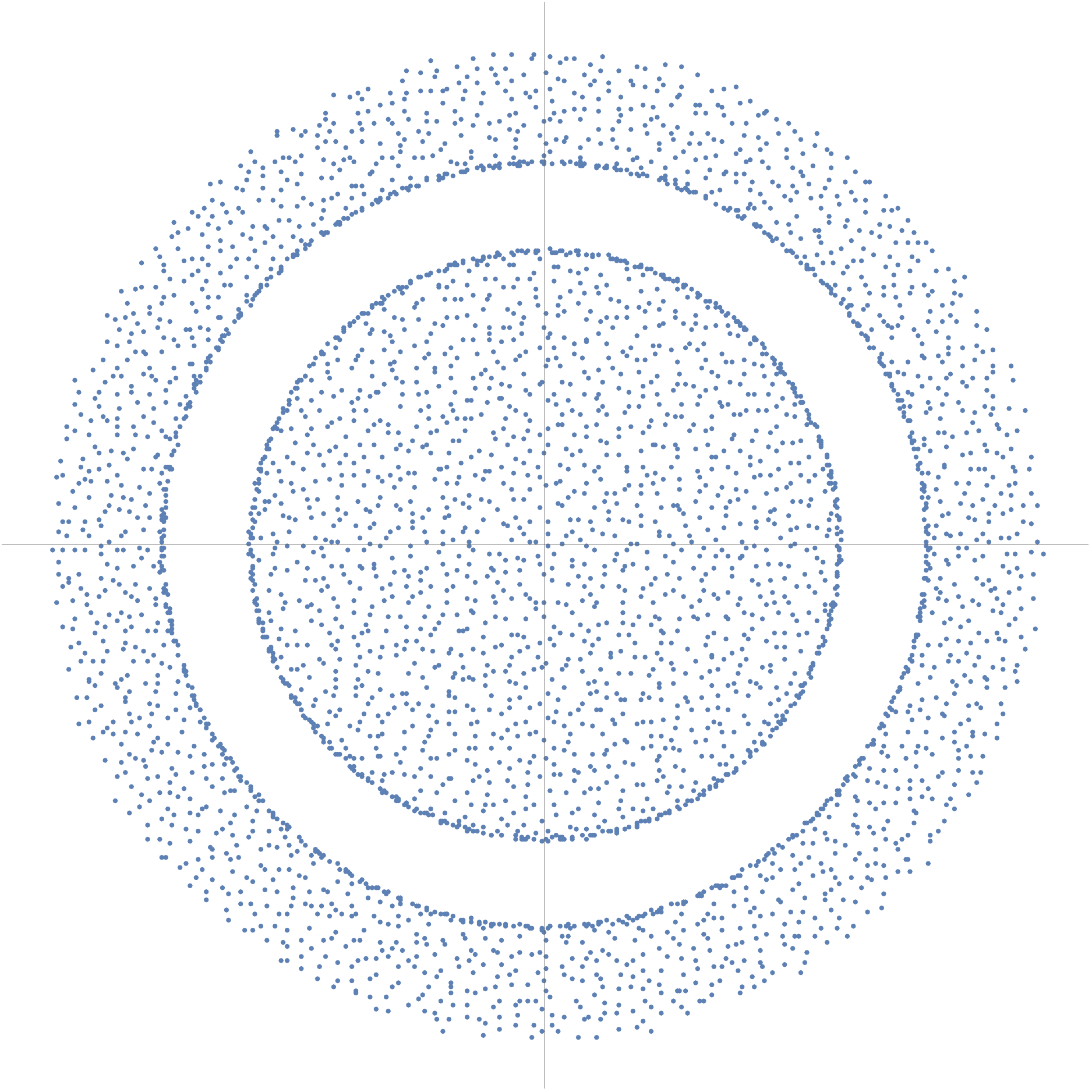}};
\node at (0,-2.5) {$b=1$};
\end{tikzpicture}
\begin{tikzpicture}[slave]
\node at (0,0) {\includegraphics[width=4.2cm]{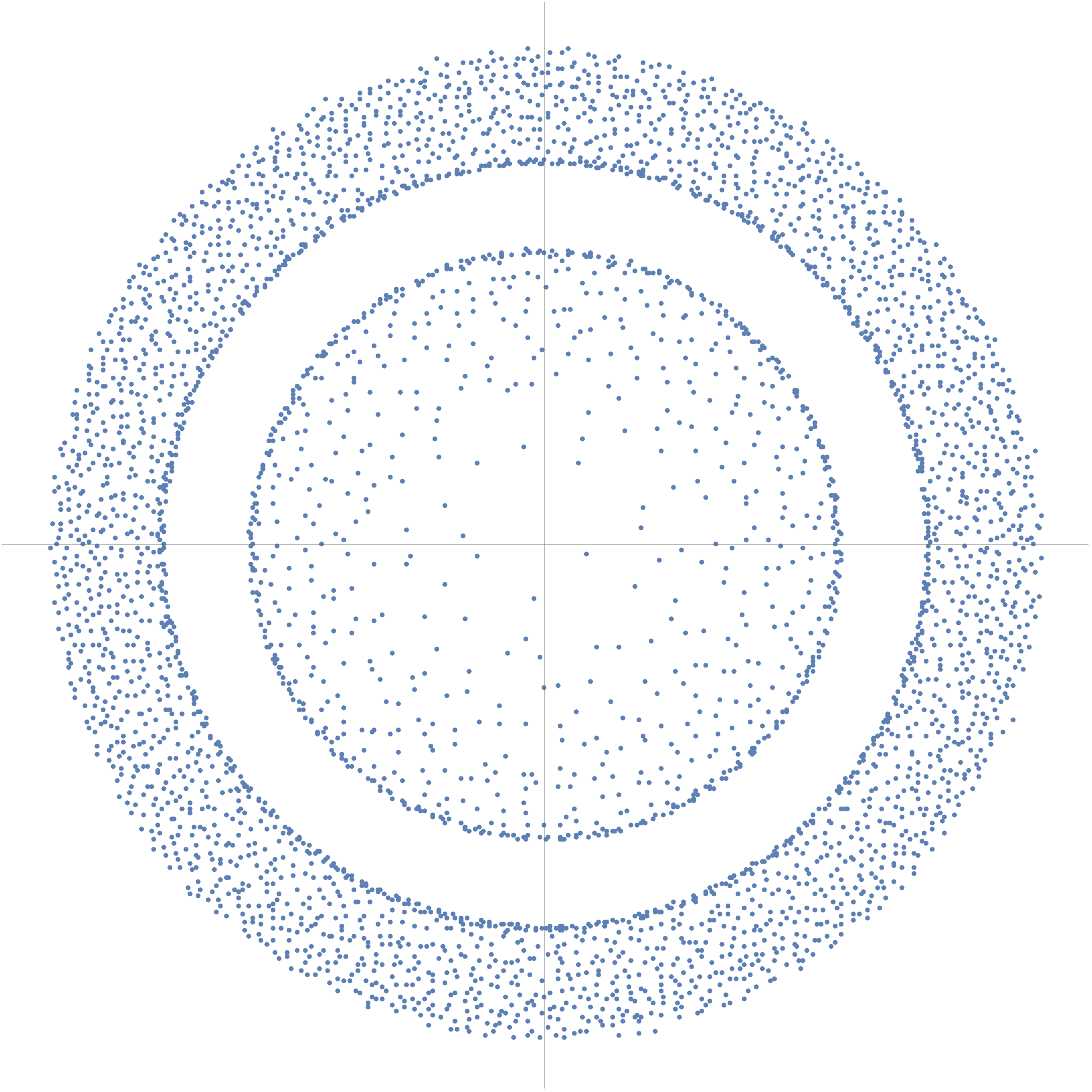}};
\node at (0,-2.5) {$b=2$};
\end{tikzpicture}
\end{center}
\caption{\label{fig: ML with hard wall} Illustration of the point process \eqref{def of point process} with $n=4096$, $r_{1}=\frac{3}{5}b^{-\frac{1}{2b}}$, $r_{2}=\frac{4}{5}b^{-\frac{1}{2b}}$, $\alpha=0$ and the indicated values of $b$.}
\end{figure}

The method of balayage in \cite{ACCL2022 2} shows that the equilibrium measure associated with \eqref{def of Q hard edge} is
\begin{align}\label{def of muh}
\mu(d^{2}z) = 2b^{2}r^{2b-1}\chi_{[0,r_{1}]\cup [r_{2},b^{-\frac{1}{2b}}]}(r) dr\frac{d\theta}{2\pi} + \sigma_{1} \delta_{r_{1}}(r) dr \frac{d\theta}{2\pi} + \sigma_{2} \delta_{r_{2}}(r) dr \frac{d\theta}{2\pi},
\end{align}
where $z=re^{i\theta}$, $r>0$, $\theta \in (-\pi,\pi]$ and
\begin{align}\label{def of taustar}
\sigma_{1} = \sigma_{\star}-br_{1}^{2b}, \qquad \sigma_{2} = br_{2}^{2b}-\sigma_{\star}, \qquad \sigma_{\star} := \frac{r_{2}^{2b}-r_{1}^{2b}}{2\log(\frac{r_{2}}{r_{1}})}.
\end{align}
The droplet is given by $S=\{z:|z|\in [0,r_{1}]\cup[r_{2},b^{-\frac{1}{2b}}]\}$.
We write $U=\{z:|z|>b^{-\frac{1}{2b}}\}$ and $G=\{z:r_1<|z|<r_2\}$ for the unbounded and bounded components of $\C\setminus S$ respectively, and refer to $G$ as the ``spectral gap'', or ``hard wall''. Note also that $G$ is not simply connected.

Thus $\partial U = \{|z| = b^{-\frac{1}{2b}}\}$ is a soft edge while $\partial G = \partial S \setminus \partial U = \{|z|=r_{1}\}\cup \{|z|=r_{2}\}$ is the union of two hard edges (see also Figure \ref{fig: ML with hard wall}). The quantity $\sigma_{1}>0$ represents the proportion of the points swept out from $G$ which accumulate near $\{|z|=r_{1}\}$, $\sigma_{2}>0$ is the proportion of the points accumulating near $\{|z|=r_{2}\}$, and $\sigma_{\star}=\mu(\{z:\, |z|\le r_1\})$.

\medskip Since $Q$ is rotation-invariant, the canonical correlation kernel of \eqref{def of point process} is given by
\begin{align}\label{def of Kn}
& K_{n}(z,w) = e^{-\frac{n}{2}Q(z)}e^{-\frac{n}{2}Q(w)} \sum_{j=1}^{n} \frac{z^{j-1}\overline{w}^{j-1}}{h_{j}},
\end{align}
where
\begin{align}
h_{j} & := \|z^{j-1}e^{-\frac{n}{2}Q(z)}\|_{L^{2}(\C,d^{2}z/\pi)}^{2} =  \int_{\mathbb{C}}|z|^{2j-2} e^{-nQ(z)} \frac{d^{2}z}{\pi} \nonumber \\
& = \frac{1}{b n^{\frac{j+\alpha}{b}}} \bigg( \gamma(\tfrac{j+\alpha}{b},nr_{1}^{2b}) - \gamma(\tfrac{j+\alpha}{b},nr_{2}^{2b}) + \Gamma(\tfrac{j+\alpha}{b}) \bigg), \qquad j=1,2,\ldots,n, \label{def of hj}
\end{align}
and the gamma functions $\gamma$, $\Gamma$ are defined by
\begin{align}\label{def of gamma and Gamma}
\Gamma(a):=\int_{0}^{\infty} t^{a-1}e^{-t}dt, \qquad \gamma(a,z) := \int_{0}^{z}t^{a-1}e^{-t}dt.
\end{align}
We recall some asymptotic properties of $\gamma$ in Appendix \ref{appendix:incomplete gamma}.

\medskip Our main results are asymptotic formulas for $K_{n}(z,w)$ as $n\to +\infty$ for five regimes where $z,w\in S$ are ``close" to $\partial G$. We will explore (i) the so-called ``hard edge regime", i.e. when $z,w$ are $\bigO(n^{-1})$ away from $\partial G$, and (ii) the ``semi-hard edge regime", i.e. when $z,w$ are at a distance of order $n^{-1/2}$ from $\partial G$. (If $z,w$ are slightly further away from $\partial G$, for example if $z,w$ are at a distance of order $\sqrt{\log n}/\sqrt{n}$ from $\partial G$, then the asymptotics of $K_{n}(z,w)$ are no longer affected by the hard wall and we recover the bulk regime \cite{ACCL2022 1}. This regime is not included here as it is standard by now.) The semi-hard edge regime was recently discovered in \cite{ACCL2022 1} in the study of a problem of counting statistics, but this regime has not yet been explored at the level of the correlation kernel. This regime is genuinely different from the bulk regime and from the hard edge regime.

\medskip For both the hard and semi-hard edge regimes, we study the asymptotics of $K_{n}(z,w)$ in the microscopic (i.e. when $z \approx w$) and the macroscopic (i.e. when $|z - w| \asymp 1$) cases. We find that the asymptotics of $K_{n}(z,w)$ oscillates in $n$ in the hard edge regime, but not in the semi-hard edge regime. In the microscopic (or ``diagonal'') case, the oscillations are described in terms of the Jacobi $\theta$ function, while in the macroscopic (or ``off-diagonal'') case they are described in terms of weighted Szeg\H{o} kernels. Our main results are stated in Section \ref{section:main results} and can be summarized as follows (see also Figure \ref{fig:summary}):
\begin{enumerate}
\item Theorem \ref{thm:r1 hard} establishes the asymptotics of $K_{n}(z,w)$, up to and including the fourth term of order $\sqrt{n}$, when $|z-w|=\bigO(\frac{1}{n})$ with $z$, $w$ being both at a distance $\bigO(\frac{1}{n})$ from $\partial G$,
\item Theorem \ref{thm:r1 semi-hard} establishes the asymptotics of $K_{n}(z,w)$, up to and including the second term of order $\sqrt{n}$, when $|z-w| =\bigO(\frac{1}{\sqrt{n}})$ with $z$, $w$ being both at a distance $\asymp \frac{1}{\sqrt{n}}$ from $\partial G$,
\item Theorem \ref{thm:r1-r2 case} establishes the leading order asymptotics of $K_{n}(z,w)$ when $z$, $w$ are on different sides of $G$ and are both at a distance $\bigO(\frac{1}{n})$ from $\partial G$,
\item Theorem \ref{thm:r1r1 hard} establishes the leading order asymptotics of $K_{n}(z,w)$ when $|z-w| \asymp 1$ with $z$, $w$ being on the same side of $G$ and being both at a distance $\bigO(\frac{1}{n})$ from $\partial G$,
\item Theorem \ref{thm:r1r1 semi-hard} establishes an upper bound on $K_{n}(z,w)$ when $|z-w| \asymp 1$ with $z$, $w$ being on the same side of $G$ and being both at a distance $\asymp \frac{1}{\sqrt{n}}$ from $\partial G$.
\end{enumerate}

\begin{figure}[h!]
\begin{center}
\begin{tikzpicture}
\node at (0,0) {\includegraphics[width=14cm]{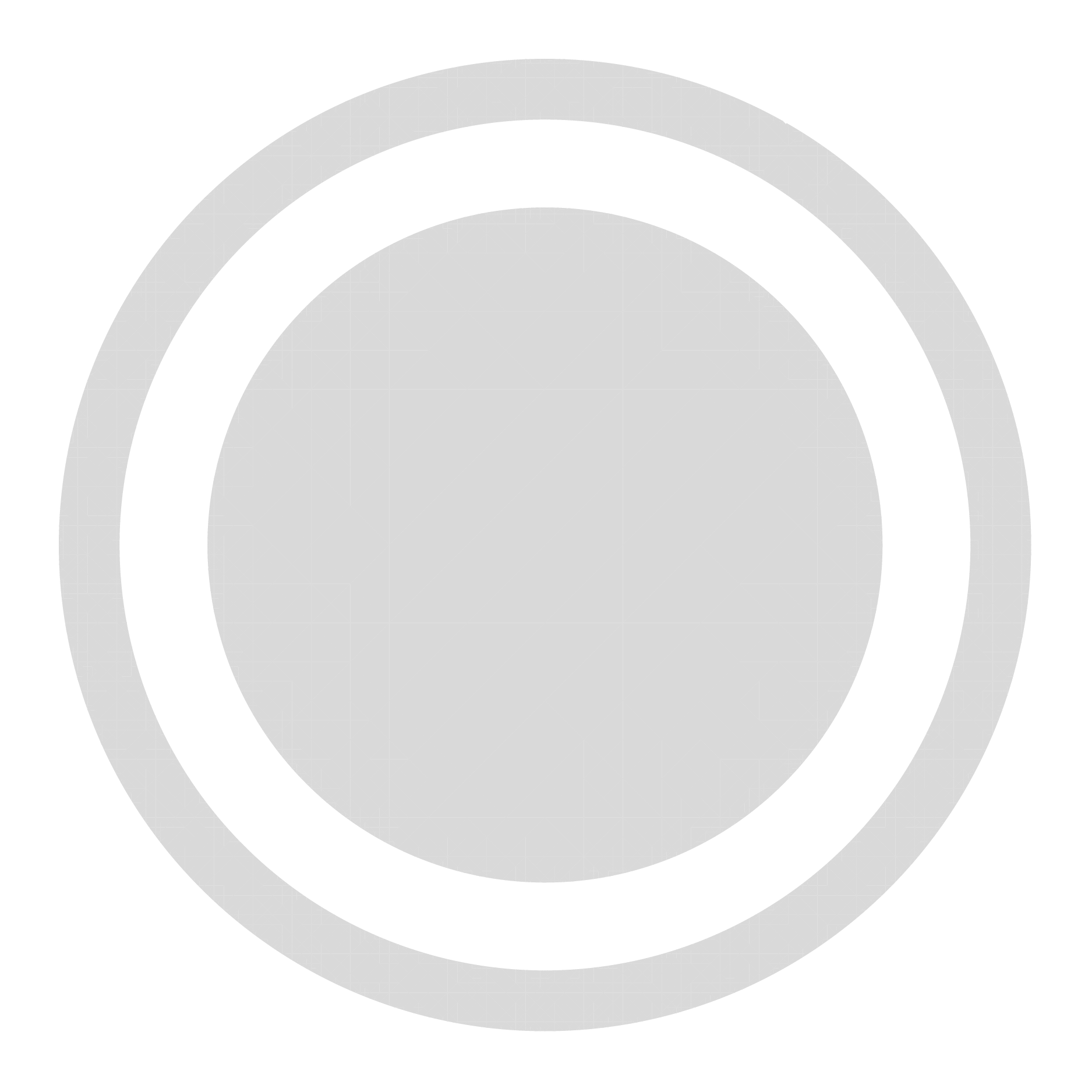}};
\node at (0,0) {\footnotesize $S$};
\node at (11:4.8) {\footnotesize $G$};
\node at (9:5.82) {\footnotesize $S$};
\node at ($(9:6.6)+(0,-0.13)$) {\footnotesize $U$};
\draw[fill] (6.22,0) circle (0.03cm);
\node at (6.35,0.27) {\footnotesize $b^{-\frac{1}{2b}}$};
\draw[fill] (5.44,0) circle (0.03cm);
\node at (5.44,0.2) {\footnotesize $r_{2}$};
\draw[fill] (4.32,0) circle (0.03cm);
\node at (4.32,0.2) {\footnotesize $r_{1}$};

% thm r1 hard
\draw[fill] (58:4.29) circle (0.03cm);
\node at (57:4.45) {\footnotesize $z$};
\draw[fill] (58:4.2) circle (0.03cm);
\node at (56.5:4.12) {\footnotesize $w$};
\draw (1.7,3.0)--(1.7,5.2)--(3.9,5.2)--(3.9,3.0)--(1.7,3.0);
\node at (2.3,5) {\tiny Thm \ref{thm:r1 hard}:};
\node at (2.78,4.8) {\tiny $r_{1}\hspace{-0.05cm}-\hspace{-0.05cm}|z|\hspace{-0.05cm}=\hspace{-0.05cm}\bigO(n^{-1})$};
\node at (2.81,4.55) {\tiny $r_{1}\hspace{-0.05cm}-\hspace{-0.05cm}|w|\hspace{-0.05cm}=\hspace{-0.05cm}\bigO(n^{-1})$};
\node at (2.76,4.3) {\tiny $|z\hspace{-0.05cm}-\hspace{-0.05cm}w|\hspace{-0.05cm}=\hspace{-0.05cm}\bigO(n^{-1})$};

\draw (1.7,5.2)--(2,6);
\node at (3.2,6.1) {\tiny Oscillatory Jacobi $\theta$ function};

% thm r1 semi-hard
\draw[fill] (118:4.05) circle (0.03cm);
\node at (117:4.2) {\footnotesize $z$};
\draw[fill] (118:3.8) circle (0.03cm);
\node at (115:3.72) {\footnotesize $w$};
\draw (-2.6,3.0)--(-2.6,5.2)--(-0.4,5.2)--(-0.4,3.0)--(-2.6,3.0);
\node at (-2,5) {\tiny Thm \ref{thm:r1 semi-hard}:};
\node at (-1.61,4.8) {\tiny $r_{1}\hspace{-0.05cm}-\hspace{-0.05cm}|z|\hspace{-0.05cm}\asymp \hspace{-0.05cm} n^{-1/2}$};
\node at (-1.58,4.55) {\tiny $r_{1}\hspace{-0.05cm}-\hspace{-0.05cm}|w|\hspace{-0.05cm}\asymp\hspace{-0.05cm}n^{-1/2}$};
\node at (-1.52,4.3) {\tiny $|z\hspace{-0.05cm}-\hspace{-0.05cm}w|\hspace{-0.05cm}=\hspace{-0.05cm}\bigO(\hspace{-0.02cm}n^{\hspace{-0.02cm}-\hspace{-0.02cm}1\hspace{-0.02cm}/\hspace{-0.02cm}2})$};

\draw (-2.6,5.2)--(-2.8,6);
\node at (-2.6,6.1) {\tiny No oscillations};

% thm r1r2
\draw[fill] (176:4.25) circle (0.03cm);
\node at (174:4.2) {\footnotesize $z$};
\draw[fill] (186:5.54) circle (0.03cm);
\node at (186:5.72) {\footnotesize $w$};
\draw (-6.1,-0.85)--(-6.1,1.35)--(-3.9,1.35)--(-3.9,-0.85)--(-6.1,-0.85);
\node at (-5.5,1.15) {\tiny Thm \ref{thm:r1-r2 case}:};
\node at (-5.03,0.95) {\tiny $r_{1}\hspace{-0.05cm}-\hspace{-0.05cm}|z|\hspace{-0.05cm}= \hspace{-0.05cm} \bigO(n^{-1})$};
\node at (-5.03,0.7) {\tiny $|w|\hspace{-0.05cm}-\hspace{-0.05cm}r_{2}\hspace{-0.05cm}=\hspace{-0.05cm}\bigO(n^{-1})$};

\draw (-3.9,1.35)--(-3.4,1.6);
\node at (-2.4,1.7) {\tiny Oscillatory Szeg\H{o} kernel};

% thm r1r1 hard
\draw[fill] (232:4.25) circle (0.03cm);
\node at (230:4.2) {\footnotesize $z$};
\draw[fill] (253:4.25) circle (0.03cm);
\node at (253:4.1) {\footnotesize $w$};
\draw (-3.1,-5.25)--(-3.1,-3)--(-0.9,-3)--(-0.9,-5.25)--(-3.1,-5.25);
\node at (-2.5,-4.4) {\tiny Thm \ref{thm:r1r1 hard}:};
\node at (-2.03,-4.6) {\tiny $r_{1}\hspace{-0.05cm}-\hspace{-0.05cm}|z|\hspace{-0.05cm}= \hspace{-0.05cm} \bigO(n^{-1})$};
\node at (-2.01,-4.85) {\tiny $r_{1}\hspace{-0.05cm}-\hspace{-0.05cm}|w|\hspace{-0.05cm}=\hspace{-0.05cm}\bigO(n^{-1})$};
\node at (-2.42,-5.1) {\tiny $|z\hspace{-0.05cm}-\hspace{-0.05cm}w|\hspace{-0.05cm}\asymp \hspace{-0.05cm} 1$};

\draw (-3.1,-3)--(-2.8,-2.5);
\node at (-1.68,-2.4) {\tiny Oscillatory regularized Szeg\H{o} kernel};

% thm r1r1 semi-hard
\draw[fill] (288:3.95) circle (0.03cm);
\node at (288:3.8) {\footnotesize $z$};
\draw[fill] (306:3.95) circle (0.03cm);
\node at (304:3.85) {\footnotesize $w$};
\draw (0.9,-5.25)--(0.9,-3)--(3.1,-3)--(3.1,-5.25)--(0.9,-5.25);
\node at (1.5,-4.4) {\tiny Thm \ref{thm:r1r1 semi-hard}:};
\node at (1.9,-4.6) {\tiny $r_{1}\hspace{-0.05cm}-\hspace{-0.05cm}|z|\hspace{-0.05cm}\asymp \hspace{-0.05cm} n^{-1/2}$};
\node at (1.935,-4.85) {\tiny $r_{1}\hspace{-0.05cm}-\hspace{-0.05cm}|w|\hspace{-0.05cm}\asymp\hspace{-0.05cm}n^{-1/2}$};
\node at (1.55,-5.1) {\tiny $|z\hspace{-0.05cm}-\hspace{-0.05cm}w|\hspace{-0.05cm}\asymp \hspace{-0.05cm} 1$};

\draw (0.9,-3)--(1.2,-2.5);
\node at (2,-2.4) {\tiny No oscillations};
\node at (2.2,-2.65) {\tiny Small correlations};
\end{tikzpicture}
\end{center}
\caption{Summary of our main results on the asymptotics of $K_{n}(z,w)$.}
\label{fig:summary}
\end{figure}

\section{Main results}\label{section:main results}
Our first main result is an asymptotic formula for $K_{n}(z,w)$ in the microscopic regime $|z-w|=\bigO(n^{-1})$ when both $z$ and $w$ are in the hard edge regime. This asymptotic formula contains oscillations that are described in terms of the Jacobi $\theta$ function. We recall that this function is defined by
\begin{align}\label{def of Jacobi theta}
\theta(z;\tau) := \sum_{\ell=-\infty}^{+\infty} e^{2\pi i \ell z}e^{\pi i \ell^{2}\tau}, \qquad z \in \mathbb{C}, \qquad \tau \in i(0,+\infty),
\end{align}
and satisfies $\theta(-z;\tau)=\theta(z;\tau)$ and $\theta(z+1;\tau)=\theta(z;\tau)$. Other properties of $\theta$ can be found in e.g. \cite[Chapter 20]{NIST}. The statement of Theorem \ref{thm:r1 hard} also involves Euler's gamma constant $\gamma_{\mathrm{E}}\approx 0.5772$, as well as the exponential integral $E_{1}$ and the complementary error function $\mathrm{erfc}$ (see e.g. \cite[eqs 6.2.1 and 7.2.2]{NIST}), which we recall are defined by
\begin{align}\label{def of erfc}
E_{1}(x) := \int_{x}^{+\infty} \frac{e^{-t}}{t}dt, \qquad \mathrm{erfc} (z) = \frac{2}{\sqrt{\pi}}\int_{z}^{\infty} e^{-t^{2}}dt.
\end{align}
\begin{theorem}\label{thm:r1 hard}\emph{(``$r_{1}$ hard edge case")}
Let $t_{1},t_{2}\geq 0$ and $\beta \in \R$ be fixed, and define
\begin{align}\label{def of z1z2 r1 r1 case intro}
z = r_{1}\Big(1-\frac{t_{1}}{\sigma_{1}n}\Big)e^{i\beta}, \qquad w = r_{1}\Big(1-\frac{t_{2}}{\sigma_{1}n}\Big)e^{i\beta}.
\end{align}
As $n \to + \infty$,
\begin{align}\label{kernel asymp r1 hard}
& K_{n}(z,w) = C_{1} n^{2} + C_{2} \, n \log n + \Big(C_{3}+\frac{\sigma_{1}}{r_{1}^{2}}e^{-t_{1}-t_{2}}\mathcal{F}_{n}\Big) n + C_{4} \sqrt{n} + \bigO(n^{\frac{2}{5}}),
\end{align}
where $\mathcal{F}_{n}$ is given by
\begin{align}
& \mathcal{F}_{n} = \frac{(\log \theta)'\Big(n \sigma_{\star}+\frac{1}{2}-\alpha+\frac{\log (\sigma_{2}/\sigma_{1})}{2\log (r_{2}/r_{1})}; \frac{\pi i}{\log (r_{2}/r_{1})}\Big) + \log (\sigma_{2}/\sigma_{1})}{2\log (r_{2}/r_{1})}, \label{Fn theta 2 intro}
\end{align}
and $\sigma_{1},\sigma_{2}$ and $\sigma_{\star}$ are defined in \eqref{def of taustar}.
If $(t_{1},t_{2}) \neq (0,0)$, the constants $C_{1},\ldots,C_{4}$ (constants with respect to $n$) are given by
\begin{align}
& C_{1} = \sigma_{1}^{2}\frac{1-e^{-t_{1}-t_{2}}( 1+t_{1}+t_{2}) }{r_{1}^{2}(t_{1}+t_{2})^{2}}, \nonumber \\
& C_{2} = \frac{b^{2}r_{1}^{2b-2}}{2}, \nonumber \\
& C_{3} = - b^{2} r_{1}^{2b-2}\bigg( E_{1}( t_{1}+t_{2} ) + \gamma_{\mathrm{E}} + \log\bigg(\frac{br_{1}^{2b}(t_{1}+t_{2})}{\sigma_{1}\sqrt{2\pi}}\bigg) \bigg) \nonumber \\
& + \frac{\sigma_{1}}{ r_{1}^{2}(t_{1}+t_{2})^{3}} \bigg\{ \frac{(t_{1}+t_{2})^{2}(t_{1}^{2}+t_{2}^{2})}{2e^{t_{1}+t_{2}}}  + \Big( 2t_{1}t_{2} - \frac{b^{2}r_{1}^{2b}}{\sigma_{1}}(t_{1}+t_{2})(t_{1}^{2}+t_{2}^{2})\Big) \bigg( 1 - \frac{1+t_{1}+t_{2}}{e^{t_{1}+t_{2}}} \bigg) \bigg\}, \nonumber \\
& C_{4} = \sqrt{2} b^{2}r_{1}^{b-2}\big( 1-2br_{1}^{2b}\frac{t_{1}+t_{2}}{\sigma_{1}} \big)\mathcal{I}, \label{def of C1C2C3C4 hard edge}
\end{align}
and $\mathcal{I}\in \R$ is given by
\begin{align}
& \mathcal{I} = \int_{-\infty}^{+\infty} \bigg\{ \frac{y\, e^{-y^{2}}}{\sqrt{\pi}\, \mathrm{erfc}(y)} - \chi_{(0,+\infty)}(y) \bigg[ y^{2}+\frac{1}{2} \bigg] \bigg\}dy \approx -0.81367. \label{def of I}
\end{align}
If $t_{1}=t_{2}=0$, then one simply needs to let $t_{1}+t_{2}\to 0$ in the above formulas for $C_{1},\ldots,C_{4}$; more precisely, for $t_{1}=t_{2}=0$ the constants $C_{1},\ldots,C_{4}$ are given by
\begin{align*}
& C_{1} = \frac{\sigma_{1}^{2}}{2 r_{1}^{2}}, \qquad C_{2} = \frac{b^{2}r_{1}^{2b-2}}{2}, \qquad C_{3} =  \frac{b^{2}r_{1}^{2b-2}}{2}\log \Big( \frac{2\pi \sigma_{1}^{2}}{b^{2}r_{1}^{2b}} \Big), \qquad C_{4} = \sqrt{2} b^{2}r_{1}^{b-2} \mathcal{I}.
\end{align*}
\end{theorem}
\begin{remark}
If either $t_{1}<0$ or $t_{2}<0$, then either $z\in G$ or $w\in G$ and $K_{n}(z,w)$ is trivially $0$ by \eqref{def of Q hard edge} and \eqref{def of Kn}; this is why we restricted ourselves to $t_{1},t_{2}\geq 0$ in the statement of Theorem \ref{thm:r1 hard}. Taking $t_{1}=t_{2}$ in Theorem \ref{thm:r1 hard} yields precise asymptotics for the one-point function $R_{n,1}\big(r_{1}e^{i\beta}(1-\frac{t_{1}}{\sigma_{1}n})\big)$ in the hard edge regime. For $t_{1}=t_{2}$, $C_{1}$ in Theorem \ref{thm:r1 hard} also appeared in \cite[eq (21) with $L=1$]{ZS2000} in the context of truncated unitary matrices, and for general $t_{1},t_{2}$ the quantity $C_{1}$ also appeared when considering perturbations of rank $L$ to Hermitian
matrices, see e.g. \cite[eq (2.64)]{BF2022}, \cite[Section 2.2 and 3.1]{FS1997} and the references therein. (For truncated unitary matrices, the situation is somewhat simpler because there is no bulk and the droplet is connected. Therefore, the asymptotics for the correlation kernel are expected to be of a simpler form than \eqref{kernel asymp r1 hard}, i.e. without oscillations, and without terms of order $n\log n$ -- see also the discussion \cite[end of Section 1]{ACM2024}.)

For comparison, in the soft edge setting, the asymptotics of the one-point function near a spectral gap take the following form: for $z=e^{i\beta}(r_{1} + \frac{t}{\sqrt{n\Delta Q(r_{1})}})$, we have
\begin{align*}
R_{n,1}(z)=\tilde{C}_{1}n + \sqrt{n}(\tilde{C}_{2} + e^{-2t^{2}}\tilde{\mathcal{F}}_{n}) + \bigO((\log n)^{4}), \qquad \mbox{as } n \to + \infty,
\end{align*}
where $\tilde{C}_{1},\tilde{C}_{2}$ are independent of $n$ and $\tilde{\mathcal{F}}_{n}$ is independent of $t$ and of order $1$, see \cite[Theorem 1.9]{ACC2022}.
\end{remark}
\begin{figure}
\begin{center}
\begin{tikzpicture}[master]
\node at (0,0) {\includegraphics[width=7cm]{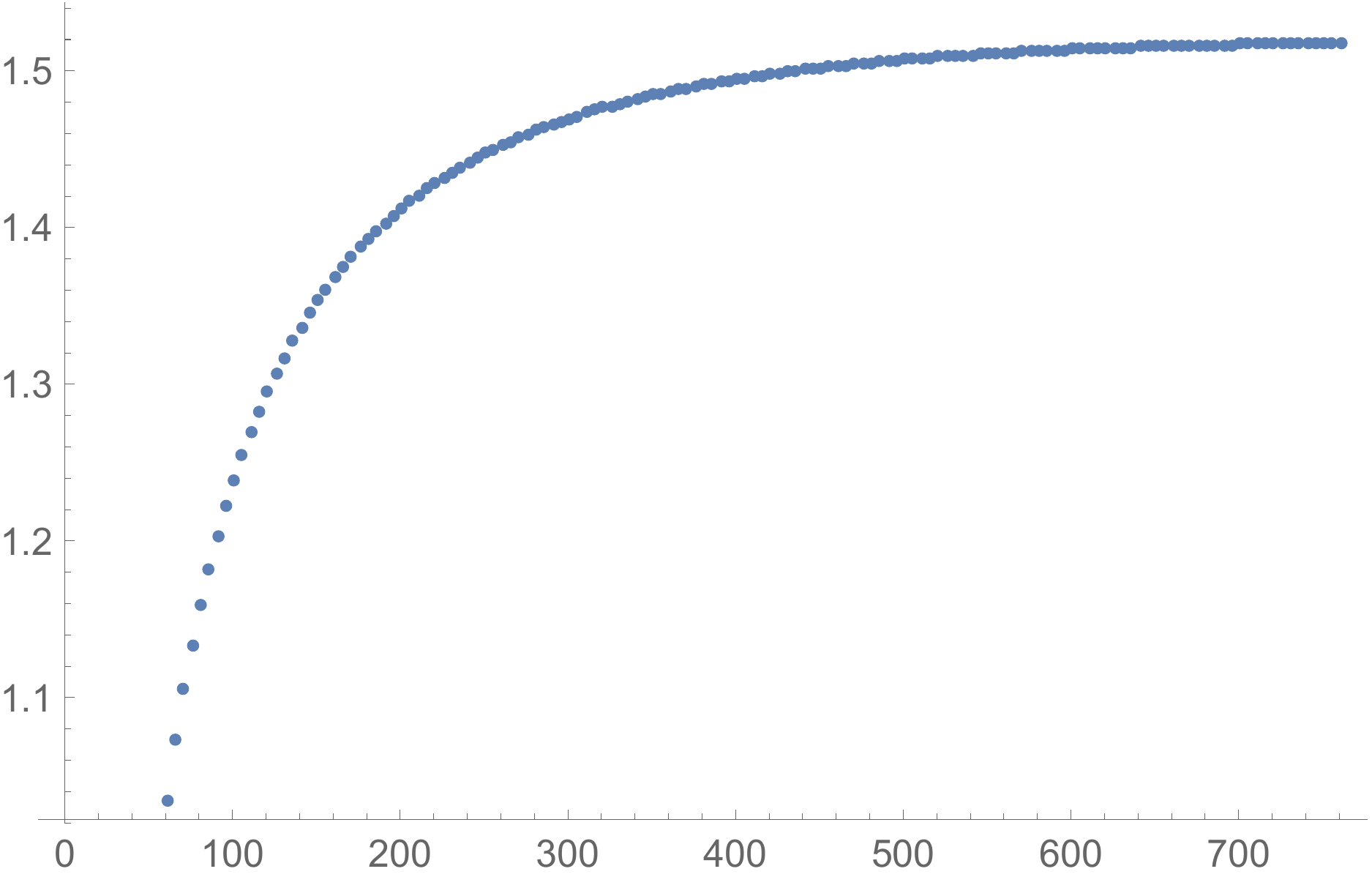}};
%\node at (0,0) {$K_{n}(z,w) - \big(C_{1} n^{2} + C_{2} \, n \log n + \Big(C_{3}+\frac{\sigma_{1}}{r_{1}^{2}}e^{-t_{1}-t_{2}}\mathcal{F}_{n}\Big) n + C_{4} \sqrt{n}\big)$};
\end{tikzpicture}
\begin{tikzpicture}[slave]
\node at (0,0) {\includegraphics[width=7cm]{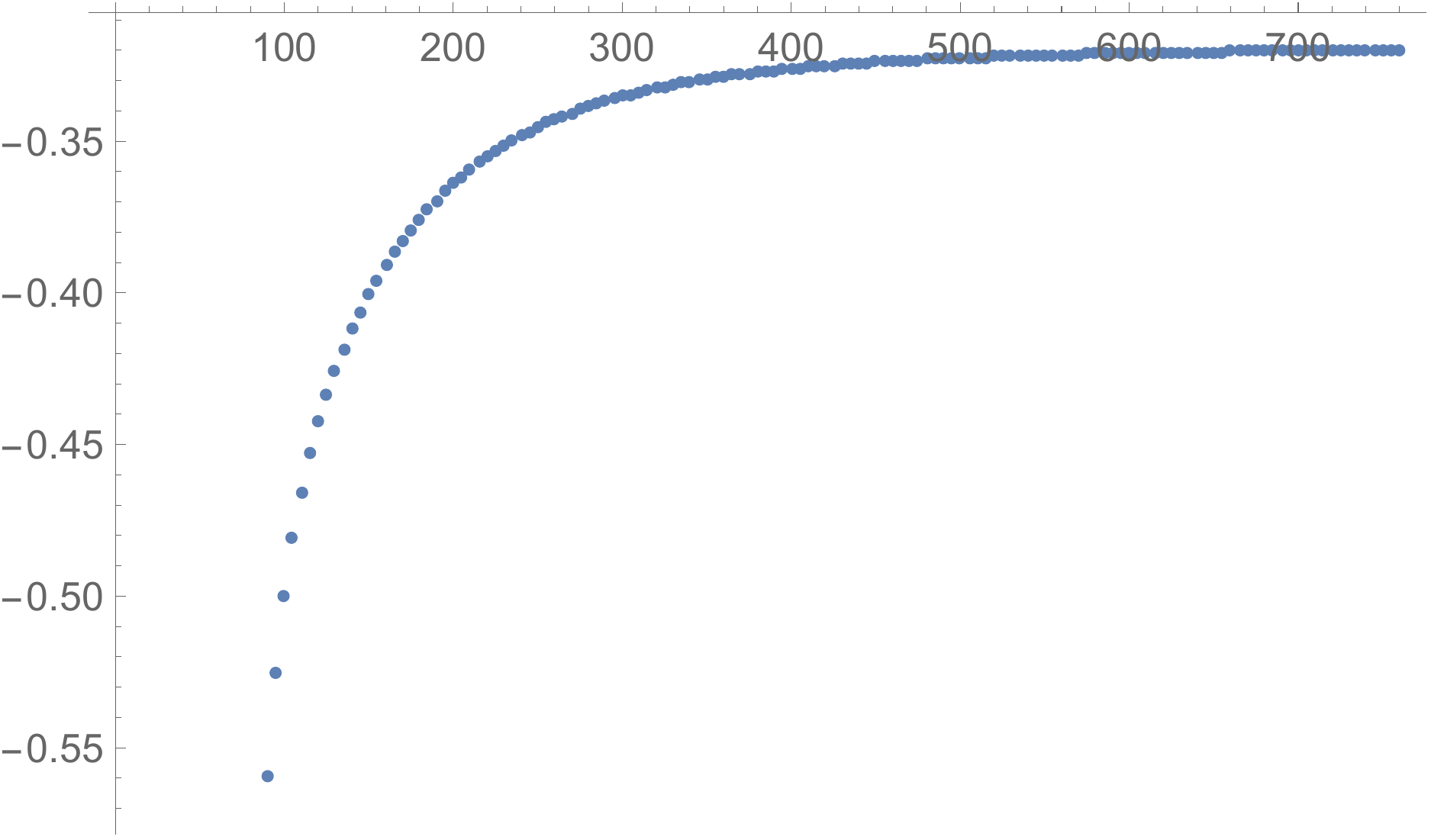}};
%\node at (0,0) {$K_{n}(z,w) - \big(C_{1} n^{2} + C_{2} \, n \log n + \Big(C_{3}+\frac{\sigma_{1}}{r_{1}^{2}}e^{-t_{1}-t_{2}}\mathcal{F}_{n}\Big) n + C_{4} \sqrt{n}\big)$};
\end{tikzpicture}
\end{center}
\caption{Numerical confirmations of Theorems \ref{thm:r1 hard} (left) and \ref{thm:r1 semi-hard} (right). For both pictures, $b=1.3$, $\alpha=1.26$, $r_{1}=0.42\smash{b^{-\frac{1}{2b}}}$, $r_{2}=0.67\smash{b^{-\frac{1}{2b}}}$. Left: the function $n\mapsto \frac{1}{\log n}\big[K_{n}(z,w) - \big(C_{1} n^{2} + C_{2} \, n \log n + \big(C_{3}+\frac{\sigma_{1}}{r_{1}^{2}}e^{-t_{1}-t_{2}}\mathcal{F}_{n}\big) n + C_{4} \sqrt{n}\big)\big]$, with $z = r_{1}(1-\frac{t_{1}}{\sigma_{1}n})$, $w = r_{1}(1-\frac{t_{2}}{\sigma_{1}n})$, $t_{1}=0.21$, $t_{2}=0.45$, and $C_{1},C_{2},C_{3},C_{4}$ as in \eqref{def of C1C2C3C4 hard edge}. This function seems to approach a constant as $n\to+\infty$, which is consistent with Theorem \ref{thm:r1 hard}, and also suggests that the term $\bigO(n^{\frac{2}{5}})$ in \eqref{kernel asymp r1 hard} is actually $\bigO(\log n)$. Right: the function $n\mapsto K_{n}(z,w) - \big(C_{1} n + C_{2} \sqrt{n}\big)$, with $z = r_{1}\big(1-\frac{\mathfrak{s}_{1}}{br_{1}^{b}\sqrt{2n}}\big)$, $w = r_{1}\big(1-\frac{\mathfrak{s}_{2}}{br_{1}^{b}\sqrt{2n}}\big)$, $\mathfrak{s}_{1}=1.21$, $\mathfrak{s}_{2}=1.45$, and $C_{1},C_{2}$ as in \eqref{def of C1C2 semi hard}. This function seems to approach a constant as $n\to+\infty$, which is consistent with Theorem \ref{thm:r1 semi-hard}.
}
\label{fig:numerical check}
\end{figure}
\begin{remark}
In this paper we focus on the case where the hard wall $G$ lies entirely in the interior of $\{z:|z|\leq b^{-\frac{1}{2b}}\}$. We do not cover the simpler case where the hard wall is of the form $\tilde{G}=\{z:|z| > r_{1}\}$; in this case we expect that a similar formula as \eqref{kernel asymp r1 hard} holds but with $\mathcal{F}_{n}=0$ (there should be no oscillations since the droplet is then given by $\{z:|z|\leq r_{1}\}$, which is a connected set).

This expectation is supported by the works \cite{Seo2021}. Consider a radially symmetric potential $Q$ whose associated equilibrium measure is $\mu(d^{2}z)=2\Delta Q(z) \, \chi_{[\rho_{1},r_{1}]}(r) r\, dr \frac{d\theta}{2\pi} + \sigma_{1} \delta_{r_{1}}(r)dr \frac{d\theta}{2\pi} $ ($z=re^{i\theta}$), supported on a connected set of the form $\tilde{S}=\{z:|z|\in [\rho_{1},r_{1}] \}$ with $r_{1}>\rho_{1}\geq 0$. Let $p\in \partial \tilde{U} = \{z:|z|=r_{1}\}$. Under general conditions on $Q$, it is proved in \cite[Theorem 2.1]{Seo2021} that, as $n\to + \infty$,
\begin{align*}
K_{n}\Big( r_{1}e^{i\beta} \Big[1 - \frac{t_{1}}{n \sigma_{1}}\Big], r_{1}e^{i\beta}\Big[1 - \frac{t_{2}}{n \sigma_{1}} \Big] \Big) & = \frac{n^{2}\sigma_{1}^{2}}{r_{1}^{2}} \int_{0}^{1} t e^{-x(t_{1}+t_{2})}dx + o(n^{2}) = C_{1}n^{2} + o(n^{2}).
\end{align*}
Comparing the above with \eqref{kernel asymp r1 hard}, we conclude that the leading order behavior of $K_{n}\big( r_{1}e^{i\beta} (1 - \frac{t_{1}}{\sigma_{1}n})$, $r_{1}e^{i\beta} (1 - \frac{t_{2}}{\sigma_{1}n} ) \big)$ in the two situations $G=\{z:|z|\in (r_{1},r_{2})\}$ (considered here) and $\tilde{G}=\{z:|z|\in (r_{1},+\infty)\}$ (covered in \cite{Seo2021}) are identical.
\end{remark}
\begin{remark}
For two dimensional point processes in the multi-component regime (i.e. when $S$ consists of several disjoint components), the Jacobi $\theta$ function is known to also describe large gap fluctuations \cite{CharlierAnnuli}, smooth linear statistics and microscopic correlations for ensembles with soft edges \cite{ACC2022}, and disk counting statistics for ensembles with hard edges \cite{ACCL2022 2}. It is still an open problem to establish the emergence of the theta function for two-dimensional point processes that are not rotation-invariant; such ensembles include the so-called lemniscate ensemble (see e.g. \cite{BY2022}) and the elliptic Ginibre point process with a large point charge \cite{Byun}.
\end{remark}
Our next theorem on the asymptotics of $K_{n}(z,w)$ concerns the semi-hard edge regime in the microscopic regime $|z-w|=\bigO(n^{-1/2})$. Let $\tilde{\Delta} Q(r_{1}) := \lim_{r\to r_{1}, r<r_{1}}\Delta Q(r) = b^{2}r_{1}^{2b-2}$.
\begin{theorem}\label{thm:r1 semi-hard}\emph{(``$r_{1}$ semi-hard edge case")}
Let $\mathfrak{s}_{1}, \mathfrak{s}_{2} > 0$, $\beta \in \mathbb{R}$ be fixed parameters, and define
\begin{align}\label{def of z1z2 r1 r1 case o thm intro}
& z = r_{1}\Big(1-\frac{\mathfrak{s}_{1}}{br_{1}^{b}\sqrt{2n}}\Big)e^{i\beta}, & & w = r_{1}\Big(1-\frac{\mathfrak{s}_{2}}{br_{1}^{b}\sqrt{2n}}\Big)e^{i\beta} = \bigg(r_{1} - \frac{\mathfrak{s}_{2}}{(2 \tilde{\Delta} Q(r_{1}))^{\frac{1}{2}}\sqrt{ n}}\bigg)e^{i\beta}.
\end{align}
As $n \to + \infty$,
\begin{align}\label{asymp semi hard}
& K_{n}(z,w) = C_{1} n + C_{2} \sqrt{n} + \bigO(1),
\end{align}
where
\begin{align}
& C_{1} = 2\tilde{\Delta} Q(r_{1}) \int_{-\infty}^{+\infty} \frac{e^{-\frac{(y+\mathfrak{s}_{1})^{2} + (y+\mathfrak{s}_{2})^{2}}{2}}}{\sqrt{\pi} \mathrm{erfc}(y)}dy , \nonumber \\
& C_{2} = b\frac{(2 \tilde{\Delta} Q(r_{1}))^{\frac{1}{2}} }{r_{1}}\int_{-\infty}^{+\infty} \frac{e^{-\frac{(y+\mathfrak{s}_{1})^{2} + (y+\mathfrak{s}_{2})^{2}}{2}}}{\sqrt{\pi} \mathrm{erfc}(y)} \bigg\{ \frac{(10y^{2}-2)e^{-y^{2}}}{3\sqrt{\pi}\mathrm{erfc}(y)} + 5y - \frac{10y^{3}}{3} + \frac{\mathfrak{s}_{1}+\mathfrak{s}_{2}}{b} - y \frac{\mathfrak{s}_{1}^{2}+\mathfrak{s}_{2}^{2}}{2b} \nonumber \\
& \hspace{6.5cm} - 2y^{2}(\mathfrak{s}_{1}+\mathfrak{s}_{2}) + \frac{2b-3}{b} \frac{\mathfrak{s}_{1}^{3}+\mathfrak{s}_{2}^{3}}{6} \bigg\}dy. \label{def of C1C2 semi hard}
\end{align}
\end{theorem}
\begin{figure}
\begin{center}
\begin{tikzpicture}[master]
\node at (0,0) {\includegraphics[width=8cm]{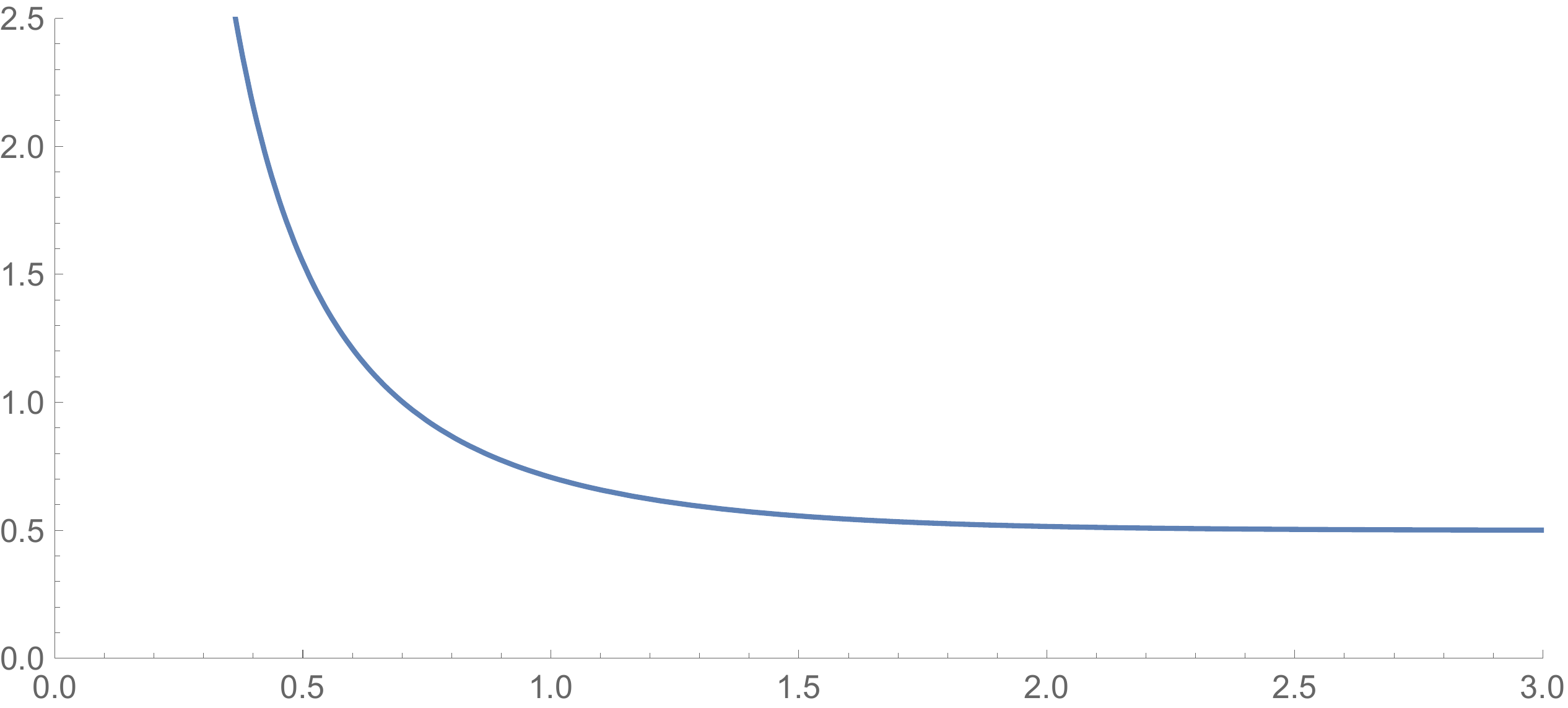}};
\end{tikzpicture}
\end{center}
\caption{The semi-hard edge density profile $\rho(x)$.
}
\label{fig:C1 semi hard}
\end{figure}

\begin{remark}
In the statement of Theorem \ref{thm:r1 semi-hard}, note that $\mathfrak{s}_{1},\mathfrak{s}_{2}>0$. The case $\mathfrak{s}_{1}=\mathfrak{s}_{2}=0$ is very different and is covered as a special case of Theorem \ref{thm:r1 hard}.
\end{remark}
\begin{remark}\label{ss}
The constant $C_{1}$ in \eqref{asymp semi hard} is universal for general potentials in the semi-hard regime: if things are defined appropriately, the associated ``microscopic semi-hard correlation kernel''
$$\tilde{K}_n(\mathfrak{s}_{1},\mathfrak{s}_{2}):=\frac{1}{2n\tilde{\Delta} Q(r_1)} K_n\Big(r_1-\tfrac{\mathfrak{s}_{1}}{\sqrt{2n\tilde{\Delta} Q(r_1)}},r_1-\tfrac{\mathfrak{s}_{2}}{\sqrt{2n\tilde{\Delta} Q(r_1)}}\Big)$$
satisfies
\begin{align}\label{zipp}
 \tilde{K}_n(\mathfrak{s}_{1},\mathfrak{s}_{2})= \int_{-\infty}^{+\infty} \frac{e^{-\frac{(y+\mathfrak{s}_{1})^{2} + (y+\mathfrak{s}_{2})^{2}}{2}}}{\sqrt{\pi} \mathrm{erfc}(y)}dy + o(1), \quad \mbox{as } n \to + \infty.
\end{align}
For convenience, precise conditions implying the universality \eqref{zipp} are discussed in Appendix \ref{Appendix semi-hard universal}.

Passing to the limit as $n\to\infty$ in \eqref{zipp} with $\mathfrak{s}_{2} = \mathfrak{s}_{1}=x$, we obtain the microscopic semi-hard edge one-point density profile
\begin{equation}\label{1dp}\rho(x):=\int_{-\infty}^{+\infty} \frac{e^{-(y+x)^{2}}}{\sqrt{\pi} \mathrm{erfc}(y)}\,dy,\qquad (x>0),\end{equation}
which is believed to be universal at any kinds hard walls in the semi-hard regime, see Appendix \ref{Appendix semi-hard universal}. The density profile $\rho(x)$ is depicted in Figure \ref{fig:C1 semi hard}.
\end{remark}

\begin{remark}
In the paper \cite{RB}, the authors consider two independent plasmas whose droplets are annuli which are adjacent to each other, and where the common boundary between the plasmas is an impermeable membrane, i.e. it is a soft/hard edge for each of the plasmas.

Interestingly, a density profile for the $\bigO(1/\sqrt{n})$-interface about the edge is computed in \cite[Eq. (3.4)]{RB}.
The profile bears a formal resemblance to \eqref{1dp}, and also to the soft/hard edge plasma function studied in \cite{J,AKM,AKMW}.
%A full investigation of the details, and how this setting compares to other types of walls, seems to be an appropriate topic for a follow-up work.
 \end{remark}

Our next three theorems on the asymptotics of $K_{n}(z,w)$ treat some macroscopic regimes when $|z-w|\asymp 1$.

\medskip For comparison purposes, we first briefly recall the results from \cite{AC, ACC2022} about asymptotics for $K_{n}(z,w)$ in some macroscopic regimes for ensembles with soft edges. When $z,w$ are far from each other but both close to $\partial U$, the asymptotics of $K_{n}(z,w)$ involves the (unweighted) Szeg\H{o} kernel $S^U_{\mathrm{soft}}(z,w)$, which is the reproducing kernel for the Hardy space $H^2_0(U)$ of all holomorphic functions on $U$ vanishing at infinity, supplied with the norm $\|f\|_{H^2_{0}(U)}^2 := \int_{\partial U}|f|^2\,|dz|$, see \cite{AC}. On the other hand, when $z,w$ are far from each other and both close to a (soft) spectral gap of the form $G=\{z:|z|\in (r_{1},r_{2})\}$, there is no convergence towards a limiting kernel \cite{ACC2022}: for example, if $z$ and $w$ are close to $\partial G$ but lie on different regions of $\C\setminus G$ (i.e. $|z|-r_{1} \asymp n^{-1/2}$ and $|w|-r_{2} \asymp n^{-1/2}$, or vice-versa), the asymptotics are described by the oscillatory Szeg\H{o} kernel
\begin{equation}\label{wsg soft}
\mathcal{S}^G_{\mathrm{soft}}(z,w;n):=\frac 1 {2\pi}\sum_{\ell=-\infty}^{+\infty}\frac {(z\bar{w})^\ell}{\frac{{r_1}^{2\ell+1-2x}}{\sqrt{\Delta Q(r_1)}}+\frac {r_2^{2\ell+1-2x}}{\sqrt{\Delta Q(r_2)}}},
\end{equation}
where $x=x(n):= n \int_{|z|\leq r_{1}}\mu(d^{2}z)-\lfloor n \int_{|z|\leq r_{1}}\mu(d^{2}z) \rfloor \in [0,1)$ and $\mu$ is the equilibrium measure associated with $Q$. $\mathcal{S}^G_{\mathrm{soft}}(z,w;n)$ is an analytic function of $z,\bar{w}\in G$ which is also well-defined when $|z|-r_{1} \asymp n^{-1/2}$ and $|w|-r_{2} \asymp n^{-1/2}$ (or vice-versa). This function is also the reproducing kernel for the weighted Hardy space $H^2(G;n)$ consisting of all analytic functions on $G$ such that
\begin{equation}\label{wnorm soft}
\|f\|_{H^2(G;n)}^2:=\int_{\partial G}|f(z)|^2|z|^{-2x(n)}(\Delta Q(z))^{-\frac 1 2}\,|dz|<\infty.
\end{equation}
The situation when $z,w$ are far from each other but are on the same side of $G$ gets more complicated because the series \eqref{wsg soft} is divergent if $z,w \in \partial G$ with $|z|=|w|$. For $z=r_{1}e^{i\theta_{1}}$ and $w=r_{1}e^{i\theta_{2}}$, the weighted Szeg\H{o} kernel appearing in the asymptotics of $K_{n}(z,w)$ turns out to be
\begin{multline}\label{r1r1form}
\mathcal{S}^G_{\mathrm{soft}}(r_{1}e^{i\theta_{1}},r_{1}e^{i\theta_{2}};n):= \frac{1}{2\pi} \frac{\sqrt{\Delta Q(r_1)}}{r_1^{1-2x}}  \bigg(  \frac{1}{e^{i( \theta_1-\theta_2)}-1}   \\
+ \sum\limits_{\ell=0}^{\infty} \frac{e^{i(\theta_1-\theta_2)\ell }}{1 + \sqrt{\frac{\Delta Q(r_{1})}{\Delta Q(r_{2})}}  (\frac{r_{2}}{r_{1}})^{2\ell-2x+1}}    -\sum\limits_{\ell=-\infty}^{-1} \frac{ e^{i(\theta_1-\theta_2)\ell}  }{1+\sqrt{\frac{\Delta Q(r_{2})}{\Delta Q(r_{1})}} (\frac{r_{1}}{r_{2}})^{2\ell -2x+1}}\bigg).
\end{multline}
As noticed in \cite{ACC2022}, the series in \eqref{r1r1form} can be recognized as the Abel limit of \eqref{wsg soft}; more precisely, $\mathcal{S}^G_{\mathrm{soft}}(r_{1}e^{i\theta_{1}},r_{1}e^{i\theta_{2}};n)$ in \eqref{r1r1form} is equal to $\lim_{r \searrow r_{1}} \mathcal{S}^G_{\mathrm{soft}}(re^{\theta_{1}},re^{\theta_{2}};n)$, where $\mathcal{S}^G_{\mathrm{soft}}(re^{\theta_{1}},re^{\theta_{2}};n)$ is as in \eqref{wsg soft} with $z=re^{i\theta_{1}}$ and $w=re^{i\theta_{2}}$.

\medskip It is apriori not clear at all whether analogues to the above results from \cite{AC, ACC2022} exist near a hard edge. For example, $\Delta Q(r_{1})$ appears in \eqref{wsg soft}, \eqref{wnorm soft} and \eqref{r1r1form}, but in our setting $\{|z|=r_{1}\}$ is a hard edge and $\Delta Q(r_{1})$ does not even make sense. It therefore came as a surprise to us that our findings completely mimic the results of \cite{ACC2022}: the main difference is that the quantities $\sqrt{\Delta Q(r_{1})}$ and $\sqrt{\Delta Q(r_{2})}$ appearing in \eqref{wsg soft}--\eqref{r1r1form} should here be replaced by $\sigma_{1}/r_{1}$ and $\sigma_{2}/r_{2}$, respectively. More precisely, we have the following definition.
\begin{definition}
The oscillatory Szeg\H{o} kernel $\mathcal{S}^{G}_{\mathrm{hard}}(z,w;n)$ associated with $G = \{r_{1} < |z| < r_{2}\}$ and $n$ is defined for $z, w \in G$ by
\begin{align}\label{wsg}
\mathcal{S}^{G}_{\mathrm{hard}}(z,w;n) := \frac{1}{2\pi} \sum_{\ell=-\infty}^{+\infty} \frac{(z\bar{w})^{\ell}}{\frac{r_{1}^{2(\ell+1-x)}}{\sigma_{1}} + \frac{r_{2}^{2(\ell+1-x)}}{\sigma_{2}} }.
\end{align}
This is an analytic function of $z, \bar{w} \in G$. This function is also well-defined as long as $|z\overline{w}|$ lies in a compact subset of $(r_{1}^{2},r_{2}^{2})$; for example, it is well-defined for $|z|=r_{1}$ and $|w|=r_{2}$ (and vice-versa, it is also well-defined for $|z|=r_{2}$ and $|w|=r_{1}$). For $|z|=|w|=r_{1}$, $z\neq w$, the above series does not converge absolutely. Nevertheless, for $\theta_{1} \neq \theta_{2}$, the Abel limit $\lim_{r \searrow 1} \mathcal{S}^{G}_{\mathrm{hard}}(re^{\theta_{1}},re^{\theta_{2}};n)$ exists, we denote it by $\mathcal{S}^G_{\mathrm{hard}}(r_{1}e^{i\theta_{1}},r_{1}e^{i\theta_{2}};n)$, and is equal to
\begin{multline}\label{Abel series}
\mathcal{S}^G_{\mathrm{hard}}(r_{1}e^{i\theta_{1}},r_{1}e^{i\theta_{2}};n) = \\
\frac{1}{2\pi} \frac{\sigma_{1}}{r_{1}^{2(1-x)}} \bigg( \frac{1}{e^{i(\theta_{1}-\theta_{2})}-1} + \sum_{\ell=0}^{+\infty} \frac{e^{i(\theta_{1}-\theta_{2})\ell}}{1+\frac{\sigma_{1}}{\sigma_{2}}(\frac{r_{2}}{r_{1}})^{2(\ell+1-x)}} - \sum_{\ell=-\infty}^{-1} \frac{e^{i(\theta_{1}-\theta_{2})\ell}}{1+\frac{\sigma_{2}}{\sigma_{1}}(\frac{r_{1}}{r_{2}})^{2(\ell+1-x)}} \bigg).
\end{multline}
Because one had to take the Abel limit to make sense of $\mathcal{S}^G_{\mathrm{hard}}(r_{1}e^{i\theta_{1}},r_{1}e^{i\theta_{2}};n)$, we call $\mathcal{S}^G_{\mathrm{hard}}(r_{1}e^{i\theta_{1}}$,  $r_{1}e^{i\theta_{2}};n)$ an oscillatory ``regularized" Szeg\H{o} kernel, see also Figure \ref{fig:summary}.
\end{definition}
\begin{remark} The kernel \eqref{wsg} is the reproducing kernel for the weighted Hardy space $H^2(G;n)$ of all analytic functions $g$ on $G$ satisfying
\begin{equation}\label{wnorm}
\|f\|_{H^2(G;n)}^2:=\int_{\{|z|=r_{1}\}}|f(z)|^2 \frac{r_{1}^{1-2x(n)}}{\sigma_{1}}|dz| + \int_{\{|z|=r_{2}\}}|f(z)|^2 \frac{r_{2}^{1-2x(n)}}{\sigma_{2}}\,|dz|<\infty.
\end{equation}
\end{remark}
We next state our results on the asymptotics of $K_{n}(z,w)$ in the macroscopic regime when $z$ and $w$ are in different regions of $\C\setminus G$.
\begin{theorem}\label{thm:r1-r2 case}\emph{(``$r_{1}$-$r_{2}$ hard edge case")}
Let $t_{1},t_{2}\geq 0$ and $\theta_{1},\theta_{2}\in \R$ be fixed, and define
\begin{align}\label{def of z1z2 r1 r2 case intro}
z = r_{1}\Big(1-\frac{t_{1}}{\sigma_{1}n}\Big)e^{i\theta_{1}}, \qquad w = r_{2}\Big(1+\frac{t_{2}}{\sigma_{2}n}\Big)e^{i\theta_{2}}, \qquad t_{1}, t_{2} \geq 0, \;\; \theta_{1},\theta_{2}\in \mathbb{R}.
\end{align}
As $n \to + \infty$, we have
\begin{align*}
& K_{n}(z,w) = 2 \pi n \cdot \mathcal{S}^{G}_{\mathrm{hard}}(z,w;n) \cdot e^{i \lfloor j_{\star}\rfloor (\theta_{1}-\theta_{2})} (r_{1}r_{2})^{-x} e^{-t_{1}-t_{2}}  + \bigO((\log n)^{2}),
\end{align*}
where $j_{\star} := n \sigma_{\star}-\alpha$, and $x=x(n)$ is given by
\begin{align}\label{def of theta star intro}
x = j_{\star}-\lfloor j_{\star} \rfloor = n \sigma_{\star}-\alpha - \lfloor n \sigma_{\star}-\alpha \rfloor \in [0,1).
\end{align}
\end{theorem}
\begin{remark}\label{remark:r1r2 case}
The analogue of Theorem \ref{thm:r1-r2 case} near soft edges is given by \cite[Corollary 1.14]{ACC2022} and is as follows: for $z=\big(r_1 + \frac{t}{\sqrt{n\Delta Q(r_1)}}\big)e^{i\theta_1}$, $w=\big(r_2 + \frac{s}{\sqrt{n\Delta Q(r_2)}}\big)e^{i\theta_2}$ with $s,t\in \R$, we have
\begin{align*}
K_{n}(z,w) = \sqrt{2\pi n} \cdot \mathcal{S}_{\mathrm{soft}}^G(z,w;n)\cdot e^{i\lfloor j_{\star}\rfloor(\theta_1-\theta_2)} (r_1r_2)^{-x} e^{-t^2-s^2}
+\bigO((\log n)^{3}),
\end{align*}
where $\mathcal{S}_{\mathrm{soft}}^G(z,w;n)$ is given by \eqref{wsg soft}, $j_{\star}:=n \int_{|z|\leq r_{1}}\mu(d^{2}z)$ and $x$ and $\mu$ are as in \eqref{wsg soft}.
\end{remark}
\begin{figure}
\begin{center}
\begin{tikzpicture}[master]
\node at (0,0) {\includegraphics[width=7cm]{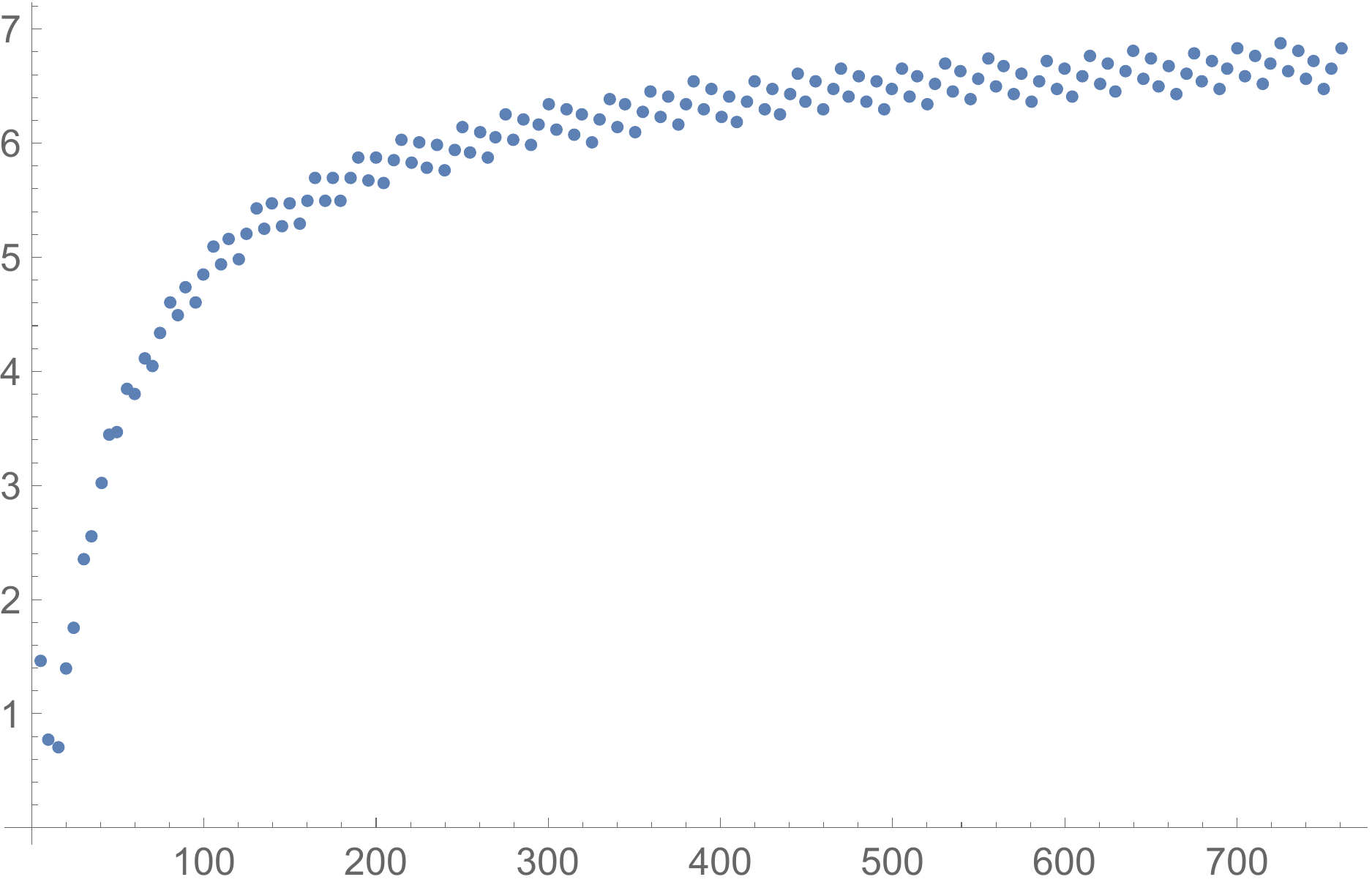}};
%\node at (0,0) {$K_{n}(z,w) - \big(C_{1} n^{2} + C_{2} \, n \log n + \Big(C_{3}+\frac{\sigma_{1}}{r_{1}^{2}}e^{-t_{1}-t_{2}}\mathcal{F}_{n}\Big) n + C_{4} \sqrt{n}\big)$};
\end{tikzpicture}
\begin{tikzpicture}[slave]
\node at (0,0) {\includegraphics[width=7cm]{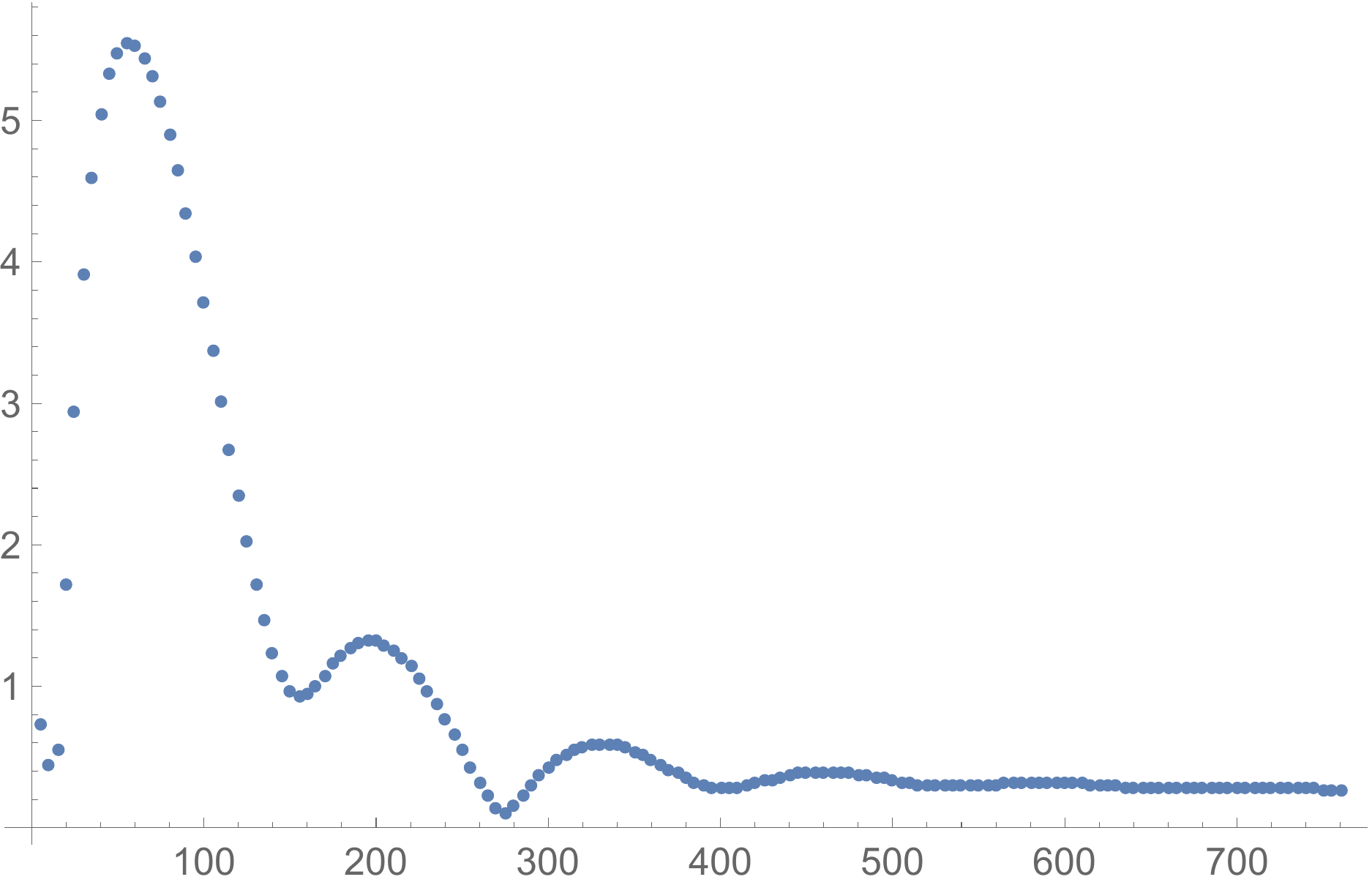}};
%\node at (0,0) {$K_{n}(z,w) - \big(C_{1} n^{2} + C_{2} \, n \log n + \Big(C_{3}+\frac{\sigma_{1}}{r_{1}^{2}}e^{-t_{1}-t_{2}}\mathcal{F}_{n}\Big) n + C_{4} \sqrt{n}\big)$};
\end{tikzpicture}
\end{center}
\caption{Numerical confirmations of Theorems \ref{thm:r1-r2 case} (left) and \ref{thm:r1r1 hard} (right). For both pictures, $b=1.3$, $\alpha=1.26$, $r_{1}=0.42\smash{b^{-\frac{1}{2b}}}$, $r_{2}=0.67\smash{b^{-\frac{1}{2b}}}$. Left: the function $n\mapsto \big|\big[K_{n}(z,w) - 2 \pi n \cdot \mathcal{S}^{G}_{\mathrm{hard}}(z,w;n) \cdot e^{i \lfloor j_{\star}\rfloor (\theta_{1}-\theta_{2})} (r_{1}r_{2})^{-x} e^{-t_{1}-t_{2}}\big|$, with $z = r_{1}(1-\frac{t_{1}}{\sigma_{1}n})e^{i\theta_{1}}$, $w = r_{2}(1+\frac{t_{2}}{\sigma_{2}n})e^{i\theta_{2}}$, $t_{1}=0.21$, $t_{2}=0.45$, $\theta_{1}=0$, $\theta_{2}=0.312$. This function seems to grow very slowly as $n\to+\infty$, which is consistent with Theorem \ref{thm:r1-r2 case}. Right: the function $n\mapsto \frac{1}{\sqrt{n}} \big|K_{n}(z,w) - 2\pi n \cdot \mathcal{S}^{G}_{\mathrm{hard}}(r_{1}e^{\theta_{1}},r_{1}e^{\theta_{2}};n) \cdot e^{i \lfloor j_{\star} \rfloor (\theta_{1}-\theta_{2})} r_{1}^{-2x} e^{-t_{1}-t_{2}} \big|$, with $z = r_{1}(1-\frac{t_{1}}{\sigma_{1}n})e^{i\theta_{1}}$, $w = r_{1}(1-\frac{t_{2}}{\sigma_{1}n})e^{i\theta_{2}}$, $t_{1}=0.91$, $t_{2}=1.45$, $\theta_{1}=0$, $\theta_{2}=0.312$. This function seems to approach a constant as $n\to+\infty$, which is consistent with Theorem \ref{thm:r1r1 hard}, and also suggests that the term $\bigO(\sqrt{n \log n} )$ in \eqref{lol19} is actually $\bigO(\sqrt{n} )$.
}
\label{fig:numerical check}
\end{figure}
We next consider the macroscopic regime where $z$ and $w$ are on different sides of $G$.
\begin{theorem}\label{thm:r1r1 hard}\emph{(``$r_{1}$-$r_{1}$ hard edge case")}
Let $t_{1},t_{2}\geq 0$ and $\theta_{1}\neq \theta_{2} \in \R$ be fixed, and define
\begin{align}\label{def of z1z2 r1 r1 case intro r1r1}
z = r_{1}\Big(1-\frac{t_{1}}{\sigma_{1}n}\Big)e^{i\theta_{1}}, \qquad w = r_{1}\Big(1-\frac{t_{2}}{\sigma_{1}n}\Big)e^{i\theta_{2}}.
\end{align}
As $n \to \infty$, we have
\begin{align}\label{lol19}
K_{n}(z,w) = 2\pi n \cdot \mathcal{S}^{G}_{\mathrm{hard}}(r_{1}e^{\theta_{1}},r_{1}e^{\theta_{2}};n) \cdot e^{i \lfloor j_{\star} \rfloor (\theta_{1}-\theta_{2})} r_{1}^{-2x} e^{-t_{1}-t_{2}} +  \bigO(\sqrt{n \log n} ),
\end{align}
where $j_{\star} := n \sigma_{\star}-\alpha$, $x = j_{\star}-\lfloor j_{\star} \rfloor$.
\end{theorem}
\begin{remark}
The analogue of Theorem \ref{thm:r1r1 hard} near a soft edge is given by \cite[Corollary 1.14]{ACC2022} and is as follows: for $z=\big(r_1 + \frac{t}{\sqrt{n\Delta Q(r_1)}}\big)e^{i\theta_1}$, $w=\big(r_1 + \frac{s}{\sqrt{n\Delta Q(r_2)}}\big)e^{i\theta_2}$ with $s,t\in \R$, $\theta_{1}\neq \theta_{2}$, we have
\begin{align*}
K_{n}(z,w) = \sqrt{2\pi n} \cdot \mathcal{S}_{\mathrm{soft}}^G(r_{1}e^{i\theta_{1}},r_{1}e^{i\theta_{2}};n)\cdot e^{i\lfloor j_{\star}\rfloor(\theta_1-\theta_2)} r_1^{-2x} e^{-t^2-s^2}
+\bigO((\log n)^{5}),
\end{align*}
where $\mathcal{S}_{\mathrm{soft}}^G(r_{1}e^{i\theta_{1}},r_{1}e^{i\theta_{2}};n)$ is given by \eqref{r1r1form}, $j_{\star}=n \int_{|z|\leq r_{1}}\mu(d^{2}z)$, and $x$ and $\mu$ are as in \eqref{wsg soft}.
\end{remark}

Our last result shows that $K_{n}(z,w)$ is small whenever $z$ and $w$ are in the semi-hard edge regime and $|z-w|\asymp 1$.
\begin{theorem}\label{thm:r1r1 semi-hard}\emph{(``$r_{1}$-$r_{1}$ semi-hard edge case")}
Let $\mathfrak{s}_{1}, \mathfrak{s}_{2} > 0$ and $\theta_{1}\neq \theta_{2} \in \R$ be fixed, and define
\begin{align*}
& z = r_{1}\Big(1-\frac{\mathfrak{s}_{1}}{br_{1}^{b}\sqrt{2n}}\Big)e^{i\theta_{1}}, \qquad w = r_{1}\Big(1-\frac{\mathfrak{s}_{2}}{br_{1}^{b}\sqrt{2n}}\Big)e^{i\theta_{2}}.
\end{align*}
As $n \to + \infty$,
\begin{align}\label{small corr}
& K_{n}(z,w) = \bigO(1).
\end{align}
\end{theorem}
\begin{remark}
The estimate \eqref{small corr} can be probably be strengthened with more efforts. In fact, in the setting of Theorem \ref{thm:r1r1 semi-hard}, we believe that $K_{n}(z,w) = \bigO(n^{-K})$ holds as $n \to  \infty$ for any fixed $K>0$.
\end{remark}

\textbf{Outline of the paper.} In Section \ref{section:hj}, we obtain large $n$ estimates for $h_{j}$ valid uniformly in different ranges of $j\in \{1,\ldots,n\}$. These estimates will be useful for the proofs of each of our five theorems. The proof of Theorem \ref{thm:r1-r2 case} is given in Section \ref{section:4}, the proofs of Theorems \ref{thm:r1 hard} and \ref{thm:r1r1 hard} are given in Section \ref{section:r1r1 1/n}, and the proofs of Theorems \ref{thm:r1 semi-hard} and \ref{thm:r1r1 semi-hard} are given in Section \ref{section:r1r1 1/sqrtn}.

\section{Asymptotics of $h_{j}$}\label{section:hj}
Recall that $h_{j}$ is defined in \eqref{def of hj}. Following \cite{CharlierFH, CharlierAnnuli}, we introduce the following quantities: for $j=1,\ldots,n$ and $\ell=1,2$, we define
\begin{align}\label{def of aj lambdajl etajl}
a_{j} := \frac{j+\alpha}{b}, \qquad \lambda_{j,\ell} := \frac{bnr_{\ell}^{2b}}{j+\alpha}, \qquad \eta_{j,\ell} := (\lambda_{j,\ell}-1)\sqrt{\frac{2 (\lambda_{j,\ell}-1-\log \lambda_{j,\ell})}{(\lambda_{j,\ell}-1)^{2}}}.
\end{align}
Let $\epsilon > 0$ be a small constant independent of $n$. Define
\begin{align}
& j_{\ell,-}:=\lceil \tfrac{bnr_{\ell}^{2b}}{1+\epsilon} - \alpha \rceil, & & j_{\ell,+} := \lfloor  \tfrac{bnr_{\ell}^{2b}}{1-\epsilon} - \alpha \rfloor, & & \ell=1,2, \label{def of jk plus and minus} \\
& j_{0,-}:=1, & & j_{0,+}:=M', & & j_{3,-}:=n+1, \nonumber
\end{align}
where $\lceil x \rceil$ denotes the smallest integer $\geq x$, $\lfloor  x \rfloor$ denotes the largest integer $\leq x$, and $M'$ is a large but fixed constant. We take $\epsilon$ sufficiently small such that
\begin{align}\label{cond on epsilon 1}
\frac{br_{1}^{2b}}{1-\epsilon} < \frac{br_{2}^{2b}}{1+\epsilon},  \qquad \mbox{and} \qquad \frac{br_{2}^{2b}}{1-\epsilon} < 1.
\end{align}
The quantity $\sigma_{\star}$, which we recall is defined in \eqref{def of taustar}, will appear naturally in our analysis. It is easy to check that $\sigma_{\star} \in (br_{1}^{2b},br_{2}^{2b})$. For technical reasons, we also assume $\epsilon>0$ is small enough so that
\begin{align}\label{cond on epsilon 2}
\frac{br_{1}^{2b}}{1-\epsilon} < \sigma_{\star} < \frac{br_{2}^{2b}}{1+\epsilon}.
\end{align}
Recall also that $j_{\star} := n \sigma_{\star} -\alpha$, and note that
\begin{align}\label{asymp etajk-etajkm1 r1r1}
\frac{a_{j}(\eta_{j,2}^{2}-\eta_{j,1}^{2})}{2} = 2(j_{\star}-j) \log\frac{r_{2}}{r_{1}}.
\end{align}
Let $M$ be such that $M'\sqrt{\log n}\leq M \leq n^{\frac{1}{10}}$, and define
\begin{align*}
g_{\ell,-} := \bigg\lceil \frac{bnr_{\ell}^{2b}}{1+\frac{M}{\sqrt{n}}}-\alpha \bigg\rceil, \qquad g_{\ell,+} := \bigg\lfloor \frac{bnr_{\ell}^{2b}}{1-\frac{M}{\sqrt{n}}}-\alpha \bigg\rfloor, \qquad \ell=1,2,
\end{align*}
and
\begin{align}\label{def of Mjk}
M_{j,k} := \sqrt{n}(\lambda_{j,k}-1), \qquad \mbox{for all } k \in \{1,2\} \mbox{ and } j \in \{g_{k,-},\ldots,g_{k,+}\}.
\end{align}
In this section, $M$ can be arbitrary within the range $M \in [M'\sqrt{\log n}, n^{\frac{1}{10}}]$, but we already mention that in Section \ref{section:4} and Subsection \ref{subsection:5.2} we will take $M=M'\sqrt{\log n}$, in Subsection \ref{subsection:5.1} we will take $M=n^{\frac{1}{10}}$, and in Section \ref{section:r1r1 1/sqrtn} we will take $M=M' \log n$. For $k=1,2$ and $g_{k,-}\leq j \leq g_{k,+}$, we also define
\begin{align*}
\chi_{j,k}^{-} := \chi_{(-\infty,0)}(M_{j,k}), \qquad \chi_{j,k}^{+} := \chi_{(0,+\infty)}(M_{j,k}).
\end{align*}
We start with the following useful lemma.
\begin{lemma}\label{lemma:asymp prefactor}
As $n \to + \infty$ with $M' \leq j \leq n$, we have
\begin{align}\label{lol10}
\frac{n^{\frac{j+\alpha}{b}}}{\Gamma (\frac{j+\alpha}{b})} = \frac{\sqrt{n} \, e^{\frac{n}{b}(j/n-j/n \log(\frac{j/n}{b}))}}{\sqrt{2\pi}}\bigg( \frac{b}{j/n} \bigg)^{\frac{\alpha}{b}} \sqrt{\frac{j/n}{b}} \bigg( 1 + \frac{-b^{2}+6b \alpha - 6 \alpha^{2}}{12 b j/n} \frac{1}{n} + \bigO(j^{-2}) \bigg).
\end{align}
As $n \to + \infty$ with $g_{k,-} \leq j \leq g_{k,+}$, $k=1,2$, we have
\begin{align}
& \frac{n^{\frac{j+\alpha}{b}}}{\Gamma (\frac{j+\alpha}{b})} = e^{r_{k}^{2b}(1-\log(r_{k}^{2b}))n}e^{M_{j,k} r_{k}^{2b} \log(r_{k}^{2b})\sqrt{n}}e^{-r_{k}^{2b}(\frac{1}{2}+\log(r_{k}^{2b}))M_{j,k}^{2}} \frac{r_{k}^{b}\sqrt{n}}{\sqrt{2\pi}} \nonumber \\
& \times \bigg( 1+\frac{M_{j,k}}{6}(-3+5M_{j,k}^{2}r_{k}^{2b}+6M_{j,k}^{2}r_{k}^{2b}\log(r_{k}^{2b})) \frac{1}{\sqrt{n}} + \frac{1}{n}\bigg\{ r_{k}^{4b} M_{j,k}^{6} \bigg( \frac{25}{72} + \frac{5\log(r_{k}^{2b})}{6} + \frac{(\log(r_{k}^{2b}))^{2}}{2} \bigg)  \nonumber \\
& - \frac{3}{2}r_{k}^{2b}M_{j,k}^{4} (1+\log(r_{k}^{2b})) \bigg\} + \frac{125+450 \log(r_{k}^{2b}) + 540 \log^{2}(r_{k}^{2b}) + 216\log^{3}(r_{k}^{2b})}{1296}r_{k}^{6b}M_{j,k}^{9}n^{-3/2} \nonumber \\
& + \bigO\bigg(\frac{1+M_{j,k}^{2}}{n} + \frac{1+M_{j,k}^{7}}{n^{3/2}} + \frac{1+M_{j,k}^{12}}{n^{2}} \bigg) \bigg). \label{lol11}
\end{align}
\end{lemma}

\begin{proof}
Formula \eqref{lol10} directly follows from the well-known large $j$ asymptotics of $\Gamma(\frac{j+\alpha}{b})$ (see e.g. \cite[5.11.1]{NIST}), and \eqref{lol11} directly follows from the large $j$ asymptotics of $\Gamma(\frac{j+\alpha}{b})$ and the fact that, by \eqref{def of aj lambdajl etajl} and \eqref{def of Mjk}, we have
\begin{align*}
\frac{j}{n} = -\frac{\alpha}{n} + \frac{br_{k}^{2b}}{1+\frac{M_{j,k}}{\sqrt{n}}}, \qquad g_{k,-} \leq j \leq g_{k,+}, \; k=1,2.
\end{align*}
\end{proof}
\begin{lemma}\label{lemma: asymp of hjinv} $M'$ can be chosen sufficiently large such that the following asymptotics hold uniformly for $M'\sqrt{\log n}\leq M \leq n^{\frac{1}{10}}$.
\begin{enumerate}
\item Let $j$ be fixed. As $n \to + \infty$, we have
\begin{align}\label{hj asymp 1}
h_{j}^{-1} = \frac{b n^{\frac{j+\alpha}{b}}}{\Gamma (\frac{j+\alpha}{b})} \bigg( 1 + \bigO(n^{\frac{j+\alpha}{b}-1}e^{-nr_{1}^{2b}}) \bigg).
\end{align}
\item As $n \to + \infty$ with $M' \leq j \leq j_{1,-}$, we have
\begin{align}
h_{j}^{-1} & = \frac{b n^{\frac{j+\alpha}{b}}}{\Gamma (\frac{j+\alpha}{b})} \bigg( 1 + \bigO(n^{-\frac{1}{2}}e^{-nr_{1}^{2b}\frac{\epsilon - \log(1+\epsilon)}{1+\epsilon}} \bigg) \nonumber \\
& = b \frac{\sqrt{n} \, e^{\frac{n}{b}(j/n-j/n \log(\frac{j/n}{b}))}}{\sqrt{2\pi}}\bigg( \frac{b}{j/n} \bigg)^{\frac{\alpha}{b}} \sqrt{\frac{j/n}{b}} \bigg( 1 + \frac{-b^{2}+6b \alpha - 6 \alpha^{2}}{12 b j/n} \frac{1}{n} + \bigO(j^{-2}) \bigg). \label{hj asymp 2}
\end{align}
\item As $n \to + \infty$ with $j_{1,-} \leq j \leq g_{1,-}$, we have
\begin{align}
h_{j}^{-1} & = \frac{b n^{\frac{j+\alpha}{b}}}{\Gamma (\frac{j+\alpha}{b})} \bigg( 1 + \bigO(n^{-100}) \bigg) \nonumber \\
& = b \frac{\sqrt{n} \, e^{\frac{n}{b}(j/n-j/n \log(\frac{j/n}{b}))}}{\sqrt{2\pi}}\bigg( \frac{b}{j/n} \bigg)^{\frac{\alpha}{b}} \sqrt{\frac{j/n}{b}} \bigg( 1 + \frac{-b^{2}+6b \alpha - 6 \alpha^{2}}{12 b j/n} \frac{1}{n} + \bigO(n^{-2}) \bigg). \label{hj asymp 3}
\end{align}
\item As $n \to + \infty$ with $g_{1,-} \leq j \leq g_{1,+}$, we have
\begin{align}
& h_{j}^{-1} = \frac{b n^{\frac{j+\alpha}{b}}}{\Gamma (\frac{j+\alpha}{b})} \frac{1}{\frac{1}{2}\mathrm{erfc}(-\frac{M_{j,1}r_{1}^{b}}{\sqrt{2}}) } \bigg( 1 + \frac{(5M_{j,1}^{2}r_{1}^{2b}-2)e^{-\frac{r_{1}^{2b}M_{j,1}^{2}}{2}}}{3\sqrt{2\pi} r_{1}^{b} \mathrm{erfc}(-\frac{M_{j,1}r_{1}^{b}}{\sqrt{2}}) \sqrt{n}} + \frac{\chi_{j,1}^{-}}{n} \bigg\{ \frac{25r_{1}^{4b}M_{j,1}^{6}}{72}  \nonumber \\
& + \frac{3r_{1}^{2b}M_{j,1}^{4}}{2} \bigg\} - \chi_{j,1}^{-}\frac{125r_{1}^{6b}M_{j,1}^{9}}{1296n^{3/2}} + \bigO\bigg(\frac{1+M_{j,1}^{2}}{n}+\frac{1+|M_{j,1}^{7}|}{n^{3/2}} + \frac{1+M_{j,1}^{12}}{n^{2}}\bigg)     \bigg) \nonumber \\
& = b \frac{e^{r_{1}^{2b}(1-\log(r_{1}^{2b}))n}e^{M_{j,1} r_{1}^{2b} \log(r_{1}^{2b})\sqrt{n}}e^{-r_{1}^{2b}(\frac{1}{2}+\log(r_{1}^{2b}))M_{j,1}^{2}} r_{1}^{b}\sqrt{n}}{\frac{1}{2}\mathrm{erfc}(-\frac{M_{j,1}r_{1}^{b}}{\sqrt{2}})\sqrt{2\pi}} \nonumber \\
& \times \bigg( 1 + \frac{1}{\sqrt{n}}\bigg\{\frac{(5M_{j,1}^{2}r_{1}^{2b}-2)e^{-\frac{r_{1}^{2b}M_{j,1}^{2}}{2}}}{3\sqrt{2\pi} r_{1}^{b} \mathrm{erfc}(-\frac{M_{j,1}r_{1}^{b}}{\sqrt{2}}) } + M_{j,1}^{3}r_{1}^{2b}\log(r_{1}^{2b}) - \frac{M_{j,1}}{2} + \frac{5M_{j,1}^{3}r_{1}^{2b}}{6} \bigg\} \nonumber \\
& + \frac{\frac{r_{1}^{4b}}{2} (\log(r_{1}^{2b}))^{2} M_{j,1}^{6} - 2 r_{1}^{2b} \log(r_{1}^{2b})M_{j,1}^{4}}{n} + \frac{r_{1}^{6b} (\log(r_{1}^{2b}))^{3}M_{j,1}^{9}}{6n^{3/2}} \nonumber \\
& + \bigO\Big(\frac{1+M_{j,1}^{2}+M_{j,1}^{6}\chi_{j,1}^{+}}{n}+\frac{1+|M_{j,1}^{7}|+M_{j,1}^{9}\chi_{j,1}^{+}}{n^{3/2}} + \frac{1+M_{j,1}^{12}}{n^{2}}\Big)   \bigg). \label{hj asymp 4}
\end{align}
\item As $n \to + \infty$ with $g_{1,+} \leq j \leq j_{1,+}$, we have
\begin{align}
& h_{j}^{-1} = \frac{b n^{\frac{j+\alpha}{b}}}{\Gamma (\frac{j+\alpha}{b})} \sqrt{n} \, e^{\frac{n}{b}(br_{1}^{2b}-j/n+j/n \log (\frac{j/n}{br_{1}^{2b}}))} \frac{\sqrt{2\pi}}{\sqrt{b \, j/n}} \bigg( \frac{j/n}{br_{1}^{2b}} \bigg)^{\frac{\alpha}{b}} ( j/n - br_{1}^{2b} )  \nonumber \\
& \hspace{1cm} \times \bigg( 1 + \frac{12b^{3} r_{1}^{2b} j/n + 6b((j/n)^{2}-(b r_{1}^{2b})^{2})\alpha + (b^{2}+6\alpha^{2}) (j/n-br_{1}^{2b})^{2}}{12 b j/n (j/n-br_{1}^{2b})^{2} n} \nonumber \\
& \hspace{1cm} - \frac{2 b^{4} r_{1}^{4b}}{n^{2}(j/n-br_{1}^{2b})^{4}} + \frac{10b^{6}r_{1}^{6b}}{n^{3}(j/n-br_{1}^{2b})^{6}} + \bigO\Big( \frac{1}{n^{2}(j/n-br_{1}^{2b})^{3}} + \frac{1}{n^{4}(j/n-br_{1}^{2b})^{8}} \Big) \bigg) \nonumber \\
& = \frac{n(j/n-br_{1}^{2b})}{r_{1}^{2\alpha}}e^{n(r_{1}^{2b}-2 j/n \log(r_{1}))} \bigg( 1 + \frac{b^{2}r_{1}^{2b}+(j/n-br_{1}^{2b})\alpha}{n(j/n-br_{1}^{2b})^{2}} - \frac{2 b^{4} r_{1}^{4b}}{n^{2}(j/n-br_{1}^{2b})^{4}} \nonumber \\
& + \frac{10b^{6}r_{1}^{6b}}{n^{3}(j/n-br_{1}^{2b})^{6}} + \bigO\Big( \frac{1}{n^{2}(j/n-br_{1}^{2b})^{3}} + \frac{1}{n^{4}(j/n-br_{1}^{2b})^{8}} \Big) \bigg). \label{hj asymp 5}
\end{align}
\item As $n \to + \infty$ with $j_{1,+} \leq j \leq \lfloor j_{\star}\rfloor$, we have
\begin{align}
& h_{j}^{-1} = \frac{b n^{\frac{j+\alpha}{b}}}{\Gamma (\frac{j+\alpha}{b})} \sqrt{n} \, e^{\frac{n}{b}(br_{1}^{2b}-j/n+j/n \log (\frac{j/n}{br_{1}^{2b}}))} \frac{\sqrt{2\pi}}{\sqrt{b \, j/n}} \frac{\big( \frac{j/n}{br_{1}^{2b}} \big)^{\frac{\alpha}{b}}}{\frac{1}{j/n-br_{1}^{2b}}+\frac{(\frac{r_{1}}{r_{2}})^{2(j_{\star}-j)}}{br_{2}^{2b}-j/n}} \nonumber \\
& \hspace{1cm} \times \bigg( 1 + \frac{12b^{3} r_{1}^{2b} j/n + 6b((j/n)^{2}-(b r_{1}^{2b})^{2})\alpha + (b^{2}+6\alpha^{2}) (j/n-br_{1}^{2b})^{2}}{12 b j/n (j/n-br_{1}^{2b})^{2} n} \nonumber \\
& \hspace{1cm} + \bigO\Big( \frac{(\frac{r_{1}}{r_{2}})^{2(j_{\star}-j)}}{n} + \frac{1}{n^{2}} \Big) \bigg) \nonumber \\
& = \frac{n}{r_{1}^{2\alpha}}\frac{e^{n(r_{1}^{2b}-2 j/n \log(r_{1}))}}{\frac{1}{j/n-br_{1}^{2b}}+\frac{(\frac{r_{1}}{r_{2}})^{2(j_{\star}-j)}}{br_{2}^{2b}-j/n}} \bigg( 1 + \frac{b^{2}r_{1}^{2b}+(j/n-br_{1}^{2b})\alpha}{n(j/n-br_{1}^{2b})^{2}} + \bigO\Big( \frac{(\frac{r_{1}}{r_{2}})^{2(j_{\star}-j)}}{n} + \frac{1}{n^{2}} \Big) \bigg). \label{hj asymp 6}
\end{align}
\item As $n \to + \infty$ with $\lfloor j_{\star}\rfloor+1 \leq j \leq j_{2,-}$, we have
\begin{align}
& h_{j}^{-1} = \frac{b n^{\frac{j+\alpha}{b}}}{\Gamma (\frac{j+\alpha}{b})} \sqrt{n} \, e^{\frac{n}{b}(br_{2}^{2b}-j/n+j/n \log (\frac{j/n}{br_{2}^{2b}}))} \frac{\sqrt{2\pi}}{\sqrt{b \, j/n}} \frac{\big( \frac{j/n}{br_{2}^{2b}} \big)^{\frac{\alpha}{b}}}{\frac{1}{br_{2}^{2b} - j/n}+\frac{(\frac{r_{1}}{r_{2}})^{2(j-j_{\star})}}{j/n-br_{1}^{2b}}} \nonumber \\
& \hspace{1cm} \times \bigg( 1 + \frac{12b^{3} r_{2}^{2b} j/n + 6b((j/n)^{2}-(b r_{2}^{2b})^{2})\alpha + (b^{2}+6\alpha^{2}) (j/n-br_{2}^{2b})^{2}}{12 b j/n (j/n-br_{2}^{2b})^{2} n} \nonumber \\
& \hspace{1cm} + \bigO\Big( \frac{(\frac{r_{1}}{r_{2}})^{2(j-j_{\star})}}{n} + \frac{1}{n^{2}} \Big) \bigg) \nonumber \\
& = \frac{n}{r_{2}^{2\alpha}} \frac{e^{n(r_{2}^{2b}-2j/n \log(r_{2}))}}{\frac{1}{br_{2}^{2b} - j/n}+\frac{(\frac{r_{1}}{r_{2}})^{2(j-j_{\star})}}{j/n-br_{1}^{2b}}} \bigg( 1 + \frac{b^{2}r_{2}^{2b}+(j/n-br_{2}^{2b})\alpha}{n(j/n-br_{2}^{2b})^{2}} + \bigO\Big( \frac{(\frac{r_{1}}{r_{2}})^{2(j-j_{\star})}}{n} + \frac{1}{n^{2}} \Big) \bigg). \label{hj asymp 7}
\end{align}
\item As $n \to + \infty$ with $j_{2,-} \leq j \leq g_{2,-}$, we have
\begin{align}
& h_{j}^{-1} = \frac{b n^{\frac{j+\alpha}{b}}}{\Gamma (\frac{j+\alpha}{b})} \sqrt{n} \, e^{\frac{n}{b}(br_{2}^{2b}-j/n+j/n \log (\frac{j/n}{br_{2}^{2b}}))} \frac{\sqrt{2\pi}}{\sqrt{b \, j/n}} \bigg( \frac{j/n}{br_{2}^{2b}} \bigg)^{\frac{\alpha}{b}} ( br_{2}^{2b} - j/n )  \nonumber \\
& \hspace{1cm} \times \bigg( 1 + \frac{12b^{3} r_{2}^{2b} j/n + 6b((j/n)^{2}-(b r_{2}^{2b})^{2})\alpha + (b^{2}+6\alpha^{2}) (j/n-br_{2}^{2b})^{2}}{12 b j/n (j/n-br_{2}^{2b})^{2} n} \nonumber \\
& \hspace{1cm} + \bigO\Big( \frac{1}{n^{2}(j/n-br_{2}^{2b})^{4}} \Big) \bigg) \nonumber \\
& = \frac{n(br_{2}^{2b}-j/n)}{r_{2}^{2\alpha}} e^{n(r_{2}^{2b}-2j/n \log(r_{2}))} \bigg( 1 + \frac{b^{2}r_{2}^{2b}+(j/n-br_{2}^{2b})\alpha}{n(j/n-br_{2}^{2b})^{2}} + \bigO\Big( \frac{1}{n^{2}(j/n-br_{2}^{2b})^{4}} \Big) \bigg). \label{hj asymp 8}
\end{align}
\item As $n \to + \infty$ with $g_{2,-} \leq j \leq g_{2,+}$, we have
\begin{align}
& h_{j}^{-1} = \frac{b n^{\frac{j+\alpha}{b}}}{\Gamma (\frac{j+\alpha}{b})} \frac{1}{1-\frac{1}{2}\mathrm{erfc}(-\frac{M_{j,2}r_{2}^{b}}{\sqrt{2}}) } \bigg( 1 + \frac{(2-5M_{j,2}^{2}r_{2}^{2b})e^{-\frac{r_{2}^{2b}M_{j,2}^{2}}{2}}}{3\sqrt{2\pi} r_{2}^{b} (2-\mathrm{erfc}(-\frac{M_{j,2}r_{2}^{b}}{\sqrt{2}})) \sqrt{n}} + \frac{\chi_{j,2}^{+}}{n} \bigg\{ \frac{25r_{2}^{4b}M_{j,2}^{6}}{72} \nonumber \\
& + \frac{3r_{2}^{2b}M_{j,2}^{4}}{2} \bigg\} - \chi_{j,2}^{+}\frac{125r_{2}^{6b}M_{j,2}^{9}}{1296n^{3/2}} + \bigO\bigg(\frac{1+M_{j,2}^{2}}{n} + \frac{1+|M_{j,2}^{7}|}{n^{3/2}} + \frac{1+M_{j,2}^{12}}{n^{2}} \bigg)  \bigg) \nonumber \\
& = b \frac{e^{r_{2}^{2b}(1-\log(r_{2}^{2b}))n}e^{M_{j,2} r_{2}^{2b} \log(r_{2}^{2b})\sqrt{n}}e^{-r_{2}^{2b}(\frac{1}{2}+\log(r_{2}^{2b}))M_{j,2}^{2}} r_{2}^{b}\sqrt{n}}{(1-\frac{1}{2}\mathrm{erfc}(-\frac{M_{j,2}r_{2}^{b}}{\sqrt{2}}))\sqrt{2\pi}} \nonumber \\
& + \frac{\frac{r_{2}^{4b}}{2} (\log(r_{2}^{2b}))^{2} M_{j,2}^{6} - 2 r_{2}^{2b} \log(r_{2}^{2b})M_{j,2}^{4}}{n} + \frac{r_{2}^{6b} (\log(r_{2}^{2b}))^{3}M_{j,2}^{9}}{6n^{3/2}} \nonumber \\
& + \bigO\Big(\frac{1+M_{j,2}^{2}+M_{j,2}^{6}\chi_{j,2}^{-}}{n}+\frac{1+|M_{j,2}^{7}|+M_{j,2}^{9}\chi_{j,2}^{-}}{n^{3/2}} + \frac{1+M_{j,2}^{12}}{n^{2}}\Big)  \bigg). \label{hj asymp 9}
\end{align}
\item As $n \to + \infty$ with $g_{2,+} \leq j \leq j_{2,+}$, we have
\begin{align}
h_{j}^{-1} & = \frac{b n^{\frac{j+\alpha}{b}}}{\Gamma (\frac{j+\alpha}{b})} \bigg( 1 + \bigO(n^{-100}) \bigg) \nonumber \\
& = b \frac{\sqrt{n} \, e^{\frac{n}{b}(j/n-j/n \log(\frac{j/n}{b}))}}{\sqrt{2\pi}}\bigg( \frac{b}{j/n} \bigg)^{\frac{\alpha}{b}} \sqrt{\frac{j/n}{b}} \bigg( 1 + \frac{-b^{2}+6b \alpha - 6 \alpha^{2}}{12 b j/n} \frac{1}{n} + \bigO(n^{-2}) \bigg). \label{hj asymp 10}
\end{align}
\item As $n \to + \infty$ with $j_{2,+} \leq j \leq n$, we have
\begin{align}
h_{j}^{-1} & = \frac{b n^{\frac{j+\alpha}{b}}}{\Gamma (\frac{j+\alpha}{b})} \bigg( 1 + \bigO(n^{-\frac{1}{2}}e^{-nr_{2}^{2b}\frac{-\epsilon - \log(1-\epsilon)}{1-\epsilon}} \bigg) \nonumber \\
& = b \frac{\sqrt{n} \, e^{\frac{n}{b}(j/n-j/n \log(\frac{j/n}{b}))}}{\sqrt{2\pi}}\bigg( \frac{b}{j/n} \bigg)^{\frac{\alpha}{b}} \sqrt{\frac{j/n}{b}} \bigg( 1 + \frac{-b^{2}+6b \alpha - 6 \alpha^{2}}{12 b j/n} \frac{1}{n} + \bigO(n^{-2}) \bigg). \label{hj asymp 11}
\end{align}
\end{enumerate}
\end{lemma}
\begin{proof}
Note from \eqref{def of hj} and \eqref{def of aj lambdajl etajl} that
\begin{align}
h_{j}^{-1} & = \frac{b n^{\frac{j+\alpha}{b}}}{\Gamma(\tfrac{j+\alpha}{b})} \bigg( \frac{\gamma(\tfrac{j+\alpha}{b},nr_{1}^{2b})}{\Gamma(\tfrac{j+\alpha}{b})} - \frac{\gamma(\tfrac{j+\alpha}{b},nr_{2}^{2b})}{\Gamma(\tfrac{j+\alpha}{b})} + 1 \bigg)^{-1} \nonumber \\
& = \frac{b n^{\frac{j+\alpha}{b}}}{\Gamma(a_{j})} \bigg( \frac{\gamma(a_{j},\lambda_{j,1}a_{j})}{\Gamma(a_{j})} - \frac{\gamma(a_{j},\lambda_{j,2}a_{j})}{\Gamma(a_{j})} + 1 \bigg)^{-1}, \qquad j=1,2,\ldots,n. \label{lol12}
\end{align}
For fixed $j$, Lemma \ref{lemma:various regime of gamma} implies that $\gamma(\tfrac{j+\alpha}{b},nr_{k}^{2b})=\Gamma(\tfrac{j+\alpha}{b})+\bigO(n^{\tfrac{j+\alpha}{b}-1}e^{-nr_{k}^{2b}})$ as $n \to \infty$, and \eqref{hj asymp 1} follows. Now we turn to \eqref{hj asymp 2}. By Lemma \ref{lemma: asymp of gamma for lambda bounded away from 1} (i),
\begin{align*}
h_{j}^{-1} = \frac{b n^{\frac{j+\alpha}{b}}}{\Gamma (\frac{j+\alpha}{b})} \bigg( 1 + \bigO(n^{-\frac{1}{2}}e^{-\frac{a_{j}\eta_{j,1}^{2}}{2}}) + \bigO(n^{-\frac{1}{2}}e^{-\frac{a_{j}\eta_{j,2}^{2}}{2}}) \bigg) = \frac{b n^{\frac{j+\alpha}{b}}}{\Gamma (\frac{j+\alpha}{b})} \bigg( 1 + \bigO(n^{-\frac{1}{2}}e^{-\frac{a_{j}\eta_{j,1}^{2}}{2}})  \bigg)
\end{align*}
as $n \to \infty$, $M' \leq j \leq j_{1,-}$. Moreover,
\begin{align*}
\frac{a_{j}\eta_{j,1}^{2}}{2} = n \frac{j/n}{b}\bigg(\frac{br_{1}^{2b}}{j/n}-1-\log \Big(\frac{br_{1}^{2b}}{j/n}\Big)\bigg) + \bigO(1) \geq nr_{1}^{2b}\frac{\epsilon - \log(1+\epsilon)}{1+\epsilon} + \bigO(1)
\end{align*}
as $n \to \infty, \; M' \leq j \leq j_{1,-}$, which implies the first line in \eqref{hj asymp 2}. The second line in \eqref{hj asymp 2} then directly follows from Lemma \ref{lemma:asymp prefactor}. Now we prove \eqref{hj asymp 3}. By \eqref{lol12} and Lemmas \ref{lemma: uniform} and \ref{lemma: asymp of gamma for lambda bounded away from 1} (i),
\begin{align*}
h_{j}^{-1} = \frac{b n^{\frac{j+\alpha}{b}}}{\Gamma (\frac{j+\alpha}{b})} \bigg( \frac{1}{\frac{1}{2}\mathrm{erfc}(-\eta_{j,1}\sqrt{\frac{a_{j}}{2}}) } + \bigO(n^{-\frac{1}{2}}e^{-\frac{1}{2}a_{j}\eta_{j,1}^{2}}) \bigg)
\end{align*}
as $n \to + \infty$ with $j_{1,-} \leq j \leq g_{1,-}$. On the other hand,
\begin{align*}
-\eta_{j,1}\sqrt{\frac{a_{j}}{2}} \leq -\frac{Mr_{1}^{2}}{\sqrt{2}} + \bigO\Big(\frac{M^{2}}{\sqrt{n}}\Big), \qquad \mbox{as } n \to + \infty, \; j_{1,-} \leq j \leq g_{1,-},
\end{align*}
which implies the first line in \eqref{hj asymp 3} provided $M'$ is chosen large enough (recall that $M'\sqrt{\log n}\leq M$). The second line in \eqref{hj asymp 3} then directly follows from Lemma \ref{lemma:asymp prefactor}. Now we prove \eqref{hj asymp 4}. By Lemmas \ref{lemma: uniform} and \ref{lemma: asymp of gamma for lambda bounded away from 1} (i), as $n \to + \infty$ with $g_{1,-} \leq j \leq g_{1,+}$, we have
\begin{align*}
h_{j}^{-1} = \frac{b n^{\frac{j+\alpha}{b}}}{\Gamma (\frac{j+\alpha}{b})}  \frac{1 + \bigO(e^{-\frac{1}{2}\frac{a_{j}}{2}(\eta_{j,2}^{2}-\eta_{j,1}^{2})})}{\frac{1}{2}\mathrm{erfc}(-\eta_{j,1}\sqrt{\frac{a_{j}}{2}}) - \frac{e^{-\frac{1}{2}a_{j}\eta_{j,1}^{2}}}{\sqrt{2\pi a_{j}}}\big( c_{0}(\eta_{j,1}) + \frac{c_{1}(\eta_{j,1})}{a_{j}} + \bigO(\frac{1}{n^{2}}) \big)}  .
\end{align*}
The above can be expanded in terms of $M_{j,1}$ using \eqref{def of aj lambdajl etajl} and \eqref{def of Mjk}, namely using
\begin{align*}
\eta_{j,1} = (\lambda_{j,1}-1)\sqrt{\frac{2(\lambda_{j,1}-1-\log \lambda_{j,1})}{(\lambda_{j,1}-1)^{2}}}, \qquad \lambda_{j,1} = 1 + \frac{M_{j,1}}{\sqrt{n}}, \qquad j = \frac{br_{1}^{2b} n}{\lambda_{j,1}}-\alpha.
\end{align*}
Using also
\begin{align*}
\frac{e^{-\frac{r_{1}^{2b}M_{j,1}^{2}}{2}}}{\sqrt{\pi} \mathrm{erfc}(-\frac{M_{j,1}r_{1}^{b}}{\sqrt{2}}) } = \begin{cases}
- \frac{r_{1}^{b}M_{j,1}}{\sqrt{2}} - \frac{1}{\sqrt{2}r_{1}^{b}M_{j,1}} + \bigO(M_{j,1}^{-3}), & \mbox{as } M_{j,1}\to -\infty, \\
\bigO(e^{-\frac{r_{1}^{2b}M_{j,1}^{2}}{2}}), & \mbox{as } M_{j,1}\to +\infty,
\end{cases}
\end{align*}
we obtain
\begin{multline*}
\frac{1}{\frac{1}{2}\mathrm{erfc}(-\eta_{j,1}\sqrt{\frac{a_{j}}{2}}) - \frac{e^{-\frac{1}{2}a_{j}\eta_{j,1}^{2}}}{\sqrt{2\pi a_{j}}}\big( c_{0}(\eta_{j,1}) + \frac{c_{1}(\eta_{j,1})}{a_{j}} + \bigO(\frac{1}{n^{2}}) \big)} \\
= \frac{1}{\frac{1}{2}\mathrm{erfc}(-\frac{M_{j,1}r_{1}^{b}}{\sqrt{2}})} \bigg\{ 1 + \frac{(5M_{j,1}^{2}r_{1}^{2b}-2)e^{-\frac{r_{1}^{2b}M_{j,1}^{2}}{2}}}{3\sqrt{2\pi} r_{1}^{b} \mathrm{erfc}(-\frac{M_{j,1}r_{1}^{b}}{\sqrt{2}}) \sqrt{n}} + \frac{\chi_{j,1}^{-}}{n} \bigg( \frac{25r_{1}^{4b}M_{j,1}^{6}}{72} + \frac{3r_{1}^{2b}M_{j,1}^{4}}{2} \bigg) \nonumber \\
  - \chi_{j,1}^{-} \frac{125r_{1}^{6b}M_{j,1}^{9}}{1296n^{3/2}} + \bigO\bigg(\frac{1+M_{j,1}^{2}}{n} + \frac{1+|M_{j,1}^{7}|}{n^{3/2}} + \frac{1+M_{j,1}^{12}}{n^{2}} \bigg)  \bigg) \bigg\}
\end{multline*}
as $n \to + \infty$ uniformly for $g_{1,-} \leq j \leq g_{1,+}$, and the first expansion in \eqref{hj asymp 4} follows. The second expansion in \eqref{hj asymp 4} then follows from Lemma \ref{lemma:asymp prefactor}. Now we prove \eqref{hj asymp 5}. By \eqref{cond on epsilon 2}, \eqref{lol12}, Lemma \ref{lemma: uniform} and Lemma \ref{lemma: asymp of gamma for lambda bounded away from 1} (i), we have
\begin{align*}
h_{j}^{-1} = \frac{b n^{\frac{j+\alpha}{b}}}{\Gamma (\frac{j+\alpha}{b})} \frac{1}{\frac{1}{2}\mathrm{erfc}(-\eta_{j,1}\sqrt{a_{j}/2})-\frac{e^{-\frac{a_{j}\eta_{j,1}^{2}}{2}}}{\sqrt{2\pi a_{j}}}(c_{0}(\eta_{j,1})+\frac{c_{1}(\eta_{j,1})}{a_{j}})} \bigg( 1 + \bigO\big(n^{-2} + e^{-\frac{a_{j}}{2}(\eta_{j,2}^{2}-\eta_{j,1}^{2})}\big) \bigg),
\end{align*}
and the first expansion in \eqref{hj asymp 5} follows from a long but direct computation using \eqref{def of aj lambdajl etajl}, \eqref{def of c0 and c1} and \eqref{large y asymp of erfc}. The second expansion in \eqref{hj asymp 5} then follows from Lemma \ref{lemma:asymp prefactor}. Now we prove \eqref{hj asymp 6} and \eqref{hj asymp 7}. Using Lemma \ref{lemma: asymp of gamma for lambda bounded away from 1} (i)--(ii), we find
\begin{align*}
& h_{j}^{-1} = \frac{b n^{\frac{j+\alpha}{b}}}{\Gamma (\frac{j+\alpha}{b})} \Big( 1+\bigO(n^{-2}) \Big) \\
& \times  \frac{1}{\frac{e^{-\frac{a_{j}\eta_{j,1}^{2}}{2}}}{\sqrt{2\pi}}(\frac{1}{1-\lambda_{j,1}}\frac{1}{\sqrt{a_{j}}}+\frac{1+10 \lambda_{j,1}+\lambda_{j,1}^{2}}{12(\lambda_{j,1}-1)^{3}}\frac{1}{a_{j}^{3/2}}) - \frac{e^{-\frac{a_{j}\eta_{j,2}^{2}}{2}}}{\sqrt{2\pi}}(\frac{-1}{\lambda_{j,2}-1}\frac{1}{\sqrt{a_{j}}}+\frac{1+10 \lambda_{j,2}+\lambda_{j,2}^{2}}{12(\lambda_{j,2}-1)^{3}}\frac{1}{a_{j}^{3/2}})}
\end{align*}
as $n \to + \infty$ with $j_{1,+} \leq j \leq j_{2,-}$. Then \eqref{hj asymp 6} and \eqref{hj asymp 7} follow from a long but direct computation using \eqref{def of aj lambdajl etajl}, \eqref{asymp etajk-etajkm1 r1r1} and Lemma \ref{lemma:asymp prefactor}.

The proofs of \eqref{hj asymp 8}, \eqref{hj asymp 9}, \eqref{hj asymp 10} and \eqref{hj asymp 11} are similar to that of \eqref{hj asymp 5}, \eqref{hj asymp 4}, \eqref{hj asymp 3} and \eqref{hj asymp 2}, respectively, so we omit them.
\end{proof}

\section{Proof of Theorem \ref{thm:r1-r2 case}}\label{section:4}
In this section $z$ and $w$ are given by
\begin{align}\label{def of z1z2 r1 r2 case}
z = r_{1}\Big(1-\frac{t_{1}}{\sigma_{1}n}\Big)e^{i\theta_{1}}, \qquad w = r_{2}\Big(1+\frac{t_{2}}{\sigma_{2}n}\Big)e^{i\theta_{2}}, \qquad t_{1}, t_{2} \geq 0, \theta_{1},\theta_{2}\in \mathbb{R}.
\end{align}
Let us rewrite \eqref{def of Kn} as $K_{n}(z,w) = \sum_{j=1}^{n}F_{j}$, where
\begin{align}\label{def of Fj}
& F_{j} := e^{-\frac{n}{2}Q(z)}e^{-\frac{n}{2}Q(w)} \frac{z^{j-1}\overline{w}^{j-1}}{h_{j}} = e^{-\frac{n}{2}|z|^{2b}}e^{-\frac{n}{2}|w|^{2b}}|z|^{\alpha}|w|^{\alpha}\frac{z^{j-1}\overline{w}^{j-1}}{h_{j}}.
\end{align}
We will obtain the asymptotics of the summand $F_{j}$ as $n\to +\infty$ for several regimes of the indice $j$; such splitting can also be found in e.g. \cite{ForresterHoleProba, APS2009, SP2024} in the context of large gap probabilities.
In this section we take $M=M' \sqrt{\log n}$. The following lemma is proved using Lemma \ref{lemma: asymp of hjinv}. 
\begin{lemma}\label{lemma: asymp of Fj r1r2} Let $\delta\in (0,\frac{1}{100})$ be fixed. $M'$ can be chosen sufficiently large and independently of $\delta$ such that the following hold.
\begin{enumerate}
\item Let $j$ be fixed. As $n \to + \infty$, we have
\begin{align}\label{hj asymp 1 r1r2}
F_{j} = \bigO(e^{-\frac{n}{2}(1-\delta)(r_{1}^{2b}+r_{2}^{2b})}).
\end{align}
\item As $n \to + \infty$ with $M' \leq j \leq j_{1,-}$, we have
\begin{align}
F_{j} & = \bigO(e^{-n(1-\delta) \log(\frac{r_{2}}{r_{1}})\sigma_{1}}). \label{hj asymp 2 r1r2}
\end{align}
\item As $n \to + \infty$ with $j_{1,-} \leq j \leq g_{1,-}$, we have
\begin{align}
F_{j} & = \bigO(e^{-n(1-\delta) \log(\frac{r_{2}}{r_{1}})\sigma_{1}}). \label{hj asymp 3 r1r2}
\end{align}
\item As $n \to + \infty$ with $g_{1,-} \leq j \leq g_{1,+}$, we have
\begin{align}
& F_{j} = \bigO(e^{-n(1-\delta) \log(\frac{r_{2}}{r_{1}})\sigma_{1}}). \label{hj asymp 4 r1r2}
\end{align}
\item As $n \to + \infty$ with $g_{1,+} \leq j \leq j_{1,+}$, we have
\begin{align}
& F_{j} = \bigO(e^{-n(1-\delta) \log(\frac{r_{2}}{r_{1}})(\sigma_{\star}-\frac{br_{1}^{2b}}{1-\epsilon})}). \label{hj asymp 5 r1r2}
\end{align}
\item As $n \to + \infty$ with $j_{1,+} \leq j \leq \lfloor j_{\star}\rfloor$, we have
\begin{align}
& F_{j} = \frac{n}{r_{1}r_{2}}\frac{(\frac{r_{1}}{r_{2}})^{j_{\star}-j}e^{(j-1)i(\theta_{1}-\theta_{2})}}{\frac{1}{j/n-br_{1}^{2b}}+\frac{(\frac{r_{1}}{r_{2}})^{2(j_{\star}-j)}}{br_{2}^{2b}-j/n}}e^{-(j/n-br_{1}^{2b})\frac{t_{1}}{\sigma_{1}}-(br_{2}^{2b}-j/n)\frac{t_{2}}{\sigma_{2}}} \Big( 1 + \bigO\big( n^{-1} \big) \Big). \label{hj asymp 6 r1r2}
\end{align}
\item As $n \to + \infty$ with $\lfloor j_{\star}\rfloor+1 \leq j \leq j_{2,-}$, we have
\begin{align}
& F_{j} = \frac{n}{r_{1}r_{2}}\frac{(\frac{r_{1}}{r_{2}})^{j-j_{\star}}e^{(j-1)i(\theta_{1}-\theta_{2})}}{\frac{1}{br_{2}^{2b}-j/n}+\frac{(\frac{r_{1}}{r_{2}})^{2(j-j_{\star})}}{j/n-br_{1}^{2b}}}e^{-(j/n-br_{1}^{2b})\frac{t_{1}}{\sigma_{1}}-(br_{2}^{2b}-j/n)\frac{t_{2}}{\sigma_{2}}} \Big( 1 + \bigO\big( n^{-1} \big) \Big). \label{hj asymp 7 r1r2}
\end{align}
\item As $n \to + \infty$ with $j_{2,-} \leq j \leq g_{2,-}$, we have
\begin{align}
& F_{j} = \bigO(e^{-n(1-\delta) \log(\frac{r_{2}}{r_{1}})(\frac{br_{2}^{2b}}{1+\epsilon}-\sigma_{\star})}). \label{hj asymp 8 r1r2}
\end{align}
\item As $n \to + \infty$ with $g_{2,-} \leq j \leq g_{2,+}$, we have
\begin{align}
& F_{j} = \bigO(e^{-n(1-\delta) \log(\frac{r_{2}}{r_{1}})\sigma_{2}}). \label{hj asymp 9 r1r2}
\end{align}
\item As $n \to + \infty$ with $g_{2,+} \leq j \leq j_{2,+}$, we have
\begin{align}
F_{j} & = \bigO(e^{-n(1-\delta) \log(\frac{r_{2}}{r_{1}})\sigma_{2}}). \label{hj asymp 10 r1r2}
\end{align}
\item As $n \to + \infty$ with $j_{2,+} \leq j \leq n$, we have
\begin{align}
F_{j} & = \bigO(e^{-n(1-\delta) \log(\frac{r_{2}}{r_{1}})\sigma_{2}}). \label{hj asymp 11 r1r2}
\end{align}
\end{enumerate}
\end{lemma}
\begin{proof}
\eqref{hj asymp 1 r1r2} directly follows from \eqref{hj asymp 1} and \eqref{def of Fj}. Let us now prove \eqref{hj asymp 2 r1r2}. By \eqref{hj asymp 2} and \eqref{def of Fj}, as $n \to + \infty$ with $M' \leq j \leq j_{1,-}$, we have
\begin{align*}
F_{j} = \bigO\bigg(\exp \bigg\{ n(1-\delta) \bigg[ \frac{1}{b}(j/n-j/n \log(\frac{j/n}{b})) - \frac{1}{2}r_{1}^{2b} - \frac{1}{2}r_{2}^{2b} + j/n \log(r_{1}r_{2}) \bigg] \bigg\} \bigg).
\end{align*}
It is easy to check that the function $[0,br_{1}^{b}r_{2}^{b}] \in y \mapsto y-y \log(\frac{y}{b}) + y \log(r_{1}^{b}r_{2}^{b})$ is increasing. Since $j/n \leq br_{1}^{2b}$, we thus have
\begin{align*}
F_{j} & = \bigO\bigg(\exp \bigg\{ n(1-\delta) \bigg[ r_{1}^{2b}-r_{1}^{2b} \log(r_{1}^{2b}) - \frac{1}{2}r_{1}^{2b} - \frac{1}{2}r_{2}^{2b} + br_{1}^{2b} \log(r_{1}r_{2}) \bigg] \bigg\} \bigg),
\end{align*}
which is \eqref{hj asymp 2 r1r2} (recall the definitions of $\sigma_{1}$ and $\sigma_{\star}$ in \eqref{def of taustar}). The proof of \eqref{hj asymp 3 r1r2} is identical (but uses \eqref{hj asymp 3} instead of \eqref{hj asymp 2}). We now turn to the proof of \eqref{hj asymp 4 r1r2}. By \eqref{hj asymp 4} and \eqref{def of Fj}, as $n \to +\infty$ with $g_{1,-} \leq j \leq g_{1,+}$ we have
\begin{align*}
F_{j} = \bigO \bigg( \exp \bigg\{ n(1-\delta/2) \bigg[ r_{1}^{2b} - r_{1}^{2b}\log(r_{1}^{2b}) - \frac{1}{2}r_{1}^{2b} - \frac{1}{2}r_{2}^{2b} + j/n \log(r_{1}r_{2})  \bigg] \bigg\} \bigg).
\end{align*}
Since $j/n = br_{1}^{2b} + \bigO(n^{-\frac{1}{2}})$,
\begin{align*}
F_{j} = \bigO \bigg( \exp \bigg\{ n(1-\delta) \bigg[ r_{1}^{2b} - r_{1}^{2b}\log(r_{1}^{2b}) - \frac{1}{2}r_{1}^{2b} - \frac{1}{2}r_{2}^{2b} + br_{1}^{2b} \log(r_{1}r_{2})  \bigg] \bigg\} \bigg),
\end{align*}
and \eqref{hj asymp 4 r1r2} follows. Now we prove \eqref{hj asymp 5 r1r2}. By \eqref{hj asymp 5} and \eqref{def of Fj}, as $n \to + \infty$ with $g_{1,+} \leq j \leq j_{1,+}$, we have
\begin{align*}
F_{j} = \bigO \bigg( \exp \bigg\{ n(1-\delta) \bigg[ r_{1}^{2b} - 2j/n \log(r_{1}) - \frac{1}{2}r_{1}^{2b} - \frac{1}{2}r_{2}^{2b} + j/n \log(r_{1}r_{2})  \bigg] \bigg\} \bigg).
\end{align*}
Recall from \eqref{cond on epsilon 2} that $\sigma_{\star} > \frac{br_{1}^{2b}}{1-\epsilon}$. Since $j/n \leq \frac{br_{1}^{2b}}{1-\epsilon} + \bigO(n^{-1})$ by \eqref{def of jk plus and minus}, we have
\begin{align*}
F_{j} = \bigO \bigg( \exp \bigg\{ n(1-\delta) \bigg[ - \frac{r_{2}^{2b}-r_{1}^{2b}}{2} + \frac{br_{1}^{2b}}{1-\epsilon} \log(\frac{r_{2}}{r_{1}})  \bigg] \bigg\} \bigg),
\end{align*}
which is \eqref{hj asymp 5 r1r2}. Now we prove \eqref{hj asymp 6 r1r2}. As $n \to + \infty$ with $j_{1,+} \leq j \leq \lfloor j_{\star}\rfloor$, we have
\begin{align*}
& e^{n(r_{1}^{2b}-2 j/n \log(r_{1}))} e^{-\frac{n}{2}|z|^{2b}}e^{-\frac{n}{2}|w|^{2b}}|z|^{\alpha}|w|^{\alpha} z^{j-1}\overline{w}^{j-1} \\
& = e^{-\frac{n}{2}(r_{2}^{2b}-r_{1}^{2b}-2j/n \log(\frac{r_{2}}{r_{1}}))} e^{(j-1)i(\theta_{1}-\theta_{2})} e^{- (j/n-br_{1}^{2b})\frac{t_{1}}{\sigma_{1}}}e^{-(br_{2}^{2b}-j/n)\frac{t_{2}}{\sigma_{2}}}r_{1}^{\alpha-1}r_{2}^{\alpha-1} \big( 1+\bigO(n^{-1}) \big)
\end{align*}
Now \eqref{hj asymp 6 r1r2} follows directly from \eqref{hj asymp 6}, \eqref{def of Fj} and the identity (recall $j_{\star} = n \sigma_{\star} -\alpha$)
\begin{align*}
\frac{r_{2}^{\alpha}}{r_{1}^{\alpha}} e^{-\frac{n}{2}(r_{2}^{2b}-r_{1}^{2b}-2j/n \log(\frac{r_{2}}{r_{1}}))} = \Big(\frac{r_{2}}{r_{1}}\Big)^{-(j_{\star}-j)}.
\end{align*}

The proofs of \eqref{hj asymp 7 r1r2}, \eqref{hj asymp 8 r1r2}, \eqref{hj asymp 9 r1r2}, \eqref{hj asymp 10 r1r2} and \eqref{hj asymp 11 r1r2} are similar to that of \eqref{hj asymp 6 r1r2}, \eqref{hj asymp 5 r1r2}, \eqref{hj asymp 4 r1r2}, \eqref{hj asymp 3 r1r2} and \eqref{hj asymp 2 r1r2}, respectively, so we omit them.
\end{proof}
\begin{theorem}
Let $z,w,t_{1},t_{2},\theta_{1},\theta_{2}$ be as in \eqref{def of z1z2 r1 r2 case}. As $n \to + \infty$, we have
\begin{align*}
& K_{n}(z,w) = \frac{n}{r_{1}r_{2}} e^{-t_{1}-t_{2}} e^{i(\theta_{1}-\theta_{2})\lfloor j_{\star}\rfloor} \nonumber \\
& \times \Bigg\{ \sum_{j=0}^{+\infty} \frac{(\frac{r_{1}}{r_{2}})^{j+x}e^{(j+1)i(\theta_{2}-\theta_{1})}}{\frac{1}{\sigma_{1}}+\frac{(\frac{r_{1}}{r_{2}})^{2(j+x)}}{\sigma_{2}}} + \sum_{j=0}^{+\infty} \frac{(\frac{r_{1}}{r_{2}})^{j+1-x}e^{j i(\theta_{1}-\theta_{2})}}{\frac{1}{\sigma_{2}}+\frac{(\frac{r_{1}}{r_{2}})^{2(j+1-x)}}{\sigma_{1}}} \Bigg\} + \bigO((\log n)^{2}),
\end{align*}
where $\sigma_{\star}, \sigma_{1}, \sigma_{2}$ are given by \eqref{def of taustar}, $j_{\star} = n \sigma_{\star}-\alpha$, and $x=x(n)$ is given by \eqref{def of theta star intro}.
\end{theorem}
\begin{proof}
By Lemma \ref{lemma: asymp of Fj r1r2},
\begin{align*}
K_{n} = \sum_{j=j_{1,+}}^{j_{2,-}}F_{j} + \bigO(e^{-cn}), \qquad \mbox{as } n \to + \infty
\end{align*}
for some $c>0$. Furthermore, by \eqref{hj asymp 6 r1r2} and \eqref{hj asymp 7 r1r2},
\begin{align*}
& \sum_{j=j_{1,+}}^{\lfloor j_{\star} \rfloor }F_{j} = \sum_{j=j_{1,+}}^{\lfloor j_{\star} \rfloor } \frac{n}{ r_{1}r_{2}}\frac{(\frac{r_{1}}{r_{2}})^{j_{\star}-j}e^{(j-1)i(\theta_{1}-\theta_{2})}}{\frac{1}{j/n-br_{1}^{2b}}+\frac{(\frac{r_{1}}{r_{2}})^{2(j_{\star}-j)}}{br_{2}^{2b}-j/n}}e^{-(j/n-br_{1}^{2b})\frac{t_{1}}{\sigma_{1}}-(br_{2}^{2b}-j/n)\frac{t_{2}}{\sigma_{2}}} \Big( 1 + \bigO\big( n^{-1} \big) \Big), \\
& \sum_{j=\lfloor j_{\star} \rfloor +1}^{j_{2,-}}F_{j} = \sum_{j=\lfloor j_{\star} \rfloor +1}^{j_{2,-}} \frac{n}{ r_{1}r_{2}}\frac{(\frac{r_{1}}{r_{2}})^{j-j_{\star}}e^{(j-1)i(\theta_{1}-\theta_{2})}}{\frac{1}{br_{2}^{2b}-j/n}+\frac{(\frac{r_{1}}{r_{2}})^{2(j-j_{\star})}}{j/n-br_{1}^{2b}}}e^{-(j/n-br_{1}^{2b})\frac{t_{1}}{\sigma_{1}}-(br_{2}^{2b}-j/n)\frac{t_{2}}{\sigma_{2}}} \Big( 1 + \bigO\big( n^{-1} \big) \Big).
\end{align*}
For $j-\lfloor j_{\star} \rfloor \geq M' \log n$, we have $(\frac{r_{1}}{r_{2}})^{j-\lfloor j_{\star}}=\bigO(n^{-200})$ provided that $M'$ is chosen large enough. We infer that
\footnotesize
\begin{align*}
& \sum_{j=j_{1,+}}^{\lfloor j_{\star} \rfloor }F_{j} = \bigO(n^{-100}) + \hspace{-0.1cm} \sum_{j=\lfloor j_{\star} \rfloor - M' \log n}^{\lfloor j_{\star} \rfloor } \hspace{-0.2cm} \frac{n}{r_{1}r_{2}}\frac{(\frac{r_{1}}{r_{2}})^{j_{\star}-j}e^{(j-1)i(\theta_{1}-\theta_{2})}}{\frac{1}{j/n-br_{1}^{2b}}+\frac{(\frac{r_{1}}{r_{2}})^{2(j_{\star}-j)}}{br_{2}^{2b}-j/n}}e^{-(j/n-br_{1}^{2b})\frac{t_{1}}{\sigma_{1}}-(br_{2}^{2b}-j/n)\frac{t_{2}}{\sigma_{2}}} \Big( 1 + \bigO\big( n^{-1} \big) \Big), \\
& \sum_{j=\lfloor j_{\star} \rfloor +1}^{j_{2,-}}F_{j} = \bigO(n^{-100}) + \hspace{-0.1cm} \sum_{j=\lfloor j_{\star} \rfloor +1}^{\lfloor j_{\star} \rfloor +M' \log n} \hspace{-0.2cm} \frac{n}{r_{1}r_{2}}\frac{(\frac{r_{1}}{r_{2}})^{j-j_{\star}}e^{(j-1)i(\theta_{1}-\theta_{2})}}{\frac{1}{br_{2}^{2b}-j/n}+\frac{(\frac{r_{1}}{r_{2}})^{2(j-j_{\star})}}{j/n-br_{1}^{2b}}}e^{-(j/n-br_{1}^{2b})\frac{t_{1}}{\sigma_{1}}-(br_{2}^{2b}-j/n)\frac{t_{2}}{\sigma_{2}}} \Big( 1 + \bigO\big( n^{-1} \big) \Big).
\end{align*}
\normalsize
provided that $M'$ is chosen sufficiently large. For $|j-\lfloor j_{\star} \rfloor| \leq M' \log n$, we have $j/n = \sigma_{\star} + \bigO(\frac{\log n}{n})$, and thus
\begin{align*}
& \sum_{j=j_{1,+}}^{\lfloor j_{\star} \rfloor }F_{j} = \bigO(n^{-100}) + \frac{n}{r_{1}r_{2}} e^{- t_{1}- t_{2}} \hspace{-0.1cm} \sum_{j=\lfloor j_{\star} \rfloor - M' \log n}^{\lfloor j_{\star} \rfloor } \hspace{-0cm} \frac{(\frac{r_{1}}{r_{2}})^{j_{\star}-j}e^{(j-1)i(\theta_{1}-\theta_{2})}}{\frac{1}{\sigma_{1}}+\frac{(\frac{r_{1}}{r_{2}})^{2(j_{\star}-j)}}{\sigma_{2}}} \Big( 1 + \bigO\big( \frac{\log n}{n} \big) \Big), \\
& \sum_{j=\lfloor j_{\star} \rfloor +1}^{j_{2,-}}F_{j} = \bigO(n^{-100}) + \frac{n}{r_{1}r_{2}} e^{- t_{1}- t_{2}} \hspace{-0.1cm} \sum_{j=\lfloor j_{\star} \rfloor +1}^{\lfloor j_{\star} \rfloor +M' \log n} \hspace{-0cm} \frac{(\frac{r_{1}}{r_{2}})^{j-j_{\star}}e^{(j-1)i(\theta_{1}-\theta_{2})}}{\frac{1}{\sigma_{2}}+\frac{(\frac{r_{1}}{r_{2}})^{2(j-j_{\star})}}{\sigma_{1}}} \Big( 1 + \bigO\big( \frac{\log n}{n} \big) \Big).
\end{align*}
We can rewrite this as
\begin{align*}
& \sum_{j=j_{1,+}}^{\lfloor j_{\star} \rfloor }F_{j} = \bigO\big((\log n)^{2}\big) + \frac{n}{r_{1}r_{2}} e^{- t_{1}- t_{2}} \sum_{j=\lfloor j_{\star} \rfloor - M' \log n}^{\lfloor j_{\star} \rfloor }  \frac{(\frac{r_{1}}{r_{2}})^{j_{\star}-j}e^{(j-1)i(\theta_{1}-\theta_{2})}}{\frac{1}{\sigma_{1}}+\frac{(\frac{r_{1}}{r_{2}})^{2(j_{\star}-j)}}{\sigma_{2}}}, \\
& \sum_{j=\lfloor j_{\star} \rfloor +1}^{j_{2,-}}F_{j} = \bigO\big((\log n)^{2}\big) + \frac{n}{r_{1}r_{2}} e^{- t_{1}- t_{2}} \sum_{j=\lfloor j_{\star} \rfloor +1}^{\lfloor j_{\star} \rfloor +M' \log n} \frac{(\frac{r_{1}}{r_{2}})^{j-j_{\star}}e^{(j-1)i(\theta_{1}-\theta_{2})}}{\frac{1}{\sigma_{2}}+\frac{(\frac{r_{1}}{r_{2}})^{2(j-j_{\star})}}{\sigma_{1}}}.
\end{align*}
In the above sums appearing on the right-hand sides, we can replace $M' \log n$ by $+\infty$ at the cost of an error $\bigO(n^{-100})$; this error can in turn be absorbed in $\bigO\big((\log n)^{2}\big)$. We then find the claim after changing indices.
\end{proof}

A computation using \eqref{wsg} gives
\begin{multline*}
\frac{1}{2\pi r_{1}r_{2}}  \Bigg\{ \sum_{j=0}^{+\infty} \frac{(\frac{r_{1}}{r_{2}})^{j+x}e^{(j+1)i(\theta_{2}-\theta_{1})}}{\frac{1}{\sigma_{1}}+\frac{(\frac{r_{1}}{r_{2}})^{2(j+x)}}{\sigma_{2}}} + \sum_{j=0}^{+\infty} \frac{(\frac{r_{1}}{r_{2}})^{j+1-x}e^{j i(\theta_{1}-\theta_{2})}}{\frac{1}{\sigma_{2}}+\frac{(\frac{r_{1}}{r_{2}})^{2(j+1-x)}}{\sigma_{1}}} \Bigg\} \\
= (r_{1}r_{2})^{-x} \mathcal{S}^{G}_{\mathrm{hard}}(r_{1}e^{i\theta_{1}},r_{2}e^{i\theta_{2}};n) = (r_{1}r_{2})^{-x} \mathcal{S}^{G}_{\mathrm{hard}}(z,w;n) + \bigO(n^{-1}),
\end{multline*}
and Theorem \ref{thm:r1-r2 case} follows.

\section{Proofs of Theorems \ref{thm:r1 hard} and \ref{thm:r1r1 hard}}\label{section:r1r1 1/n}
In this section $z$ and $w$ are given by
\begin{align}\label{def of z1z2 r1 r1 case}
z = r_{1}\Big(1-\frac{t_{1}}{\sigma_{1}n}\Big)e^{i\theta_{1}}, \qquad w = r_{1}\Big(1-\frac{t_{2}}{\sigma_{1}n}\Big)e^{i\theta_{2}}, \qquad t_{1}, t_{2} \geq 0, \theta_{1},\theta_{2}\in \mathbb{R},
\end{align}
and the parameters $t_{1},t_{2},\theta_{1},\theta_{2}$ are independent of $n$. At this point, $\theta_{1},\theta_{2}\in \R$ can be arbitrary, but we already mention that in Section \ref{subsection:5.1} below, we will focus on the case $\theta_{1}=\theta_{2}$, and in Section \ref{subsection:5.2} we will consider the case $\theta_{1}\neq \theta_{2} \mod 2\pi$. As in \eqref{def of Fj} we write $K_{n}(z,w) = \sum_{j=1}^{n}F_{j}$, where
\begin{align}\label{def of Fj r1r1}
& F_{j} := e^{-\frac{n}{2}Q(z)}e^{-\frac{n}{2}Q(w)} \frac{z^{j-1}\overline{w}^{j-1}}{h_{j}} = e^{-\frac{n}{2}|z|^{2b}}e^{-\frac{n}{2}|w|^{2b}}|z|^{\alpha}|w|^{\alpha}\frac{z^{j-1}\overline{w}^{j-1}}{h_{j}}.
\end{align}
Here $M$ can be arbitrary within the range $M \in [M'\sqrt{\log n}, n^{\frac{1}{10}}]$ (the values of $M$ will be more restricted in Subsections \ref{subsection:5.1} and \ref{subsection:5.2} below). The following lemma is proved using Lemma \ref{lemma: asymp of hjinv}.
\begin{lemma}\label{lemma: asymp of Fj r1r1} Let $\delta \in (0,\frac{1}{100})$ be fixed. $M'$ can be chosen sufficiently large and independently of $\delta$ such that the following hold.
\begin{enumerate}
\item Let $j$ be fixed. As $n \to + \infty$, we have
\begin{align}\label{hj asymp 1 r1r1}
F_{j} = \bigO(e^{-n (1-\delta)r_{1}^{2b}}).
\end{align}
\item As $n \to + \infty$ with $M' \leq j \leq j_{1,-}$, we have
\begin{align}
F_{j} & = \bigO(e^{-nr_{1}^{2b}(1-\delta) \frac{\epsilon - \log(1+\epsilon)}{1+\epsilon}}). \label{hj asymp 2 r1r1}
\end{align}
\item As $n \to + \infty$ with $j_{1,-} \leq j \leq g_{1,-}$, we have
\begin{align}
F_{j} & = \bigO(n^{-100}). \label{hj asymp 3 r1r1}
\end{align}
\item As $n \to + \infty$ with $g_{1,-} \leq j \leq g_{1,+}$, we have
\begin{align}
& F_{j} = e^{i(j-1)(\theta_{1}-\theta_{2})} \sqrt{n}\frac{\sqrt{2} \, b r_{1}^{b-2} e^{-\frac{M_{j,1}^{2}r_{1}^{2b}}{2}}}{\sqrt{\pi} \mathrm{erfc}(-\frac{M_{j,1}r_{1}^{b}}{\sqrt{2}})} \bigg( 1 + \frac{1}{\sqrt{n}}\bigg\{\frac{(5M_{j,1}^{2}r_{1}^{2b}-2)e^{-\frac{r_{1}^{2b}M_{j,1}^{2}}{2}}}{3\sqrt{2\pi} r_{1}^{b} \mathrm{erfc}(-\frac{M_{j,1}r_{1}^{b}}{\sqrt{2}}) }  \nonumber \\
& \qquad + \frac{M_{j,1}}{6} \Big( 5M_{j,1}^{2}r_{1}^{2b}-3+6br_{1}^{2b}\frac{t_{1}+t_{2}}{\sigma_{1}}  \Big) \bigg\} \nonumber \\
& + \bigO\bigg(\frac{1+M_{j,1}^{2}+M_{j,1}^{6}\chi_{j,1}^{+}}{n}+\frac{1+M_{j,1}^{7}+M_{j,1}^{9}\chi_{j,1}^{+}}{n^{3/2}} + \frac{1+M_{j,1}^{12}}{n^{2}}\bigg)  \bigg). \label{hj asymp 4 r1r1}
\end{align}
\item As $n \to + \infty$ with $g_{1,+} \leq j \leq j_{1,+}$, we have
\begin{align}
& F_{j} = e^{i(j-1)(\theta_{1}-\theta_{2})} n \frac{e^{-(j/n-br_{1}^{2b})\frac{t_{1}+t_{2}}{\sigma_{1}}}(j/n-br_{1}^{2b})}{r_{1}^{2}}\bigg( 1 + \frac{1}{n} \bigg\{ \frac{t_{1}+t_{2}}{\sigma_{1}}(1-\alpha) - \frac{j/n}{2 \sigma_{1}^{2}}(t_{1}^{2}+t_{2}^{2}) \nonumber \\
& + \frac{b}{2\sigma_{1}^{2}}(1-2b)r_{1}^{2b}(t_{1}^{2}+t_{2}^{2}) + \frac{b^{2}r_{1}^{2b}+(j/n-br_{1}^{2b})\alpha}{(j/n-br_{1}^{2b})^{2}} \bigg\} - \frac{2 b^{4} r_{1}^{4b}}{n^{2}(j/n-br_{1}^{2b})^{4}} \nonumber \\
& + \frac{10b^{6}r_{1}^{6b}}{n^{3}(j/n-br_{1}^{2b})^{6}} + \bigO\Big( \frac{1}{n^{2}(j/n-br_{1}^{2b})^{3}} + \frac{1}{n^{4}(j/n-br_{1}^{2b})^{8}} \Big) \bigg). \label{hj asymp 5 r1r1}
\end{align}
\item As $n \to + \infty$ with $j_{1,+} \leq j \leq \lfloor j_{\star}\rfloor$, we have
\begin{align}
& F_{j} = e^{i(j-1)(\theta_{1}-\theta_{2})} \frac{n}{r_{1}^{2}} \frac{e^{-(j/n-br_{1}^{2b})\frac{t_{1}+t_{2}}{\sigma_{1}}}}{\frac{1}{j/n-br_{1}^{2b}}+\frac{(\frac{r_{1}}{r_{2}})^{2(j_{\star}-j)}}{br_{2}^{2b}-j/n}} \bigg( 1 + \frac{1}{n} \bigg\{ \frac{t_{1}+t_{2}}{\sigma_{1}}(1-\alpha) - \frac{j/n}{2\sigma_{1}^{2}}(t_{1}^{2}+t_{2}^{2}) \nonumber \\
& + \frac{b}{2\sigma_{1}^{2}}(1-2b)r_{1}^{2b}(t_{1}^{2}+t_{2}^{2}) + \frac{b^{2}r_{1}^{2b}+(j/n-br_{1}^{2b})\alpha}{(j/n-br_{1}^{2b})^{2}} \bigg\} + \bigO\Big( \frac{(\frac{r_{1}}{r_{2}})^{2(j_{\star}-j)}}{n} + \frac{1}{n^{2}} \Big) \bigg). \label{hj asymp 6 r1r1}
\end{align}
\item As $n \to + \infty$ with $\lfloor j_{\star}\rfloor+1 \leq j \leq j_{2,-}$, we have
\begin{align}
& F_{j} = e^{i(j-1)(\theta_{1}-\theta_{2})}\frac{n}{r_{1}^{2}} \frac{e^{-(j/n-br_{1}^{2b})\frac{t_{1}+t_{2}}{\sigma_{1}}}(\frac{r_{1}}{r_{2}})^{2(j-j_{\star})}}{\frac{1}{br_{2}^{2b} - j/n}+\frac{(\frac{r_{1}}{r_{2}})^{2(j-j_{\star})}}{j/n-br_{1}^{2b}}} \bigg( 1 + \frac{1}{n} \bigg\{ \frac{t_{1}+t_{2}}{\sigma_{1}}(1-\alpha) - \frac{j/n}{2\sigma_{1}^{2}}(t_{1}^{2}+t_{2}^{2}) \nonumber \\
& + \frac{b}{2\sigma_{1}^{2}}(1-2b)r_{1}^{2b}(t_{1}^{2}+t_{2}^{2}) + \frac{b^{2}r_{2}^{2b}+(j/n-br_{2}^{2b})\alpha}{(j/n-br_{2}^{2b})^{2}} \bigg\} + \bigO\Big( \frac{(\frac{r_{1}}{r_{2}})^{2(j-j_{\star})}}{n} + \frac{1}{n^{2}} \Big) \bigg). \label{hj asymp 7 r1r1}
\end{align}
\item As $n \to + \infty$ with $j_{2,-} \leq j \leq g_{2,-}$, we have
\begin{align}
& F_{j} = \bigO(e^{-n(1-\delta) \log(\frac{r_{2}}{r_{1}})(\frac{br_{2}^{2b}}{1+\epsilon}-\sigma_{\star})}). \label{hj asymp 8 r1r1}
\end{align}
\item As $n \to + \infty$ with $g_{2,-} \leq j \leq g_{2,+}$, we have
\begin{align}
& F_{j} = \bigO(e^{-n(1-\delta) \log(\frac{r_{2}}{r_{1}})\sigma_{2}}). \label{hj asymp 9 r1r1}
\end{align}
\item As $n \to + \infty$ with $g_{2,+} \leq j \leq j_{2,+}$, we have
\begin{align}
F_{j} & = \bigO(e^{-n(1-\delta) \log(\frac{r_{2}}{r_{1}})\sigma_{2}}). \label{hj asymp 10 r1r1}
\end{align}
\item As $n \to + \infty$ with $j_{2,+} \leq j \leq n$, we have
\begin{align}
F_{j} & = \bigO(e^{-n(1-\delta) \log(\frac{r_{2}}{r_{1}})\sigma_{2}}). \label{hj asymp 11 r1r1}
\end{align}
\end{enumerate}
\end{lemma}
\begin{proof}
\eqref{hj asymp 1 r1r1} follows from \eqref{hj asymp 1} and \eqref{def of Fj r1r1}. Let us now prove \eqref{hj asymp 2 r1r1}. By \eqref{hj asymp 2} and \eqref{def of Fj r1r1}, as $n \to + \infty$ with $M' \leq j \leq j_{1,-}$, we have
\begin{align*}
F_{j} = \bigO\bigg(\exp \bigg\{ n(1-\delta) \bigg[ \frac{1}{b}(j/n-j/n \log(\frac{j/n}{b})) - r_{1}^{2b} + 2j/n \log(r_{1}) \bigg] \bigg\} \bigg).
\end{align*}
It is easy to check that the function $y \mapsto y-y \log(\frac{y}{b}) + y \log(r_{1}^{2b})$ is increasing for $y\in [0,br_{1}^{2b}]$. Since $j/n \leq \frac{br_{1}^{2b}}{1+\epsilon}+\bigO(n^{-1})$, we thus have
\begin{align*}
F_{j} & = \bigO\bigg(\exp \bigg\{ -nr_{1}^{2b}(1-\delta) \frac{\epsilon - \log(1+\epsilon)}{1+\epsilon} \bigg\} \bigg),
\end{align*}
which is \eqref{hj asymp 2 r1r1}. Now we prove \eqref{hj asymp 3 r1r1}. By \eqref{hj asymp 3} and \eqref{def of Fj r1r1}, as $n \to + \infty$ with $j_{1,-} \leq j \leq g_{1,-}$, we have
\begin{align*}
F_{j} = \bigO\bigg( \exp \bigg\{ n(1-\delta) \bigg[ \frac{j/n - j/n \log(\frac{j/n}{b})}{b}-r_{1}^{2b} + 2j/n \log(r_{1}) \bigg] \bigg\} \bigg).
\end{align*}
Since $[0,br_{1}^{2b}] \in y \mapsto y-y \log(\frac{y}{b}) + y \log(r_{1}^{2b})$ is increasing and $j/n \leq \frac{b r_{1}^{2b}}{1+\frac{M}{\sqrt{n}}} + \bigO(n^{-1})$, we have
\begin{align*}
F_{j} = \bigO\bigg( \exp \bigg\{ -(1-\delta) \frac{M^{2}r_{1}^{2b}}{2} \bigg\} \bigg).
\end{align*}
Since $M\geq M' \sqrt{\log n}$, \eqref{hj asymp 3 r1r1} holds provided that $M'$ is chosen sufficiently large. Now we prove \eqref{hj asymp 4 r1r1}. As $n \to + \infty$ with $g_{1,-} \leq j \leq g_{1,+}$, we have
\begin{align*}
& e^{r_{1}^{2b}(1-\log(r_{1}^{2b}))n}e^{M_{j,1} r_{1}^{2b} \log(r_{1}^{2b})\sqrt{n}}e^{-r_{1}^{2b}(\frac{1}{2}+\log(r_{1}^{2b}))M_{j,1}^{2}} e^{-\frac{n}{2}|z|^{2b}}e^{-\frac{n}{2}|w|^{2b}}|z|^{\alpha}|w|^{\alpha}z^{j-1}\overline{w}^{j-1} \\
& = e^{i(j-1)(\theta_{1}-\theta_{2})}e^{-\frac{1}{2}M_{j,1}^{2}r_{1}^{2b}}r_{1}^{-2} \bigg\{ 1 - \frac{b M_{j,1}r_{1}^{2b}}{\sqrt{n}}\Big(2M_{j,1}^{2}\log r_{1}-\frac{t_{1}+t_{2}}{\sigma_{1}}\Big) \\
& + \frac{1}{n}\bigg( 2b^{2}r_{1}^{4b}(\log r_{1})^{2} M_{j,1}^{6} + 2br_{1}^{2b}\big( 1-br_{1}^{2b} \frac{t_{1}+t_{2}}{\sigma_{1}} \big) \log(r_{1})M_{j,1}^{4} \bigg) - \frac{r_{1}^{6b}(\log(r_{1}^{2b}))^{3}M_{j,1}^{9}}{6n^{3/2}} \\
& + \bigO\bigg(\frac{1+M_{j,1}^{2}}{n}+\frac{1+M_{j,1}^{7}}{n^{3/2}} + \frac{1+M_{j,1}^{12}}{n^{2}}\bigg) \bigg\}.
\end{align*}
Now \eqref{hj asymp 4 r1r1} follows directly from \eqref{hj asymp 4} and \eqref{def of Fj r1r1}. The proofs of \eqref{hj asymp 5 r1r1}, \eqref{hj asymp 6 r1r1} and \eqref{hj asymp 7 r1r1} follow in a similar way, using \eqref{hj asymp 5}, \eqref{hj asymp 6} and \eqref{hj asymp 7}, respectively.

Let us now prove \eqref{hj asymp 8 r1r1}. By \eqref{hj asymp 8} and \eqref{def of Fj r1r1}, as $n \to + \infty$ with $j_{2,-} \leq j \leq g_{2,-}$, we have
\begin{align*}
F_{j} = \bigO\bigg(\exp \bigg\{ n(1-\delta) \bigg[ r_{2}^{2b} - 2 j/n \log r_{2} - r_{1}^{2b} + 2j/n \log(r_{1}) \bigg] \bigg\} \bigg).
\end{align*}
Since $j/n = \frac{br_{2}^{2b}}{1+\epsilon}$, \eqref{hj asymp 8 r1r1} follows. The proof of \eqref{hj asymp 9 r1r1} is similar, and uses the fact that for $g_{2,-} \leq j \leq g_{2,+}$, we have $j/n = br_{2}^{2b} + \bigO(n^{-1/2})$. The proofs of \eqref{hj asymp 10 r1r1} and \eqref{hj asymp 11 r1r1} are also similar, and uses the fact that for $g_{2,+} \leq j \leq n$, we have $\frac{1}{b}(j/n - j/n \log(\frac{j/n}{b})) - r_{1}^{2b} + 2j/n \log(r_{1}) \leq -2 \log (\frac{r_{2}}{r_{1}})\sigma_{2}$.
\end{proof}

\begin{lemma}\label{lemma:r1r1 second lemma}
As $n \to + \infty$, we have
\begin{align*}
& K_{n}(z,w) = S_{4} + S_{5} + S_{6} + S_{7} + \bigO(n^{-100}),
\end{align*}
where
\begin{align*}
S_{4} := \sum_{j=g_{1,-}}^{g_{1,+}}F_{j}, \qquad S_{5} := \sum_{j=g_{1,+}+1}^{j_{1,+}}F_{j}, \qquad S_{6} := \sum_{j=j_{1,+}+1}^{\lfloor j_{\star}\rfloor}F_{j}, \qquad S_{7} := \sum_{j=\lfloor j_{\star}\rfloor+1}^{j_{2,-}} F_{j}.
\end{align*}
Furthermore, as $n \to + \infty$ we have
\begin{align}
& S_{4} = \sum_{j=g_{1,-}}^{g_{1,+}} e^{i(j-1)(\theta_{1}-\theta_{2})} \sqrt{n}\frac{\sqrt{2} \, b r_{1}^{b-2} e^{-\frac{M_{j,1}^{2}r_{1}^{2b}}{2}}}{\sqrt{\pi} \mathrm{erfc}(-\frac{M_{j,1}r_{1}^{b}}{\sqrt{2}})} \bigg( 1 + \frac{1}{\sqrt{n}}\bigg\{\frac{(5M_{j,1}^{2}r_{1}^{2b}-2)e^{-\frac{r_{1}^{2b}M_{j,1}^{2}}{2}}}{3\sqrt{2\pi} r_{1}^{b} \mathrm{erfc}(-\frac{M_{j,1}r_{1}^{b}}{\sqrt{2}}) }  \nonumber \\
& \qquad + \frac{M_{j,1}}{6} \Big( 5M_{j,1}^{2}r_{1}^{2b}-3+6br_{1}^{2b}\frac{t_{1}+t_{2}}{\sigma_{1}}  \Big) \bigg\} \nonumber \\
& \qquad + \bigO\bigg(\frac{1+M_{j,1}^{2}+M_{j,1}^{6}\chi_{j,1}^{+}}{n}+\frac{1+M_{j,1}^{7}+M_{j,1}^{9}\chi_{j,1}^{+}}{n^{3/2}} + \frac{1+M_{j,1}^{12}}{n^{2}}\bigg)  \bigg), \label{S4 asymp 1} \\
& S_{5} = \sum_{j=g_{1,+}+1}^{j_{1,+}} e^{i(j-1)(\theta_{1}-\theta_{2})} n \frac{e^{-(j/n-br_{1}^{2b})\frac{t_{1}+t_{2}}{\sigma_{1}}}(j/n-br_{1}^{2b})}{r_{1}^{2}}\bigg( 1 + \frac{1}{n} \bigg\{ \frac{t_{1}+t_{2}}{\sigma_{1}}(1-\alpha) - \frac{j/n}{2\sigma_{1}^{2}}(t_{1}^{2}+t_{2}^{2}) \nonumber \\
& \qquad + \frac{b}{2\sigma_{1}^{2}}(1-2b)r_{1}^{2b}(t_{1}^{2}+t_{2}^{2}) + \frac{b^{2}r_{1}^{2b}+(j/n-br_{1}^{2b})\alpha}{(j/n-br_{1}^{2b})^{2}} \bigg\} - \frac{2 b^{4} r_{1}^{4b}}{n^{2}(j/n-br_{1}^{2b})^{4}} \nonumber \\
& \qquad + \frac{10b^{6}r_{1}^{6b}}{n^{3}(j/n-br_{1}^{2b})^{6}} + \bigO\Big( \frac{1}{n^{2}(j/n-br_{1}^{2b})^{3}} + \frac{1}{n^{4}(j/n-br_{1}^{2b})^{8}} \Big) \bigg), \label{S5 asymp 1} \\
& S_{6} = \sum_{j=j_{1,+}+1}^{\lfloor j_{\star}\rfloor}e^{i(j-1)(\theta_{1}-\theta_{2})} \frac{n}{r_{1}^{2}} \frac{e^{-(j/n-br_{1}^{2b})\frac{t_{1}+t_{2}}{\sigma_{1}}}}{\frac{1}{j/n-br_{1}^{2b}}+\frac{(\frac{r_{1}}{r_{2}})^{2(j_{\star}-j)}}{br_{2}^{2b}-j/n}} \bigg( 1 + \frac{1}{n} \bigg\{ \frac{t_{1}+t_{2}}{\sigma_{1}}(1-\alpha) - \frac{j/n}{2\sigma_{1}^{2}}(t_{1}^{2}+t_{2}^{2}) \nonumber \\
& + \frac{b}{2 \sigma_{1}^{2}}(1-2b)r_{1}^{2b}(t_{1}^{2}+t_{2}^{2}) + \frac{b^{2}r_{1}^{2b}+(j/n-br_{1}^{2b})\alpha}{(j/n-br_{1}^{2b})^{2}} \bigg\} + \bigO\Big( \frac{(\frac{r_{1}}{r_{2}})^{2(j_{\star}-j)}}{n} + \frac{1}{n^{2}} \Big) \bigg), \label{S6 asymp 1} \\
& S_{7} = \sum_{j=\lfloor j_{\star}\rfloor+1}^{j_{2,-}} e^{i(j-1)(\theta_{1}-\theta_{2})}\frac{n}{r_{1}^{2}} \frac{e^{-(j/n-br_{1}^{2b})\frac{t_{1}+t_{2}}{\sigma_{1}}}(\frac{r_{1}}{r_{2}})^{2(j-j_{\star})}}{\frac{1}{br_{2}^{2b} - j/n}+\frac{(\frac{r_{1}}{r_{2}})^{2(j-j_{\star})}}{j/n-br_{1}^{2b}}} \bigg( 1 + \frac{1}{n} \bigg\{ \frac{t_{1}+t_{2}}{\sigma_{1}}(1-\alpha) - \frac{j/n}{2\sigma_{1}^{2}}(t_{1}^{2}+t_{2}^{2}) \nonumber \\
& + \frac{b}{2\sigma_{1}^{2}}(1-2b)r_{1}^{2b}(t_{1}^{2}+t_{2}^{2}) + \frac{b^{2}r_{2}^{2b}+(j/n-br_{2}^{2b})\alpha}{(j/n-br_{2}^{2b})^{2}} \bigg\} + \bigO\Big( \frac{(\frac{r_{1}}{r_{2}})^{2(j-j_{\star})}}{n} + \frac{1}{n^{2}} \Big) \bigg). \label{S7 asymp 1}
\end{align}
\end{lemma}
\begin{proof}
This is a straightforward consequence of Lemma \ref{lemma: asymp of Fj r1r1}.
\end{proof}

\begin{lemma}\label{lemma:asymp of S7}
As $n \to + \infty$,
\begin{align*}
S_{7} = \frac{n  e^{i \lfloor j_{\star} \rfloor (\theta_{1}-\theta_{2})} \sigma_{1}}{r_{1}^{2} e^{t_{1}+t_{2}}}\sum_{j=0}^{+\infty}  \frac{e^{i j(\theta_{1}-\theta_{2})}}{1 + \frac{\sigma_{1}}{\sigma_{2}}(\frac{r_{2}}{r_{1}})^{2(j+1-x)}} + \bigO\big( (\log n)^{2} \big),
\end{align*}
where $x=x(n)$ is given by \eqref{def of theta star intro}.
\end{lemma}
\begin{proof}
By \eqref{S7 asymp 1}, as $n \to + \infty$,
\begin{align*}
S_{7} = \sum_{j=\lfloor j_{\star}\rfloor+1}^{j_{2,-}} e^{i(j-1)(\theta_{1}-\theta_{2})}\frac{n}{r_{1}^{2}} \frac{e^{-(j/n-br_{1}^{2b})\frac{t_{1}+t_{2}}{\sigma_{1}}}(\frac{r_{1}}{r_{2}})^{2(j-j_{\star})}}{\frac{1}{br_{2}^{2b} - j/n}+\frac{(\frac{r_{1}}{r_{2}})^{2(j-j_{\star})}}{j/n-br_{1}^{2b}}} + \sum_{j=\lfloor j_{\star}\rfloor+1}^{j_{2,-}} \bigO\bigg( \Big(\frac{r_{1}}{r_{2}}\Big)^{2(j-j_{\star})} \bigg).
\end{align*}
The last sum is $\bigO(\log n)$. For the first sum, thanks to $(\frac{r_{1}}{r_{2}})^{2(j-j_{\star})}$ appearing in the numerator, the upper bound of summation $j_{2,-}$ can be replaced by $\lfloor j_{\star} \rfloor + M' \log n$ at the cost of an error $\bigO(\log n)$. We thus have
\begin{align*}
S_{7} = \sum_{\lfloor j_{\star}\rfloor+1}^{\lfloor j_{\star} \rfloor + M' \log n} e^{i(j-1)(\theta_{1}-\theta_{2})}\frac{n}{r_{1}^{2}} \frac{e^{-(j/n-br_{1}^{2b})\frac{t_{1}+t_{2}}{\sigma_{1}}}(\frac{r_{1}}{r_{2}})^{2(j-j_{\star})}}{\frac{1}{br_{2}^{2b} - j/n}+\frac{(\frac{r_{1}}{r_{2}})^{2(j-j_{\star})}}{j/n-br_{1}^{2b}}} + \bigO\big( \log n \big).
\end{align*}
For $|j-\lfloor j_{\star} \rfloor| \leq M' \log n$, we have $j/n = \sigma_{\star} + \bigO(\frac{\log n}{n})$, and thus
\begin{align*}
S_{7} = \frac{n e^{-t_{1}-t_{2}}}{r_{1}^{2}}\sum_{\lfloor j_{\star}\rfloor+1}^{\lfloor j_{\star} \rfloor + M' \log n} e^{i(j-1)(\theta_{1}-\theta_{2})} \frac{(\frac{r_{1}}{r_{2}})^{2(j-j_{\star})}}{\frac{1}{\sigma_{2}}+\frac{(\frac{r_{1}}{r_{2}})^{2(j-j_{\star})}}{\sigma_{1}}} + \bigO\big( (\log n)^{2} \big).
\end{align*}
In the above sum, we can replace $\lfloor j_{\star} \rfloor + M' \log n$ by $+\infty$ at the cost of an error $\bigO(\log n)$ (again thanks to $(\frac{r_{1}}{r_{2}})^{2(j-j_{\star})}$ appearing in the numerator). Then the claim follows from a simple shift of the indice of summation.
\end{proof}

\begin{lemma}\label{lemma:asymp of S6}
As $n \to + \infty$,
\begin{align*}
& S_{6} = -\frac{n e^{i \lfloor j_{\star} \rfloor (\theta_{1}-\theta_{2})}\sigma_{1}}{r_{1}^{2}e^{t_{1}+t_{2}} } \sum_{j=0}^{+\infty} \frac{e^{-i(j+1)(\theta_{1}-\theta_{2})}}{1+\frac{\sigma_{2}}{\sigma_{1}}(\frac{r_{2}}{r_{1}})^{2(j+x)}} \\
& + \sum_{j=j_{1,+}+1}^{\lfloor j_{\star}\rfloor}e^{i(j-1)(\theta_{1}-\theta_{2})} \frac{n (j/n-br_{1}^{2b})}{r_{1}^{2}e^{(j/n-br_{1}^{2b})\frac{t_{1}+t_{2}}{\sigma_{1}}}} \bigg( 1 + \frac{1}{n} \bigg\{ \frac{t_{1}+t_{2}}{\sigma_{1}}(1-\alpha) - \frac{j/n}{2\sigma_{1}^{2}}(t_{1}^{2}+t_{2}^{2}) \nonumber \\
& + \frac{b}{2\sigma_{1}^{2}}(1-2b)r_{1}^{2b}(t_{1}^{2}+t_{2}^{2}) + \frac{b^{2}r_{1}^{2b}+(j/n-br_{1}^{2b})\alpha}{(j/n-br_{1}^{2b})^{2}} \bigg\} + \bigO\Big( \frac{1}{n^{2}} \Big) \bigg) + \bigO\big( (\log n)^{2} \big).
\end{align*}
\end{lemma}
\begin{proof}
Let us split $S_{6} = S_{6}^{(1)}+S_{6}^{(2)}$, where
\begin{align*}
S_{6}^{(1)} = \sum_{j=j_{1,+}+1}^{\lfloor j_{\star} \rfloor - (M'+1) \log n} F_{j}, \qquad S_{6}^{(2)} = \sum_{j=\lfloor j_{\star} \rfloor - M' \log n}^{\lfloor j_{\star} \rfloor} F_{j}.
\end{align*}
Using \eqref{S6 asymp 1} and a similar analysis as in the proof of Lemma \ref{lemma:asymp of S7}, we obtain
\begin{align}
& S_{6}^{(2)} = \bigO\big( (\log n)^{2} \big) + \frac{n}{ r_{1}^{2}} e^{-t_{1}-t_{2}}\sum_{j=\lfloor j_{\star} \rfloor - M' \log n}^{\lfloor j_{\star} \rfloor}   \frac{e^{i(j-1)(\theta_{1}-\theta_{2})}}{\frac{1}{\sigma_{1}}} \nonumber \\
& + \frac{n}{r_{1}^{2}} e^{-t_{1}-t_{2}}\sum_{j=\lfloor j_{\star} \rfloor - M' \log n}^{\lfloor j_{\star} \rfloor}  \bigg\{ \frac{e^{i(j-1)(\theta_{1}-\theta_{2})}}{\frac{1}{\sigma_{1}}+\frac{(\frac{r_{1}}{r_{2}})^{2(j_{\star}-j)}}{\sigma_{2}}} - \frac{e^{i(j-1)(\theta_{1}-\theta_{2})}}{\frac{1}{\sigma_{1}}} \bigg\} \label{lol1}
\end{align}
as $n \to + \infty$. The last term in \eqref{lol1} is easily shown to be
\begin{align*}
-\frac{n e^{i \lfloor j_{\star} \rfloor (\theta_{1}-\theta_{2})}\sigma_{1}}{r_{1}^{2}e^{t_{1}+t_{2}} } \sum_{j=0}^{+\infty} \frac{e^{-i(j+1)(\theta_{1}-\theta_{2})}}{1+\frac{\sigma_{2}}{\sigma_{1}}(\frac{r_{2}}{r_{1}})^{2(j+x)}} + \bigO(\log n)
\end{align*}
as $n \to + \infty$. For the first sum in \eqref{lol1}, since $|j-\lfloor j_{\star} \rfloor| \leq M' \log n$, we can replace $\sigma_{1}$ by $j/n-br_{1}^{2b}$ at the cost of an error $\bigO((\log n)^{2})$ in the asymptotics of $S_{6}^{(2)}$; we thus find
\begin{align*}
& S_{6}^{(2)} = \bigO\big( (\log n)^{2} \big) + \frac{n}{ r_{1}^{2}} e^{-t_{1}-t_{2}} \hspace{-0.25cm} \sum_{j=\lfloor j_{\star} \rfloor - M' \log n}^{\lfloor j_{\star} \rfloor}   \frac{e^{i(j-1)(\theta_{1}-\theta_{2})}}{\frac{1}{j/n-br_{1}^{2b}}}  -\frac{n e^{i \lfloor j_{\star} \rfloor (\theta_{1}-\theta_{2})}\sigma_{1}}{r_{1}^{2}e^{t_{1}+t_{2}} } \sum_{j=0}^{+\infty} \frac{e^{-i(j+1)(\theta_{1}-\theta_{2})}}{1+\frac{\sigma_{2}}{\sigma_{1}}(\frac{r_{2}}{r_{1}})^{2(j+x)}} \\
& = \bigO\big( (\log n)^{2} \big) -\frac{n e^{i \lfloor j_{\star} \rfloor (\theta_{1}-\theta_{2})}\sigma_{1}}{r_{1}^{2}e^{t_{1}+t_{2}} } \sum_{j=0}^{+\infty} \frac{e^{-i(j+1)(\theta_{1}-\theta_{2})}}{1+\frac{\sigma_{2}}{\sigma_{1}}(\frac{r_{2}}{r_{1}})^{2(j+x)}} \nonumber \\
& + \hspace{-0.25cm} \sum_{j=\lfloor j_{\star} \rfloor - M' \log n}^{\lfloor j_{\star} \rfloor} e^{i(j-1)(\theta_{1}-\theta_{2})} \frac{n}{r_{1}^{2}} \frac{e^{-(j/n-br_{1}^{2b})\frac{t_{1}+t_{2}}{\sigma_{1}}}}{\frac{1}{j/n-br_{1}^{2b}}} \bigg( 1 + \frac{1}{n} \bigg\{ \frac{t_{1}+t_{2}}{\sigma_{1}}(1-\alpha) - \frac{j/n}{2\sigma_{1}^{2}}(t_{1}^{2}+t_{2}^{2}) \nonumber \\
& + \frac{b}{2 \sigma_{1}^{2}}(1-2b)r_{1}^{2b}(t_{1}^{2}+t_{2}^{2}) + \frac{b^{2}r_{1}^{2b}+(j/n-br_{1}^{2b})\alpha}{(j/n-br_{1}^{2b})^{2}} \bigg\} + \bigO\Big( \frac{(\frac{r_{1}}{r_{2}})^{2(j_{\star}-j)}}{n} + \frac{1}{n^{2}} \Big) \bigg)
\end{align*}
For $j_{1,+}+1 \leq j \leq \lfloor j_{\star} \rfloor - (M'+1) \log n$, we can replace $(\frac{r_{1}}{r_{2}})^{2(j_{\star}-j)}$ by $0$ in the asymptotics \eqref{hj asymp 6 r1r1} of $F_{j}$ at the cost of an error $\bigO(n^{-100})$ (provided $M'$ is chosen large enough). We thus have
\begin{align*}
& S_{6}^{(1)} = \sum_{j=j_{1,+}+1}^{\lfloor j_{\star} \rfloor - (M'+1) \log n} e^{i(j-1)(\theta_{1}-\theta_{2})} \frac{n}{r_{1}^{2}} \frac{e^{-(j/n-br_{1}^{2b})\frac{t_{1}+t_{2}}{\sigma_{1}}}}{\frac{1}{j/n-br_{1}^{2b}}} \bigg( 1 + \frac{1}{n} \bigg\{ \frac{t_{1}+t_{2}}{\sigma_{1}}(1-\alpha) - \frac{j/n}{2\sigma_{1}^{2}}(t_{1}^{2}+t_{2}^{2}) \nonumber \\
& + \frac{b}{2 \sigma_{1}^{2}}(1-2b)r_{1}^{2b}(t_{1}^{2}+t_{2}^{2}) + \frac{b^{2}r_{1}^{2b}+(j/n-br_{1}^{2b})\alpha}{(j/n-br_{1}^{2b})^{2}} \bigg\} + \bigO\Big( \frac{(\frac{r_{1}}{r_{2}})^{2(j_{\star}-j)}}{n} + \frac{1}{n^{2}} \Big) \bigg),
\end{align*}
The claim follows after combining the above two expansions.
\end{proof}
As a direct consequence of Lemmas \ref{lemma:r1r1 second lemma}, \ref{lemma:asymp of S7} and \ref{lemma:asymp of S6}, we obtain the following.
\begin{lemma}\label{lemma:S456 valid for all theta1 and theta2}
As $n \to + \infty$,
\begin{align*}
& S_{5} + S_{6} + S_{7} = S_{567} + \frac{n  e^{i \lfloor j_{\star} \rfloor (\theta_{1}-\theta_{2})} \sigma_{1}}{ r_{1}^{2} e^{t_{1}+t_{2}}} \mathcal{Q}_{n} + \bigO\big( (\log n)^{2} \big),
\end{align*}
where
\begin{align}
& S_{567} = \sum_{j=g_{1,+}+1}^{\lfloor j_{\star}\rfloor}e^{i(j-1)(\theta_{1}-\theta_{2})} \frac{n (j/n-br_{1}^{2b})}{r_{1}^{2}e^{(j/n-br_{1}^{2b})\frac{t_{1}+t_{2}}{\sigma_{1}}}} \bigg( 1 + \frac{1}{n} \bigg\{ \frac{t_{1}+t_{2}}{\sigma_{1}}(1-\alpha) - \frac{j/n}{2\sigma_{1}^{2}}(t_{1}^{2}+t_{2}^{2}) \nonumber \\
& + \frac{b}{2\sigma_{1}^{2}}(1-2b)r_{1}^{2b}(t_{1}^{2}+t_{2}^{2}) + \frac{b^{2}r_{1}^{2b}+(j/n-br_{1}^{2b})\alpha}{(j/n-br_{1}^{2b})^{2}} \bigg\} - \frac{2 b^{4} r_{1}^{4b}}{n^{2}(j/n-br_{1}^{2b})^{4}} \nonumber \\
& + \frac{10b^{6}r_{1}^{6b}}{n^{3}(j/n-br_{1}^{2b})^{6}} + \bigO\Big( \frac{1}{n^{2}(j/n-br_{1}^{2b})^{3}} + \frac{1}{n^{4}(j/n-br_{1}^{2b})^{8}} \Big) \bigg), \label{asymp S567} \\
& \mathcal{Q}_{n} := \sum_{j=0}^{+\infty}  \frac{e^{i j(\theta_{1}-\theta_{2})}}{1 + \frac{\sigma_{1}}{\sigma_{2}}(\frac{r_{2}}{r_{1}})^{2(j+1-x)}} - \sum_{j=0}^{+\infty} \frac{e^{-i(j+1)(\theta_{1}-\theta_{2})}}{1+\frac{\sigma_{2}}{\sigma_{1}}(\frac{r_{2}}{r_{1}})^{2(j+x)}}, \label{def Fn theta1 neq theta2}
\end{align}
and $x=x(n)$ is given by \eqref{def of theta star intro}.
\end{lemma}
\subsection{Proof of Theorem \ref{thm:r1 hard}}\label{subsection:5.1}
In this subsection, we take $\theta_{1}=\theta_{2}$ and $M=n^{\frac{1}{10}}$. \begin{lemma}\label{lemma:Qn in terms of theta}
For $\theta_{1}=\theta_{2}$, we have
\begin{align}\label{Fn theta}
\mathcal{Q}_{n} = \frac{(\log \theta)'(n \sigma_{\star}-\alpha +\frac{\log(\frac{\sigma_{2}}{\sigma_{1}}\frac{r_{2}}{r_{1}})}{2\log \frac{r_{2}}{r_{1}}}; \frac{\pi i}{\log \frac{r_{2}}{r_{1}}}) + (2x-1)\log \frac{r_{2}}{r_{1}} + \log \frac{\sigma_{2}}{\sigma_{1}}}{2\log \frac{r_{2}}{r_{1}}},
\end{align}
where $x=x(n)$ is given by \eqref{def of theta star intro}.
\end{lemma}
\begin{proof}
For $y \in \mathbb{R}$, $\rho \in (0,1)$ and $a>0$, define
\begin{align}
\Theta(y;\rho,a) & = y(y-1)\log (\rho) + y \log(a) \nonumber \\
& + \sum_{j=0}^{+\infty} \log \bigg( 1+a\, \rho^{2(j+y)} \bigg) + \sum_{j=0}^{+\infty} \log \bigg( 1+a^{-1} \rho^{2(j+1-y)} \bigg). \label{def of Theta}
\end{align}
It is proved in \cite[Lemma 3.25]{CharlierAnnuli} that
\begin{align*}
\Theta(y;\rho,a) = \frac{1}{2}\log \bigg( \frac{\pi a \rho^{-\frac{1}{2}}}{\log(\rho^{-1})} \bigg) + \frac{(\log a)^{2}}{4\log(\rho^{-1})} - \sum_{j=1}^{+\infty} \log(1-\rho^{2j}) + \log \theta \bigg( y + \frac{\log(a \rho)}{2\log(\rho)} \bigg| \frac{\pi i}{\log(\rho^{-1})} \bigg),
\end{align*}
where $\theta$ is the Jacobi theta function defined in \eqref{def of Jacobi theta}. By taking the derivative with respect to $y$ in the above identity, by replacing $(y,\rho,a)$ by $(x,\rho^{-1},a^{-1})$, we get
\begin{align}\label{lol13}
\sum_{j=0}^{+\infty} \frac{1}{1+a^{-1}\rho^{2(j+1-x)}} - \sum_{j=0}^{+\infty} \frac{1}{1+a\rho^{2(j+x)}} = \frac{(\log \theta)'(x+\frac{\log(a\rho)}{2\log \rho}; \frac{\pi i}{\log \rho}) + (2x-1)\log \rho + \log a}{2\log \rho}.
\end{align}
The claim follows from \eqref{lol13} with $\rho = \frac{r_{2}}{r_{1}}>1$ and $a=\frac{\sigma_{2}}{\sigma_{1}}$, and from the facts that $x = n \sigma_{\star}-\alpha - \lfloor n \sigma_{\star}-\alpha \rfloor$ and $\theta(z+1|\tau)=\theta(z|\tau)$.
\end{proof}
To expand $S_{567}$, we will need the following Riemann sum lemma.
\begin{lemma}(taken from \cite[Lemma 3.4]{CharlierAnnuli})\label{lemma:Riemann sum NEW}
Let $A,a_{0}$, $B,b_{0}$ be bounded function of $n \in \{1,2,\ldots\}$, such that
\begin{align*}
& a_{n} := An + a_{0} \qquad \mbox{ and } \qquad b_{n} := Bn + b_{0}
\end{align*}
are integers. Assume also that $B-A$ is positive and remains bounded away from $0$. Let $f$ be a function independent of $n$, and which is $C^{4}([\min\{\frac{a_{n}}{n},A\},\max\{\frac{b_{n}}{n},B\}])$ for all $n\in \{1,2,\ldots\}$. Then as $n \to + \infty$, we have
\begin{align}
&  \sum_{j=a_{n}}^{b_{n}}f(\tfrac{j}{n}) = n \int_{A}^{B}f(y)dy + \frac{(1-2a_{0})f(A)+(1+2b_{0})f(B)}{2}  \nonumber \\
& + \frac{(-1+6a_{0}-6a_{0}^{2})f'(A)+(1+6b_{0}+6b_{0}^{2})f'(B)}{12n}+ \frac{(-a_{0}+3a_{0}^{2}-2a_{0}^{3})f''(A)+(b_{0}+3b_{0}^{2}+2b_{0}^{3})f''(B)}{12n^{2}} \nonumber \\
& + \bigO \bigg( \frac{\mathfrak{m}_{A,n}(f''')+\mathfrak{m}_{B,n}(f''')}{n^{3}} + \sum_{j=a_{n}}^{b_{n}-1} \frac{\mathfrak{m}_{j,n}(f'''')}{n^{4}} \bigg), \label{sum f asymp gap NEW}
\end{align}
where, for a given function $g$ continuous on $[\min\{\frac{a_{n}}{n},A\},\max\{\frac{b_{n}}{n},B\}]$,
\begin{align*}
\mathfrak{m}_{A,n}(g) := \max_{y \in [\min\{\frac{a_{n}}{n},A\},\max\{\frac{a_{n}}{n},A\}]}|g(y)|, \quad \mathfrak{m}_{B,n}(g) := \max_{y \in [\min\{\frac{b_{n}}{n},B\},\max\{\frac{b_{n}}{n},B\}]}|g(y)|,
\end{align*}
and for $j \in \{a_{n},\ldots,b_{n}-1\}$, $\mathfrak{m}_{j,n}(g) := \max_{y \in [\frac{j}{n},\frac{j+1}{n}]}|g(y)|$.
\end{lemma}
It turns out that $S_{567}$ has oscillatory asymptotics as $n \to + \infty$. To handle these oscillations, we follow \cite{CharlierFH} and introduce the following quantities:
\begin{align*}
& \theta_{-}^{(n,M)} := g_{-} - \bigg( \frac{bn r^{2b}}{1+\frac{M}{\sqrt{n}}} - \alpha \bigg) = \bigg\lceil \frac{bn r^{2b}}{1+\frac{M}{\sqrt{n}}} - \alpha \bigg\rceil - \bigg( \frac{bn r^{2b}}{1+\frac{M}{\sqrt{n}}} - \alpha \bigg), \\
& \theta_{+}^{(n,M)} := \bigg( \frac{bn r^{2b}}{1-\frac{M}{\sqrt{n}}} - \alpha \bigg) - g_{+} = \bigg( \frac{bn r^{2b}}{1-\frac{M}{\sqrt{n}}} - \alpha \bigg) - \bigg\lfloor \frac{bn r^{2b}}{1-\frac{M}{\sqrt{n}}} - \alpha \bigg\rfloor.
\end{align*}
Note that $\theta_{-}^{(n,M)},\theta_{+}^{(n,M)} \in [0,1)$ are oscillatory but remain bounded as $n \to + \infty$. Recall that in this section, we take $\theta_{2} = \theta_{1}$.
\begin{lemma}\label{lemma:asymp S567 theta2=theta1}
If $(t_{1},t_{2})\neq (0,0)$, then as $n \to +\infty$,
\begin{align}
& S_{567} = D_{1} n^{2} + D_{2} \, n \log n + D_{3}^{(n,M)} n + D_{4}^{(n,M)}\sqrt{n} + \bigO\Big(M^{4} + \frac{n}{M^{6}} + \frac{\sqrt{n}}{M} \Big), \label{Dj for t1 t2 new 0}\\
& D_{1} := \sigma_{1}^{2}\frac{1-e^{-t_{1}-t_{2}}\big( 1+t_{1}+t_{2} \big)}{r_{1}^{2}(t_{1}+t_{2})^{2}}, \qquad D_{2} := \frac{b^{2}r_{1}^{2b-2}}{2}, \nonumber \\
& D_{3}^{(n,M)} := -b^{2} r_{1}^{2b-2}\Big( E_{1}\big( t_{1}+t_{2} \big) + \gamma_{\mathrm{E}} + \log(br_{1}^{2b}\frac{t_{1}+t_{2}}{\sigma_{1}}M) \Big) - \frac{b^{2}M^{2}r_{1}^{4b-2}}{2} \nonumber \\
& + \frac{(1-2x)\sigma_{1}}{2 r_{1}^{2}e^{t_{1}+t_{2}}} + \frac{\sigma_{1}}{r_{1}^{2}(t_{1}+t_{2})^{3}} \bigg\{ \frac{(t_{1}+t_{2})^{2}(t_{1}^{2}+t_{2}^{2})}{2e^{t_{1}+t_{2}}} \nonumber \\
& + \Big( 2t_{1}t_{2} - b^{2}r_{1}^{2b}\frac{t_{1}+t_{2}}{\sigma_{1}}(t_{1}^{2}+t_{2}^{2}) \Big) \bigg( 1 - \frac{1+t_{1}+t_{2}}{e^{t_{1}+t_{2}}} \bigg) \bigg\} - \frac{b^{2}}{M^{2} r_{1}^{2}} + \frac{5b^{2}}{2M^{4} r_{1}^{2+2b}}, \nonumber \\
& D_{4}^{(n,M)} := \frac{b^{2} r_{1}^{4b} \big( br_{1}^{2b}\frac{t_{1}+t_{2}}{\sigma_{1}}-3 \big)}{3 r_{1}^{2}}M^{3} + \frac{br_{1}^{2b} \big( 2\theta_{+}^{(n,M)} - 1 - 2b + 2b^{2}r_{1}^{2b}\frac{t_{1}+t_{2}}{\sigma_{1}} \big)}{2 r_{1}^{2}}M + \frac{(2\alpha+2\theta_{+}^{(n,M)}-1)b}{2 r_{1}^{2} M}. \nonumber
\end{align}
If $t_{1}=t_{2}=0$, then \eqref{Dj for t1 t2 new 0} holds with
\begin{align*}
& D_{1} = \frac{\sigma_{1}^{2}}{2 r_{1}^{2}}, \quad D_{2} = \frac{b^{2}r_{1}^{2b-2}}{2}, \\
& D_{3}^{(n,M)} = - \frac{b^{2}M^{2}r_{1}^{4b-2}}{2} + \frac{(\frac{1}{2}-x)\sigma_{1}}{ r_{1}^{2}} - b^{2}r_{1}^{2b-2}\log \frac{Mbr_{1}^{2b}}{\sigma_{1}} - \frac{b^{2}}{M^{2} r_{1}^{2}} + \frac{5b^{2}}{2M^{4} r_{1}^{2+2b}}, \\
& D_{4}^{(n,M)} = -b^{2} r_{1}^{4b-2} M^{3} + \frac{br_{1}^{2b} \big( 2\theta_{+}^{(n,M)} - 1 - 2b \big)}{2 r_{1}^{2}}M + \frac{(2\alpha+2\theta_{+}^{(n,M)}-1)b}{2 r_{1}^{2} M}.
\end{align*}
\end{lemma}
\begin{proof}
By \eqref{asymp S567} with $\theta_{1}=\theta_{2}$, as $n \to + \infty$ we have
\begin{align}
& S_{567} = n \sum_{j=g_{1,+}+1}^{\lfloor j_{\star}\rfloor} f_{1}(j/n) + \sum_{j=g_{1,+}+1}^{\lfloor j_{\star}\rfloor} f_{2}(j/n) + \frac{1}{n} \sum_{j=g_{1,+}+1}^{\lfloor j_{\star}\rfloor} \frac{e^{-(j/n-br_{1}^{2b})\frac{t_{1}+t_{2}}{\sigma_{1}}}}{r_{1}^{2}}\frac{-2 b^{4} r_{1}^{4b}}{(j/n-br_{1}^{2b})^{3}} \nonumber \\
& + \frac{1}{n^{2}} \sum_{j=g_{1,+}+1}^{\lfloor j_{\star}\rfloor} \frac{e^{-(j/n-br_{1}^{2b})\frac{t_{1}+t_{2}}{\sigma_{1}}}10b^{6}}{ r_{1}^{2-6b}(j/n-br_{1}^{2b})^{5}} + \bigO\Big( \frac{1}{n}\sum_{j=g_{1,+}+1}^{\lfloor j_{\star}\rfloor}\frac{1}{(j/n-br_{1}^{2b})^{2}} + \frac{1}{n^{3}}\sum_{j=g_{1,+}+1}^{\lfloor j_{\star}\rfloor} \frac{1}{(j/n-br_{1}^{2b})^{7}} \Big), \label{lol14}
\end{align}
where $f_{1}(x):= \frac{x-br_{1}^{2b}}{r_{1}^{2}e^{(x-br_{1}^{2b})\frac{t_{1}+t_{2}}{\sigma_{1}}}}$ and
\begin{align*}
f_{2}(x) := \frac{x-br_{1}^{2b}}{r_{1}^{2}e^{(x-br_{1}^{2b})\frac{t_{1}+t_{2}}{\sigma_{1}}}} \bigg( \frac{t_{1}+t_{2}}{\sigma_{1}}(1-\alpha) + \frac{b(1-2b)r_{1}^{2b}-x}{2\sigma_{1}^{2}}(t_{1}^{2}+t_{2}^{2}) + \frac{b^{2}r_{1}^{2b}+(x-br_{1}^{2b})\alpha}{(x-br_{1}^{2b})^{2}} \bigg).
\end{align*}
Since $g_{1,+} = \big\lfloor \frac{bnr_{1}^{2b}}{1-\frac{M}{\sqrt{n}}}-\alpha \big\rfloor$, for the error term in \eqref{lol14} we have the following estimate:
\begin{multline*}
\frac{1}{n}\sum_{j=g_{1,+}+1}^{\lfloor j_{\star}\rfloor}\frac{1}{(j/n-br_{1}^{2b})^{2}} + \frac{1}{n^{3}}\sum_{j=g_{1,+}+1}^{\lfloor j_{\star}\rfloor} \frac{1}{(j/n-br_{1}^{2b})^{7}} \\ = \bigO\bigg( \int_{M/\sqrt{n}}^{1} \frac{dx}{x^{2}} + \frac{1}{n^{2}}\int_{M/\sqrt{n}}^{1} \frac{dx}{x^{7}} \bigg) = \bigO \bigg( \frac{\sqrt{n}}{M} + \frac{n}{M^{6}} \bigg).
\end{multline*}
Since $f_{1}$ has no pole at $br_{1}^{2b}$ and $f_{2}$ has a pole of order $1$ at $br_{1}^{2b}$, long but direct calculations using Lemma \ref{lemma:Riemann sum NEW} with $A=\frac{br_{1}^{2b}}{1-\frac{M}{\sqrt{n}}}$, $a_{0}=1-\alpha-\theta_{+}^{(n,M)}$, $B=\sigma_{\star}$ and $b_{0}=-\alpha-x$ show that
\begin{align*}
& n \sum_{j=g_{1,+}+1}^{\lfloor j_{\star}\rfloor} f_{1}(j/n) = d_{1} n^{2} + (d_{2,1}M^{2}+d_{2,2})n + (d_{3,1}M^{3} + d_{3,2}M)\sqrt{n} + \bigO(M^{4}), \\
& \sum_{j=g_{1,+}+1}^{\lfloor j_{\star}\rfloor} f_{2}(j/n) = e_{1} n \log n + (e_{2,1} \log M + e_{2,2}) n + (e_{3,1}M + e_{3,2}M^{-1})\sqrt{n} + \bigO(M^{2}), \\
& \frac{1}{n} \sum_{j=g_{1,+}+1}^{\lfloor j_{\star}\rfloor} \frac{e^{-(j/n-br_{1}^{2b})\frac{t_{1}+t_{2}}{\sigma_{1}}}}{r_{1}^{2}}\frac{-2 b^{4} r_{1}^{4b}}{(j/n-br_{1}^{2b})^{3}} = - \frac{b^{2}n}{ r_{1}^{2} M^{2}} + \bigO\Big( \frac{\sqrt{n}}{M}  \Big), \\
& \frac{1}{n^{2}} \sum_{j=g_{1,+}+1}^{\lfloor j_{\star}\rfloor} \frac{e^{-(j/n-br_{1}^{2b})\frac{t_{1}+t_{2}}{\sigma_{1}}}10b^{6}}{r_{1}^{2-6b}(j/n-br_{1}^{2b})^{5}} = \frac{5b^{2}n}{2 r_{1}^{2+2b}M^{4}} + \bigO\Big( \frac{\sqrt{n}}{M^{3}} \Big),
\end{align*}
as $n \to + \infty$, for some explicit coefficients $d_{1},d_{2,1},d_{2,2},d_{3,1},d_{3,2},e_{1}, e_{2,1},e_{2,2},e_{3,1},e_{3,2}$ independent of $M$ that we do not write down. We then find the claim by combining the above two asymptotic formulas and simplifying (using primitives).
\end{proof}
To expand $S_{4}$, we will need the following lemma.
\begin{lemma}(Taken from \cite[Lemma 3.10]{CharlierAnnuli})\label{lemma:Riemann sum}
Let $h \in C^{3}(\mathbb{R})$ and $k =1$. As $n \to + \infty$, we have
\begin{align}
& \sum_{j=g_{k,-}}^{g_{k,+}} \hspace{-0.12cm} h(M_{j,k}) = br_{k}^{2b} \hspace{-0.07cm} \int_{-M}^{M} \hspace{-0.07cm} h(t) dt \; \sqrt{n} - 2 b r_{k}^{2b} \hspace{-0.07cm}\int_{-M}^{M} \hspace{-0.07cm} th(t) dt + \hspace{-0.07cm} \bigg( \frac{1}{2}-\theta_{k,-}^{(n,M)} \bigg)h(M)+ \hspace{-0.07cm} \bigg( \frac{1}{2}-\theta_{k,+}^{(n,M)} \bigg)h(-M) \nonumber \\
& + \frac{1}{\sqrt{n}}\bigg[ 3br_{k}^{2b} \hspace{-0.07cm} \int_{-M}^{M} \hspace{-0.1cm} t^{2}h(t)dt + \hspace{-0.07cm} \bigg( \hspace{-0.04cm}\frac{1}{12}+\frac{\theta_{k,-}^{(n,M)}(\theta_{k,-}^{(n,M)}\hspace{-0.04cm}-\hspace{-0.04cm}1)}{2} \bigg)\frac{h'(M)}{br_{k}^{2b}} - \bigg( \hspace{-0.04cm}\frac{1}{12}+\frac{\theta_{k,+}^{(n,M)}(\theta_{k,+}^{(n,M)}-1)}{2} \bigg)\frac{h'(-M)}{br_{k}^{2b}} \bigg] \nonumber \\
& + \bigO\Bigg(  \frac{1}{n^{3/2}} \hspace{-0.07cm} \sum_{j=g_{k,-}+1}^{g_{k,+}} \bigg( \hspace{-0.07cm} (1+|M_{j,1}|^{3}) \tilde{\mathfrak{m}}_{j,n}(h) + (1+M_{j,1}^{2})\tilde{\mathfrak{m}}_{j,n}(h') + (1+|M_{j,1}|) \tilde{\mathfrak{m}}_{j,n}(h'') + \tilde{\mathfrak{m}}_{j,n}(h''') \bigg) \hspace{-0.07cm} \Bigg), \label{sum f asymp 2}
\end{align}
where, for $\tilde{h} \in C(\mathbb{R})$ and $j \in \{g_{k,-}+1,\ldots,g_{k,+}\}$, we define $\tilde{\mathfrak{m}}_{j,n}(\tilde{h}) := \max_{x \in [M_{j,k},M_{j-1,k}]}|\tilde{h}(x)|$.
\end{lemma}

\begin{lemma}\label{lemma:asymp S4 theta2=theta1}
As $n \to + \infty$,
\begin{align*}
& S_{4} = E_{3}^{(M)} n + E_{4}^{(M)} \sqrt{n} + \bigO\Big(M^{4} + \frac{M^{9}}{\sqrt{n}} + \frac{M^{14}}{n^{3/2}} \Big), \\
& E_{3}^{(M)} := br_{1}^{2b} \int_{-M}^{M} h_{1}(t)dt, \\
& E_{4}^{(M)} := br_{1}^{2b} \int_{-M}^{M} (h_{2}(t)-2th_{1}(t))dt + \bigg( \frac{1}{2}-\theta_{-}^{(n,M)} \bigg)h_{1}(M) + \bigg( \frac{1}{2} - \theta_{+}^{(n,M)} \bigg)h_{1}(-M),
\end{align*}
where
\begin{align*}
& h_{1}(y) := \frac{\sqrt{2}\, b r_{1}^{b-2}e^{-\frac{r_{1}^{2b}}{2}y^{2}}}{\sqrt{\pi} \mathrm{erfc}(-\frac{r_{1}^{b}y}{\sqrt{2}})}, \\
& h_{2}(y) := \frac{\sqrt{2}\, b r_{1}^{b-2}e^{-\frac{r_{1}^{2b}}{2}y^{2}}}{\sqrt{\pi} \mathrm{erfc}(-\frac{r_{1}^{b}y}{\sqrt{2}})} \Bigg( \frac{y}{6}\Big( 5y^{2}r_{1}^{2b}-3+6br_{1}^{2b}(t_{1}+t_{2}) \Big) + \frac{e^{-\frac{r_{1}^{2b}}{2}y^{2}}(5y^{2}r_{1}^{2b}-2)}{3\sqrt{2\pi} r_{1}^{b} \mathrm{erfc} (-\frac{r_{1}^{b}y}{\sqrt{2}})} \Bigg).
\end{align*}
\end{lemma}
\begin{proof}
By \eqref{S4 asymp 1}, as $n \to + \infty$ we have
\begin{align*}
S_{4} = \sqrt{n} \sum_{j=g_{1,-}}^{g_{1,+}}h_{1}(M_{j,1}) + \sum_{j=g_{1,-}}^{g_{1,+}}h_{2}(M_{j,1}) + \bigO \bigg( \sum_{j=g_{1,-}}^{g_{1,+}} \Big( \frac{1+M_{j,1}^{3}}{\sqrt{n}} + \frac{1+M_{j,1}^{8}}{n} + \frac{1+M_{j,1}^{13}}{n^{3/2}} \Big) \bigg).
\end{align*}
Since $\sum_{j=g_{1,-}}^{g_{1,+}} 1 = \bigO(M\sqrt{n})$, the above error term can be rewritten as $\bigO(M^{4} + \frac{M^{9}}{\sqrt{n}} + \frac{M^{14}}{n^{3/2}})$. We then use Lemma \ref{lemma:Riemann sum} to expand the two sums involving $h_{1}$ and $h_{2}$, and the claim follows.
\end{proof}

Recall that $\mathcal{I}$ is define by \eqref{def of I}. Define also
\begin{align}
& \mathcal{I}_1 = \int_{-\infty}^{+\infty} \bigg\{ \frac{e^{-y^{2}}}{\sqrt{\pi}\, \mathrm{erfc}(y)} - \chi_{(0,+\infty)}(y) \bigg[ y + \frac{y}{2(1+y^{2})} \bigg] \bigg\}dy, \label{def of I1}
	\\
& \mathcal{I}_{2} = \int_{-\infty}^{+\infty} \bigg\{ \frac{y^{3}e^{-y^{2}}}{\sqrt{\pi} \, \mathrm{erfc}(y)} - \chi_{(0,+\infty)}(y) \bigg[ y^{4}+\frac{y^{2}}{2}-\frac{1}{2} \bigg] \bigg\}dy, \label{def of I3} \\
& \mathcal{I}_{3} = \int_{-\infty}^{+\infty} \bigg\{ \bigg( \frac{e^{-y^{2}}}{\sqrt{\pi} \, \mathrm{erfc}(y)} \bigg)^{2} - \chi_{(0,+\infty)}(y) \bigg[ y^{2}+1 \bigg] \bigg\} dy, \label{def of I4}
	\\
& \mathcal{I}_{4} = \int_{-\infty}^{+\infty} \bigg\{ \bigg( \frac{y \, e^{-y^{2}}}{\sqrt{\pi} \, \mathrm{erfc}(y)} \bigg)^{2} - \chi_{(0,+\infty)}(y)\bigg[ y^{4}+y^{2}-\frac{3}{4} \bigg] \bigg\} dy. \label{def of I5}
\end{align}

\begin{lemma}\label{lemma: S4+S567}
If $(t_{1},t_{2}) \neq (0,0)$, then as $n \to + \infty$,
\begin{align}
S_{4} + S_{567} & = C_{1} n^{2} + C_{2} \, n \log n + \bigg( \tilde{C}_{3}+\frac{(1-2x)\sigma_{1}}{2\pi r_{1}^{2}e^{t_{1}+t_{2}}} \bigg) n \nonumber \\
& + \tilde{C}_{4} \sqrt{n} + \bigO\Big(M^{4} + \frac{M^{9}}{\sqrt{n}} + \frac{M^{14}}{n^{3/2}}+\frac{n}{M^{6}} + \frac{\sqrt{n}}{M}\Big), \label{asymp S4 + S567 final}
\end{align}
where $C_{1},C_{2}$ are as in the statement of Theorem \ref{thm:r1 hard} and $\tilde{C}_{3},\tilde{C}_{4}$ are given by
\begin{align*}
& \tilde{C}_{3} = 2b^{2}r_{1}^{2b-2} \mathcal{I}_{1} - \frac{b^{2}r_{1}^{2b-2}}{2}\log \Big( \frac{2}{r_{1}^{2b}} \Big)  - b^{2} r_{1}^{2b-2}\Big( E_{1}( t_{1}+t_{2}) + \gamma_{\mathrm{E}} + \log(br_{1}^{2b}\frac{t_{1}+t_{2}}{\sigma_{1}}) \Big) \\
& + \frac{\sigma_{1}}{r_{1}^{2}(t_{1}+t_{2})^{3}} \bigg\{ \frac{(t_{1}+t_{2})^{2}(t_{1}^{2}+t_{2}^{2})}{2e^{t_{1}+t_{2}}} \\
& \hspace{2.7cm} + \Big( 2t_{1}t_{2} - \frac{b^{2}r_{1}^{2b}}{\sigma_{1}}(t_{1}+t_{2})(t_{1}^{2}+t_{2}^{2})\Big) \bigg( 1 - \frac{1+t_{1}+t_{2}}{e^{t_{1}+t_{2}}} \bigg) \bigg\}, \\
& \tilde{C}_{4} = \frac{\sqrt{2} b^{2} r_{1}^{b-2}}{3} \bigg\{  \big[15-6br_{1}^{2b}\frac{t_{1}+t_{2}}{\sigma_{1}}\big]\mathcal{I} - 10 \, \mathcal{I}_{2} - 2 \, \mathcal{I}_{3} + 10 \, \mathcal{I}_{4} \bigg\}.
\end{align*}
If $t_{1}=t_{2}=0$, then \eqref{asymp S4 + S567 final} holds with
\begin{align*}
& C_{1} = \frac{\sigma_{1}^{2}}{2 r_{1}^{2}}, \qquad C_{2} = \frac{b^{2}r_{1}^{2b-2}}{2}, \qquad \tilde{C}_{3} = 2b^{2}r_{1}^{2b-2} \mathcal{I}_{1} + \frac{b^{2}r_{1}^{2b-2}}{2} \log \Big( \frac{\sigma_{1}^{2}}{2b^{2}r_{1}^{2b}} \Big), \\
& \tilde{C}_{4} = \frac{\sqrt{2} b^{2} r_{1}^{b-2}}{3} \bigg\{ 15\, \mathcal{I} - 10 \, \mathcal{I}_{2} - 2 \, \mathcal{I}_{3} + 10 \, \mathcal{I}_{4} \bigg\}.
\end{align*}
\end{lemma}
\begin{proof}
Combining Lemmas \ref{lemma:asymp S567 theta2=theta1} and \ref{lemma:asymp S4 theta2=theta1}, we infer that as $n \to + \infty$,
\begin{align*}
S_{4} + S_{567} = D_{1} n^{2} + D_{2} \, n \log n + (D_{3}^{(n,M)}+E_{3}^{(M)}) n + (D_{4}^{(M)}+E_{4}^{(M)})\sqrt{n} + \bigO\Big(M^{4}\Big).
\end{align*}
On the other hand, as $n \to + \infty$,
\begin{align*}
& (D_{3}^{(n,M)}+E_{3}^{(M)})n = \bigg( \tilde{C}_{3}+\frac{(1-2x)\sigma_{1}}{2 r_{1}^{2}e^{t_{1}+t_{2}}} \bigg)n + \bigO\Big(\frac{n}{M^{2}}\Big), \\
& (D_{4}^{(M)}+E_{4}^{(M)})\sqrt{n} = \tilde{C}_{4}\sqrt{n} + \bigO\Big(\frac{\sqrt{n}}{M}\Big),
\end{align*}
and the claim follows.
\end{proof}

Since $M=n^{\frac{1}{10}}$, we have $M^{4} = \frac{M^{9}}{\sqrt{n}} = \frac{M^{14}}{n^{3/2}} = \frac{n}{M^{6}} =  \frac{\sqrt{n}}{M} = n^{\frac{2}{5}}$, and thus the $\bigO$-term in \eqref{asymp S4 + S567 final} can be written as $\bigO(n^{\frac{2}{5}})$. Hence, combining Lemmas \ref{lemma:r1r1 second lemma}, \ref{lemma:S456 valid for all theta1 and theta2} and \ref{lemma: S4+S567} yields
\begin{align*}
K_{n}(z,w) = C_{1} n^{2} + C_{2} \, n \log n + \bigg( \tilde{C}_{3}+\frac{(\frac{1}{2}-x+\mathcal{Q}_{n})\sigma_{1}}{r_{1}^{2}e^{t_{1}+t_{2}}} \bigg) n  + \tilde{C}_{4} \sqrt{n} + \bigO(n^{\frac{2}{5}}).
\end{align*}
Furthermore, by \eqref{Fn theta 2 intro} and Lemma \ref{lemma:Qn in terms of theta}, we have $\frac{1}{2}-x+\mathcal{Q}_{n}=\mathcal{F}_{n}$, where $\mathcal{F}_{n}$ is given by \eqref{Fn theta 2 intro}. Finally, using the following identities (proved in \cite[Lemma 2.10]{ACCL2022 1})
\begin{align}\label{simplif of Ij}
& \mathcal{I}_{1} = \frac{\log (2\sqrt{\pi})}{2}, & &  \mathcal{I}_{3} = \mathcal{I}, & & \mathcal{I}_{4} = \mathcal{I}_{2}-\mathcal{I},
\end{align}
we obtain that $\tilde{C}_{3}=C_{3}$ and $\tilde{C}_{4}=C_{4}$, where $C_{3}$ and $C_{4}$ are as in the statement of Theorem \ref{thm:r1 hard}. This finishes the proof of Theorem \ref{thm:r1 hard}.

\subsection{Proof of Theorem \ref{thm:r1r1 hard}}\label{subsection:5.2}
In this subsection we assume that $\theta_{1}\neq \theta_{2} \mod 2\pi$, and we take $M=M' \sqrt{\log n}$. By \eqref{asymp S567} we have
\begin{align*}
& S_{567} = \frac{n}{r_{1}^{2}} \sum_{j=g_{1,+}+1}^{\lfloor j_{\star}\rfloor} \hspace{-0.15cm} e^{i(j-1)(\theta_{1}-\theta_{2})} \frac{j/n-br_{1}^{2b}}{e^{(j/n-br_{1}^{2b})\frac{t_{1}+t_{2}}{\sigma_{1}}}} \bigg( \hspace{-0.07cm} 1+ \bigO\bigg( \frac{1}{n(j/n-br_{1}^{2b})^{2}} \bigg) \hspace{-0.07cm} \bigg) \hspace{-0.07cm} = \frac{n e^{i \lfloor j_{\star}\rfloor (\theta_{1}-\theta_{2})}}{r_{1}^{2}} \sum_{\ell = 1}^{N} c_{\ell}d_{\ell}
\end{align*}
where $N := \lfloor j_{\star}\rfloor -g_{1,+}$, $d_{\ell} := p^{\ell}$, $p:=e^{-i (\theta_{1}-\theta_{2})}$, and $c_{\ell}$ satisfies
\begin{align}
& c_{\ell} = \frac{-\ell/n-br_{1}^{2b} + \frac{\lfloor j_{\star}\rfloor +1}{n}}{e^{(-\ell/n-br_{1}^{2b} + \frac{\lfloor j_{\star}\rfloor +1}{n})\frac{t_{1}+t_{2}}{\sigma_{1}}}}\bigg( 1 + \frac{1}{n} \bigg\{ \frac{t_{1}+t_{2}}{\sigma_{1}}(1-\alpha) - \frac{-\ell/n + \frac{\lfloor j_{\star}\rfloor +1}{n}}{2}\frac{t_{1}^{2}+t_{2}^{2}}{\sigma_{1}^{2}}  \nonumber \\
& + \frac{b}{2\sigma_{1}^{2}}(1-2b)r_{1}^{2b}(t_{1}^{2}+t_{2}^{2}) + \frac{b^{2}r_{1}^{2b}+(-\ell/n-br_{1}^{2b} + \frac{\lfloor j_{\star}\rfloor +1}{n})\alpha}{(-\ell/n-br_{1}^{2b} + \frac{\lfloor j_{\star}\rfloor +1}{n})^{2}} \bigg\} - \frac{2 b^{4} r_{1}^{4b}}{n^{2}(-\ell/n-br_{1}^{2b} + \frac{\lfloor j_{\star}\rfloor +1}{n})^{4}} \nonumber \\
& + \frac{10b^{6}r_{1}^{6b}}{n^{3}(-\ell/n-br_{1}^{2b} + \frac{\lfloor j_{\star}\rfloor +1}{n})^{6}} + \bigO\Big( \frac{1}{n^{2}(\ell/n-\sigma_{1})^{3}} + \frac{1}{n^{4}(\ell/n-\sigma_{1})^{8}} \Big) \bigg) \label{asymp c ell}
\end{align}
as $n\to+\infty$. Using summation by parts, we get
\begin{align}\label{lol6}
S_{567} = \frac{n e^{i \lfloor j_{\star}\rfloor (\theta_{1}-\theta_{2})}}{r_{1}^{2}} \bigg( c_{N}D_{N}-c_{1}D_{0}- \frac{c_{N}-c_{1}}{1-p} + \sum_{\ell=1}^{N-1} (c_{\ell+1}-c_{\ell}) \frac{p^{\ell+1}}{1-p} \bigg),
\end{align}
where $D_{\ell} := \sum_{j=0}^{\ell}d_{j}$.
Note that
\begin{align}
& \frac{n e^{i \lfloor j_{\star}\rfloor (\theta_{1}-\theta_{2})}}{r_{1}^{2}} \bigg( -c_{1}D_{0}- \frac{-c_{1}}{1-p} \bigg) = \frac{n e^{i \lfloor j_{\star}\rfloor (\theta_{1}-\theta_{2})}}{r_{1}^{2}} \frac{p}{1-p} \frac{\sigma_{1}}{e^{t_{1}+t_{2}}}\big( 1+ \bigO (n^{-1}) \big), \label{lol4} \\
& \frac{n e^{i \lfloor j_{\star}\rfloor (\theta_{1}-\theta_{2})}}{r_{1}^{2}} \bigg( c_{N}D_{N} - \frac{c_{N}}{1-p} \bigg) = \bigO(M\sqrt{n}), \label{lol5}
\end{align}
where for \eqref{lol5} we have used that $-N/n-br_{1}^{2b} + \frac{\lfloor j_{\star}\rfloor +1}{n} = \bigO(M/\sqrt{n})$. Let us now prove that
\begin{align}\label{lol3}
\bigg| \frac{n e^{i \lfloor j_{\star}\rfloor (\theta_{1}-\theta_{2})}}{r_{1}^{2}}  \sum_{\ell=1}^{N-1} (c_{\ell+1}-c_{\ell}) \frac{p^{\ell+1}}{1-p} \bigg| = \bigO\bigg( \frac{\sqrt{n}}{M} \bigg).
\end{align}
From \eqref{asymp c ell} we deduce that
\begin{align}\label{lol15}
c_{\ell+1}-c_{\ell} = c_{\ell} \bigg( \frac{1+(\ell/n-\sigma_{1})\frac{t_{1}+t_{2}}{\sigma_{1}}}{(\ell/n-\sigma_{1})n} + \bigO\bigg( \frac{1}{n^{2}(\ell/n-\sigma_{1})^{3}} \bigg) \bigg), \qquad \mbox{as } n \to + \infty.
\end{align}
Using \eqref{lol15} and then \eqref{asymp c ell}, we get
\begin{align}
& \sum_{\ell=1}^{N-1} (c_{\ell+1}-c_{\ell}) \frac{p^{\ell+1}}{1-p} \nonumber \\
&  = \frac{1}{n}\frac{p}{1-p} \sum_{\ell=1}^{N-1} c_{\ell} \frac{1+(\ell/n-\sigma_{1})\frac{t_{1}+t_{2}}{\sigma_{1}}}{\ell/n-\sigma_{1}} p^{\ell} + \sum_{\ell=1}^{N-1} \bigO\bigg( \frac{1}{n^{2}(\ell/n-\sigma_{1})^{2}} \bigg) \nonumber \\
& = \frac{1}{n}\frac{p}{1-p} \sum_{\ell=1}^{N-1} \tilde{c}_{\ell} p^{\ell} + \bigO\bigg( \frac{1}{M\sqrt{n}} \bigg), \qquad \tilde{c}_{\ell} := \frac{1+(\ell/n-\sigma_{1})\frac{t_{1}+t_{2}}{\sigma_{1}}}{e^{(\ell/n-\sigma_{1})\frac{t_{1}+t_{2}}{\sigma_{1}}}}. \label{lol2}
\end{align}
Summing again by parts, we get
\begin{align*}
\sum_{\ell=1}^{N-1} \tilde{c}_{\ell} p^{\ell} = \tilde{c}_{N-1}D_{N-1}-\tilde{c}_{1}D_{0}- \frac{\tilde{c}_{N-1}-c_{1}}{1-p} + \sum_{\ell=1}^{N-2} (\tilde{c}_{\ell+1}-\tilde{c}_{\ell}) \frac{p^{\ell+1}}{1-p} = \bigO(1),
\end{align*}
where we have used that $\tilde{c}_{\ell+1}-\tilde{c}_{\ell}=\bigO(n^{-1})$ uniformly for $\ell=1,\ldots,N-2$. Substituting the above in \eqref{lol2}, this proves \eqref{lol3}. Substituting now \eqref{lol4}, \eqref{lol5} and \eqref{lol3} in \eqref{lol6}, we get
\begin{align}\label{lol7}
S_{567} = \frac{n e^{i \lfloor j_{\star}\rfloor (\theta_{1}-\theta_{2})}}{r_{1}^{2}} \frac{p}{1-p} \frac{\sigma_{1}}{e^{t_{1}+t_{2}}} + \bigO\big( M \sqrt{n} \big).
\end{align}
Now we analyze $S_{4}$. By \eqref{S4 asymp 1}, we have
\begin{align*}
& S_{4} = \sqrt{n} \sum_{j=g_{1,-}}^{g_{1,+}} p^{-(j-1)} \hat{c}_{j} = p^{2-g_{1,-}}\sqrt{n} \sum_{\ell=1}^{\hat{N}} p^{-\ell} \hat{c}_{\ell}', \qquad \mbox{as } n \to + \infty,
\end{align*}
where $\hat{N}:=g_{1,+}-g_{1,-}+1$ and $\hat{c}_{j}, \hat{c}_{\ell}'$ satisfy
\begin{align*}
& \hat{c}_{j} = \frac{\sqrt{2} \, b r_{1}^{b-2} e^{-\frac{M_{j,1}^{2}r_{1}^{2b}}{2}}}{\sqrt{\pi} \mathrm{erfc}(-\frac{M_{j,1}r_{1}^{b}}{\sqrt{2}})} \bigg( 1 + \frac{1}{\sqrt{n}}\bigg\{\frac{(5M_{j,1}^{2}r_{1}^{2b}-2)e^{-\frac{r_{1}^{2b}M_{j,1}^{2}}{2}}}{3\sqrt{2\pi} r_{1}^{b} \mathrm{erfc}(-\frac{M_{j,1}r_{1}^{b}}{\sqrt{2}}) }  \nonumber \\
& \qquad + \frac{M_{j,1}}{6} \Big( 5M_{j,1}^{2}r_{1}^{2b}-3+6br_{1}^{2b}(t_{1}+t_{2})  \Big) \bigg\} \nonumber \\
& \qquad + \bigO\bigg(\frac{1+M_{j,1}^{2}+M_{j,1}^{6}\chi_{j,1}^{+}}{n}+\frac{1+M_{j,1}^{7}+M_{j,1}^{9}\chi_{j,1}^{+}}{n^{3/2}} + \frac{1+M_{j,1}^{12}}{n^{2}}\bigg)  \bigg), \\
& \hat{c}_{\ell}' := \hat{c}_{j}, \quad \mbox{where } j = \ell+g_{1,-}-1.
\end{align*}
Summing by parts, we get
\begin{align*}
\sum_{\ell=1}^{\hat{N}} p^{-\ell} \hat{c}_{\ell}' = \hat{c}_{\hat{N}}'\hat{D}_{\hat{N}}-\hat{c}_{1}'\hat{D}_{0}- \frac{\hat{c}'_{\hat{N}}-\hat{c}'_{1}}{1-p^{-1}} + \sum_{\ell=1}^{\hat{N}-1} (\hat{c}'_{\ell+1}-\hat{c}'_{\ell}) \frac{p^{-\ell-1}}{1-p^{-1}}
\end{align*}
where $\hat{D}_{\ell} := \sum_{j=0}^{\ell}p^{-j}$. Since $\hat{c}_{1}' = \bigO(e^{-M^{2}})$, $\hat{c}'_{\hat{N}} = \bigO(1)$, $\hat{c}'_{\ell+1}-\hat{c}'_{\ell} = \bigO(n^{-\frac{1}{2}})$ and $\hat{N} = \bigO(M\sqrt{n})$, we get $\sum_{\ell=1}^{\hat{N}} p^{-\ell} \hat{c}_{\ell}' = \bigO(M)$, and thus
\begin{align}\label{lol9}
S_{4} = \bigO(M\sqrt{n}), \qquad \mbox{as } n \to + \infty.
\end{align}
Combining \eqref{lol7}, \eqref{lol9} and Lemmas \ref{lemma:r1r1 second lemma} and \ref{lemma:S456 valid for all theta1 and theta2} yields
\begin{align}\label{lol16}
K_{n}(z,w) = \frac{n  e^{i \lfloor j_{\star} \rfloor (\theta_{1}-\theta_{2})} \sigma_{1}}{r_{1}^{2} e^{t_{1}+t_{2}}} \bigg( \frac{1}{p^{-1}-1} + \mathcal{Q}_{n} \bigg) +  \bigO(M\sqrt{n}).
\end{align}
Using \eqref{Abel series} and \eqref{def Fn theta1 neq theta2}, we infer that
\begin{align}\label{lol17}
\frac{\sigma_{1}}{r_{1}^{2}} \bigg( \frac{1}{p^{-1}-1} + \mathcal{Q}_{n} \bigg) = 2\pi \cdot r_{1}^{-2x} \cdot \mathcal{S}^{G}_{\mathrm{hard}}(r_{1}e^{\theta_{1}},r_{1}e^{\theta_{2}};n).
\end{align}
Substituting \eqref{lol17} in \eqref{lol16} yields
\begin{align*}
K_{n}(z,w) = 2\pi n \cdot \mathcal{S}^{G}_{\mathrm{hard}}(r_{1}e^{\theta_{1}},r_{1}e^{\theta_{2}};n) \cdot \frac{  e^{i \lfloor j_{\star} \rfloor (\theta_{1}-\theta_{2})}}{e^{t_{1}+t_{2}}} r_{1}^{-2x} +  \bigO(M\sqrt{n}).
\end{align*}
Since $M=M' \sqrt{\log n}$, this finishes the proof of Theorem \ref{thm:r1r1 hard}.

\section{Proofs of Theorems \ref{thm:r1 semi-hard} and \ref{thm:r1r1 semi-hard}}\label{section:r1r1 1/sqrtn}

In this section $z$ and $w$ are given by
\begin{align}\label{def of z1z2 r1 r1 case o}
z = r_{1}\Big(1-\frac{s_{1}}{\sqrt{n}}\Big)e^{i\theta_{1}}, \qquad w = r_{1}\Big(1-\frac{s_{2}}{\sqrt{n}}\Big)e^{i\theta_{2}}, \qquad s_{1}, s_{2} > 0, \theta_{1},\theta_{2}\in \mathbb{R},
\end{align}
and $s_{1}, s_{2}, \theta_{1},\theta_{2}$ are independent of $n$. At this point $\theta_{1},\theta_{2}\in \R$ can be arbitrary (in Subsection \ref{subsec:6.1} we will take $\theta_{2}=\theta_{1}$, and in Subsection \ref{subsec:6.2} we will take $\theta_{2}\neq \theta_{1} \mod 2\pi$). We emphasize that $s_{1}, s_{2}$ are strictly positive (the case $s_{1}=s_{2}=0$ is very different and covered in Section \ref{section:r1r1 1/n}). The parameters $s_{1},s_{2}$ in \eqref{def of z1z2 r1 r1 case o} are related to $\mathfrak{s}_{1}, \mathfrak{s}_{2}$ in \eqref{def of z1z2 r1 r1 case o thm intro} via $s_{j} = \mathfrak{s}_{j}/(\sqrt{2}br_{1}^{b})$, $j=1,2$. As in \eqref{def of Fj} we write $K_{n}(z,w) = \sum_{j=1}^{n}F_{j}$, where
\begin{align}\label{def of Fj r1r1 o}
& F_{j} := e^{-\frac{n}{2}Q(z)}e^{-\frac{n}{2}Q(w)} \frac{z^{j-1}\overline{w}^{j-1}}{h_{j}} = e^{-\frac{n}{2}|z|^{2b}}e^{-\frac{n}{2}|w|^{2b}}|z|^{\alpha}|w|^{\alpha}\frac{z^{j-1}\overline{w}^{j-1}}{h_{j}}.
\end{align}
In this section we take $M=M' \log n$. The following lemma is proved using Lemma \ref{lemma: asymp of hjinv}.
\begin{lemma}\label{lemma: asymp of Fj r1r1 o} Let $\delta \in (0,\frac{1}{100})$ be fixed. $M'$ can be chosen sufficiently large and independently of $\delta$ such that the following hold.
\begin{enumerate}
\item Let $j$ be fixed. As $n \to + \infty$, we have
\begin{align}\label{hj asymp 1 r1r1 o}
F_{j} = \bigO(e^{-n (1-\delta)r_{1}^{2b}}).
\end{align}
\item As $n \to + \infty$ with $M' \leq j \leq j_{1,-}$, we have
\begin{align}
F_{j} & = \bigO(e^{-nr_{1}^{2b}(1-\delta) \frac{\epsilon - \log(1+\epsilon)}{1+\epsilon} }). \label{hj asymp 2 r1r1 o}
\end{align}
\item As $n \to + \infty$ with $j_{1,-} \leq j \leq g_{1,-}$, we have
\begin{align}
F_{j} & = \bigO(n^{-100}). \label{hj asymp 3 r1r1 o}
\end{align}
\item As $n \to + \infty$ with $g_{1,-} \leq j \leq g_{1,+}$, we have
\begin{align}
& F_{j} = e^{i(j-1)(\theta_{1}-\theta_{2})} \sqrt{n}\frac{\sqrt{2} \, b r_{1}^{b-2} e^{-\frac{1}{2}r_{1}^{2b}(M_{j,1}^{2}-2bM_{j,1}(s_{1}+s_{2})+2b^{2}(s_{1}^{2}+s_{2}^{2}))}}{\sqrt{\pi} \mathrm{erfc}(-\frac{M_{j,1}r_{1}^{b}}{\sqrt{2}})}   \nonumber \\
& \qquad \times \bigg( 1 + \frac{1}{\sqrt{n}}\bigg\{\frac{(5M_{j,1}^{2}r_{1}^{2b}-2)e^{-\frac{r_{1}^{2b}M_{j,1}^{2}}{2}}}{3\sqrt{2\pi} r_{1}^{b} \mathrm{erfc}(-\frac{M_{j,1}r_{1}^{b}}{\sqrt{2}}) } - \frac{M_{j,1}}{2}+s_{1}+s_{2} + r_{1}^{2b} \bigg[ \frac{5}{6}M_{j,1}^{3} \nonumber \\
& \qquad - bM_{j,1}^{2}(s_{1}+s_{2}) + bM_{j,1} \frac{s_{1}^{2}+s_{2}^{2}}{2} + \frac{b^{2}(2b-3)}{3}(s_{1}^{3}+s_{2}^{3}) \bigg] \bigg\} + \bigO\Big(\frac{1+M_{j,1}^{6}}{n} \Big)  \bigg). \label{hj asymp 4 r1r1 o}
\end{align}
\item As $n \to + \infty$ with $g_{1,+} \leq j \leq j_{1,+}$, we have
\begin{align}
& F_{j} = \bigO(n^{-100}). \label{hj asymp 5 r1r1 o}
\end{align}
\item As $n \to + \infty$ with $j_{1,+} \leq j \leq \lfloor j_{\star}\rfloor$, we have
\begin{align}
& F_{j} = \bigO\Big( e^{-(1-\delta)\frac{\epsilon br_{1}^{2b}}{1-\epsilon}(s_{1}+s_{2})\sqrt{n}} \Big). \label{hj asymp 6 r1r1 o}
\end{align}
\item As $n \to + \infty$ with $\lfloor j_{\star}\rfloor+1 \leq j \leq j_{2,-}$, we have
\begin{align}
& F_{j} = \bigO\Big( e^{-(1-\delta)\sigma_{1}(s_{1}+s_{2})\sqrt{n}} \Big). \label{hj asymp 7 r1r1 o}
\end{align}
\item As $n \to + \infty$ with $j_{2,-} \leq j \leq g_{2,-}$, we have
\begin{align}
& F_{j} = \bigO(e^{-n(1-\delta) \log(\frac{r_{2}}{r_{1}})(\frac{br_{2}^{2b}}{1+\epsilon}-\sigma_{\star})}). \label{hj asymp 8 r1r1 o}
\end{align}
\item As $n \to + \infty$ with $g_{2,-} \leq j \leq g_{2,+}$, we have
\begin{align}
& F_{j} = \bigO(e^{-n(1-\delta) \log(\frac{r_{2}}{r_{1}})\sigma_{2}}). \label{hj asymp 9 r1r1 o}
\end{align}
\item As $n \to + \infty$ with $g_{2,+} \leq j \leq j_{2,+}$, we have
\begin{align}
F_{j} & = \bigO(e^{-n(1-\delta) \log(\frac{r_{2}}{r_{1}})\sigma_{2}}). \label{hj asymp 10 r1r1 o}
\end{align}
\item As $n \to + \infty$ with $j_{2,+} \leq j \leq n$, we have
\begin{align}
F_{j} & = \bigO(e^{-n(1-\delta) \log(\frac{r_{2}}{r_{1}})\sigma_{2}}). \label{hj asymp 11 r1r1 o}
\end{align}
\end{enumerate}
\end{lemma}
\begin{proof}
\eqref{hj asymp 1 r1r1 o} directly follows from \eqref{hj asymp 1} and \eqref{def of Fj r1r1 o}. Let us now prove \eqref{hj asymp 2 r1r1 o}. By \eqref{hj asymp 2} and \eqref{def of Fj r1r1 o}, as $n \to + \infty$ with $M' \leq j \leq j_{1,-}$, we have
\begin{align*}
F_{j} = \bigO\bigg(\exp \bigg\{ n(1-\delta) \bigg[ \frac{1}{b}(j/n-j/n \log(\frac{j/n}{b})) - r_{1}^{2b} + 2j/n \log(r_{1}) \bigg] \bigg\} \bigg).
\end{align*}
Since the function $[0,br_{1}^{2b}] \in y \mapsto y-y \log(\frac{y}{b}) + y \log(r_{1}^{2b})$ is increasing, and since $j/n \leq \frac{br_{1}^{2b}}{1+\epsilon}+\bigO(n^{-1})$, we have
\begin{align*}
F_{j} & = \bigO\bigg(\exp \bigg\{ -nr_{1}^{2b}(1-\delta) \frac{\epsilon - \log(1+\epsilon)}{1+\epsilon} \bigg\} \bigg),
\end{align*}
which is \eqref{hj asymp 2 r1r1 o}. Now we prove \eqref{hj asymp 3 r1r1 o}. By \eqref{hj asymp 3} and \eqref{def of Fj r1r1 o}, as $n \to + \infty$ with $j_{1,-} \leq j \leq g_{1,-}$, we have
\begin{align*}
F_{j} = \bigO\bigg( \exp \bigg\{ n(1-\delta) \bigg[ \frac{j/n - j/n \log(\frac{j/n}{b})}{b}-r_{1}^{2b} + 2j/n \log(r_{1}) + \frac{s_{1}+s_{2}}{\sqrt{n}}(br_{1}^{2b}-j/n) \bigg] \bigg\} \bigg).
\end{align*}
For $j_{1,-} \leq j \leq g_{1,-}$, the above exponential attains its maximum at $j=g_{1,-} = n(\frac{br_{1}^{2b}}{1+\frac{M}{\sqrt{n}}}+\bigO(n^{-1}))$, and therefore
\begin{align*}
F_{j} = \bigO\bigg( \exp \bigg\{ -(1-\delta) \bigg[\frac{M^{2}r_{1}^{2b}}{2}-M br_{1}^{2b}(s_{1}+s_{2})\bigg] \bigg\} \bigg).
\end{align*}
Since $M=M' \log n$ and $s_{1},s_{2}$ are fixed, \eqref{hj asymp 3 r1r1 o} holds provided $M'$ is chosen large enough. Now we prove \eqref{hj asymp 4 r1r1 o}. As $n \to + \infty$ with $g_{1,-} \leq j \leq g_{1,+}$, we have
\begin{align*}
& e^{r_{1}^{2b}(1-\log(r_{1}^{2b}))n}e^{M_{j,1} r_{1}^{2b} \log(r_{1}^{2b})\sqrt{n}}e^{-r_{1}^{2b}(\frac{1}{2}+\log(r_{1}^{2b}))M_{j,1}^{2}} e^{-\frac{n}{2}|z|^{2b}}e^{-\frac{n}{2}|w|^{2b}}|z|^{\alpha}|w|^{\alpha}z^{j-1}\overline{w}^{j-1} \\
& = e^{i(j-1)(\theta_{1}-\theta_{2})}e^{-\frac{1}{2}r_{1}^{2b}(M_{j,1}^{2}-2bM_{j,1}(s_{1}+s_{2})+2b^{2}(s_{1}^{2}+s_{2}^{2}))}r_{1}^{-2} \\
& \times \bigg\{ 1 + \frac{1}{\sqrt{n}}\bigg[ s_{1}+s_{2}+ br_{1}^{2b} \bigg( \frac{b(2b-3)}{3}(s_{1}^{3}+s_{2}^{3}) + M_{j,1} \frac{s_{1}^{2}+s_{2}^{2}}{2} - M_{j,1}^{2}(s_{1}+s_{2}) - 2M_{j,1}^{3}\log(r_{1}) \bigg) \bigg] \\
& + \bigO\Big(\frac{1+M_{j,1}^{6}}{n}\Big) \bigg\}.
\end{align*}
Now \eqref{hj asymp 4 r1r1 o} follows directly from a computation using \eqref{hj asymp 4} and \eqref{def of Fj r1r1 o}. Now we prove \eqref{hj asymp 5 r1r1 o}. Using \eqref{hj asymp 5} and \eqref{def of Fj r1r1 o}, as $n \to + \infty$ with $g_{1,+} \leq j \leq j_{1,+}$ we obtain
\begin{align*}
F_{j} & = \bigO \Big( \exp \Big\{ -\sqrt{n}(1-\delta)(j/n-br_{1}^{2b})(s_{1}+s_{2}) \Big\} \Big) \\
& = \bigO \Big( \exp \Big\{ -br_{1}^{2b}(s_{1}+s_{2})M \Big\} \Big).
\end{align*}
Since $s_{1},s_{2}>0$ are fixed and $M=M' \log n$, we can choose $M'$ large enough so that \eqref{hj asymp 5 r1r1 o} holds. The proofs of \eqref{hj asymp 6 r1r1 o}--\eqref{hj asymp 11 r1r1 o} are similar to the proof of \eqref{hj asymp 5 r1r1 o} and are omitted.
\end{proof}

\begin{lemma}\label{lemma:semi hard S4+O1}
As $n \to + \infty$, we have
\begin{align*}
& K_{n}(z,w) = S_{4} + \bigO(1),
\end{align*}
where
\begin{align}
& S_{4} := \sum_{j=g_{1,-}}^{g_{1,+}} e^{i(j-1)(\theta_{1}-\theta_{2})} \sqrt{n}\frac{\sqrt{2} \, b r_{1}^{b-2} e^{-\frac{1}{2}r_{1}^{2b}(M_{j,1}^{2}-2bM_{j,1}(s_{1}+s_{2})+2b^{2}(s_{1}^{2}+s_{2}^{2}))}}{\sqrt{\pi} \mathrm{erfc}(-\frac{M_{j,1}r_{1}^{b}}{\sqrt{2}})}   \nonumber \\
& \qquad \times \bigg( 1 + \frac{1}{\sqrt{n}}\bigg\{\frac{(5M_{j,1}^{2}r_{1}^{2b}-2)e^{-\frac{r_{1}^{2b}M_{j,1}^{2}}{2}}}{3\sqrt{2\pi} r_{1}^{b} \mathrm{erfc}(-\frac{M_{j,1}r_{1}^{b}}{\sqrt{2}}) } - \frac{M_{j,1}}{2}+s_{1}+s_{2} + r_{1}^{2b} \bigg[ \frac{5}{6}M_{j,1}^{3} \nonumber \\
& \qquad - bM_{j,1}^{2}(s_{1}+s_{2}) + bM_{j,1} \frac{s_{1}^{2}+s_{2}^{2}}{2} + \frac{b^{2}(2b-3)}{3}(s_{1}^{3}+s_{2}^{3}) \bigg] \bigg\} \bigg), \label{S4 asymp 1 o}
\end{align}
\end{lemma}
\begin{proof}
By Lemma \ref{lemma: asymp of Fj r1r1 o}, as $n\to+\infty$ we have
\begin{align}\label{lol18}
& K_{n}(z,w) = S_{4} + \bigO\bigg(\sum_{j=g_{1,-}}^{g_{1,+}} \sqrt{n}\frac{e^{-\frac{1}{2}r_{1}^{2b}(M_{j,1}^{2}-2bM_{j,1}(s_{1}+s_{2}))}}{\mathrm{erfc}(-\frac{M_{j,1}r_{1}^{b}}{\sqrt{2}})} \frac{1+M_{j,1}^{6}}{n}\bigg)+\bigO(n^{-50}),
\end{align}
Since $s_{1},s_{2}>0$ are fixed, the function $x \mapsto \frac{e^{-\frac{1}{2}r_{1}^{2b}(x^{2}-2bx(s_{1}+s_{2}))}}{\mathrm{erfc}(-xr_{1}^{b}/\sqrt{2})}$ is $\bigO(e^{-c|x|})$ as $x \to \pm \infty$ (for some $c>0$). Lemma \ref{lemma:Riemann sum} then implies that the $\bigO$-terms in \eqref{lol18} remain bounded as $n\to + \infty$, and the claim follows.
\end{proof}
\subsection{Proof of Theorem \ref{thm:r1 semi-hard}}\label{subsec:6.1}
Here we take $\theta_{2}=\theta_{1}$. Recall also that $s_{1},s_{2}$ in \eqref{def of z1z2 r1 r1 case o} are related to $\mathfrak{s}_{1}, \mathfrak{s}_{2}$ in \eqref{def of z1z2 r1 r1 case o thm intro} via $s_{j} = \mathfrak{s}_{j}/(\sqrt{2}br_{1}^{b})= \mathfrak{s}_{j}/(\tilde{\Delta}Q(r_{1}))^{\frac{1}{2}}$, $j=1,2$.

Using \eqref{S4 asymp 1 o}, we write $S_{4} = \sum_{j=g_{1,-}}^{g_{1,+}}(\sqrt{n} \, h_{1}(M_{j,1})+h_{2}(M_{j,1}))$, where
\begin{align*}
& h_{1}(x) = \frac{\sqrt{2} \, b r_{1}^{b-2} e^{-\frac{1}{2}((\frac{r_{1}^{b}x}{\sqrt{2}}-\mathfrak{s}_{1})^{2}+(\frac{r_{1}^{b}x}{\sqrt{2}}-\mathfrak{s}_{2})^{2})}}{\sqrt{\pi} \mathrm{erfc}(-\frac{x r_{1}^{b}}{\sqrt{2}})}, \\
& h_{2}(x) = h_{1}(x) \bigg\{\frac{(5x^{2}r_{1}^{2b}-2)e^{-\frac{r_{1}^{2b}x^{2}}{2}}}{3\sqrt{2\pi} r_{1}^{b} \mathrm{erfc}(-\frac{xr_{1}^{b}}{\sqrt{2}}) } - \frac{x}{2}+\frac{\mathfrak{s}_{1}+\mathfrak{s}_{2}}{\sqrt{2} \, b r_{1}^{b}}  \nonumber \\
& \hspace{2.3cm} + r_{1}^{2b} \bigg[ \frac{5}{6}x^{3} - x^{2}\frac{\mathfrak{s}_{1}+\mathfrak{s}_{2}}{\sqrt{2}r_{1}^{b}} + x \frac{\mathfrak{s}_{1}^{2}+\mathfrak{s}_{2}^{2}}{4br_{1}^{2b}} + \frac{2b-3}{6\sqrt{2} \, b}\frac{\mathfrak{s}_{1}^{3}+\mathfrak{s}_{2}^{3}}{r_{1}^{3b}} \bigg] \bigg\}.
\end{align*}
Since $\mathfrak{s}_{1},\mathfrak{s}_{2}>0$ are fixed, $h_{1}(x),h_{2}(x) = \bigO(e^{-c |x|})$ as $x \to \pm \infty$ for some $c>0$. Since $M=M' \log n$, by Lemma \ref{lemma:Riemann sum}, we obtain
\begin{align*}
S_{4} & = br_{1}^{2b}\int_{-M}^{M} h_{1}(t)dt \, n + br_{1}^{2b} \int_{-M}^{M} (h_{2}(t)-2th_{1}(t))dt \, \sqrt{n} + \bigO(1), \\
& = br_{1}^{2b}\int_{-\infty}^{+\infty} h_{1}(t)dt \, n + br_{1}^{2b} \int_{-\infty}^{+\infty} (h_{2}(t)-2th_{1}(t))dt \, \sqrt{n} + \bigO(1).
\end{align*}
The claim now follows from a simple change of variables. This finishes the proof of Theorem \ref{thm:r1 semi-hard}.

\subsection{Proof of Theorem \ref{thm:r1r1 semi-hard}}\label{subsec:6.2}
Here we take $\theta_{1}\neq \theta_{2} \mod 2\pi$. Let $p:=e^{-i(\theta_{1}-\theta_{2})}$. By \eqref{S4 asymp 1 o},
\begin{align*}
& S_{4} = \sqrt{n} \sum_{j=g_{1,-}}^{g_{1,+}} p^{-(j-1)} c_{j} = p^{2-g_{1,-}}\sqrt{n} \sum_{\ell=1}^{N} p^{-\ell} c_{\ell}',
\end{align*}
where $N:=g_{1,+}-g_{1,-}+1$, and $c_{j}$, $c_{\ell}'$ are defined by
\begin{align*}
& c_{j} := \frac{\sqrt{2} \, b r_{1}^{b-2} e^{-\frac{1}{2}r_{1}^{2b}(M_{j,1}^{2}-2bM_{j,1}(s_{1}+s_{2})+2b^{2}(s_{1}^{2}+s_{2}^{2}))}}{\sqrt{\pi} \mathrm{erfc}(-\frac{M_{j,1}r_{1}^{b}}{\sqrt{2}})}   \nonumber \\
& \qquad \times \bigg( 1 + \frac{1}{\sqrt{n}}\bigg\{\frac{(5M_{j,1}^{2}r_{1}^{2b}-2)e^{-\frac{r_{1}^{2b}M_{j,1}^{2}}{2}}}{3\sqrt{2\pi} r_{1}^{b} \mathrm{erfc}(-\frac{M_{j,1}r_{1}^{b}}{\sqrt{2}}) } - \frac{M_{j,1}}{2}+s_{1}+s_{2} + r_{1}^{2b} \bigg[ \frac{5}{6}M_{j,1}^{3} \nonumber \\
& \qquad - bM_{j,1}^{2}(s_{1}+s_{2}) + bM_{j,1} \frac{s_{1}^{2}+s_{2}^{2}}{2} + \frac{b^{2}(2b-3)}{3}(s_{1}^{3}+s_{2}^{3}) \bigg] \bigg\} \bigg), \\
& c_{\ell}' := c_{j}, \quad \mbox{where } j = \ell+g_{1,-}-1.
\end{align*}
Summing by parts, we get
\begin{align*}
\sum_{\ell=1}^{N} p^{-\ell} c_{\ell}' = c_{N}'\hat{D}_{N}-c_{1}'\hat{D}_{0}- \frac{c'_{N}-c'_{1}}{1-p^{-1}} + \sum_{\ell=1}^{N-1} (c'_{\ell+1}-c'_{\ell}) \frac{p^{-\ell-1}}{1-p^{-1}}
\end{align*}
where $\hat{D}_{\ell} := \sum_{j=0}^{\ell}p^{-j}$. Note that $c_{1}' = \bigO(e^{-c M})$, $c'_{N} = \bigO(e^{-c M^{2}})$, and $c_{j+1}-c_{j} = \bigO(n^{-\frac{1}{2}})c_{j}$, $c_{j}=\bigO(e^{-c|M_{j,1}|})$. Hence, using several times summation by parts, we get
\begin{align*}
\sum_{\ell=1}^{N} p^{-\ell} c_{\ell}'  \lesssim e^{-c M} + \frac{1}{\sqrt{n} }\sum_{\ell=1}^{N-1} p^{-\ell}c_{\ell}' \lesssim e^{-c M} + \frac{1}{n^{\frac{21}{2}}}\sum_{\ell=1}^{N-21} p^{-\ell}c_{\ell}' \lesssim \frac{M}{n^{10}},
\end{align*}
where we have used that $N=\bigO(M\sqrt{n})$ for the last equality. Hence $S_{4} = \bigO(\frac{M}{n^{10}})$. By Lemma \ref{lemma:semi hard S4+O1}, we thus have $K_{n}(z,w) = \bigO(1)$, which is the statement of Theorem \ref{thm:r1r1 semi-hard}.

\appendix

\section{The semi-hard 1-point function and its universality}\label{Appendix semi-hard universal}

In the following, we consider any radially symmetric potential $V$ which is smooth near a point $r_{1}>0$ in the bulk of the droplet $S=S[V]$; we also assume that
$\Delta V(r_{1})>0$. We redefine $V$ by placing a hard wall along the circle $r=r_{1}$, say
\begin{align*}
Q(z)=\begin{cases}
V(z),& |z|\le r_{1}, \\
+\infty,& |z|>r_{1}.
\end{cases}
\end{align*}

Let $K_n(z,w)$ be the associated correlation kernel, i.e., the reproducing kernel of the space
\begin{align*}
\mathcal{W}_n^h = \{p\cdot e^{-\frac{1}{2} n Q} : p \mbox{ is a holomorphic polynomial with } \mbox{deg } p \leq n-1\} \subset L^2(\C,dA).
\end{align*}
We write $R_n(z)=K_{n}(z,z)$ for the 1-point function.

Let us write $\tilde{K}_n(z,w)$ for the correlation kernel of the rescaled process pertaining to the blow-up map $z\mapsto \sqrt{2n\Delta V(r_{1})}\cdot (r_{1}-z)$, i.e.
\begin{equation}\label{scale}
\tilde{K}_n(z,w)=\frac{1}{2n\Delta V(r_{1})} K_{n}\Big(r_{1}-\frac z {\sqrt{2n\Delta V(r_{1})}},r_{1}-\frac w {\sqrt{2n\Delta V(r_{1})}}\Big).
\end{equation}
Our goal is to study the asymptotics of this rescaled ``semi-hard edge kernel'' for $z,w$ in a compact subset of the right half-plane: $\re z>0$ and $\re w>0$.

\begin{remark} In \cite{AKMW} the following slightly different rescaling was used:
\begin{align}\label{K hat}
\hat{K}_n(z,w) = \frac{1}{n\Delta V(r_{1})} K_{n}\Big(r_{1}+\frac{z}{\sqrt{n\Delta V(r_{1})}},r_{1}+\frac{w}{\sqrt{n\Delta V(r_{1})}}\Big).
\end{align}
Below we recall some theory for the limit of $\hat{K}_n$, and when we translate back to $\tilde{K}_n$, we use:
\begin{align}\label{transformation Ktilde Khat}
\tilde{K}_n(z,w)=\frac{1}{2} \hat{K}_n (-\frac z {\sqrt{2}},-\frac w {\sqrt{2}}), \qquad \hat{K}_n (z,w) = 2\tilde{K}_n(-\sqrt{2}z,-\sqrt{2}w).
\end{align}
\end{remark}

By \cite{AKMW}  we know that there is a sequence of cocycles $c_n(z,w)$ such that the kernels $c_n(z,w)\hat{K}_{n}(z,w)$ converge locally uniformly for $z,w$ in the left half-plane to a limit kernel $\hat{K}(z,w)$ of the form
\begin{equation}\label{struct}
\hat{K}(z,w)=G(z,w)\Psi(z,w) \qquad \re z<0,\re w<0,
\end{equation}
where $G(z,w)=e^{ z\bar{w}-\frac 12|z|^2-\frac 12|w|^2}$ is the ``Ginibre kernel'' while $\Psi(z,w)$ is a function analytic in $z$ and $\bar{w}$ and well-defined for $\re z<0$, $\re w<0$.
Also, the limit kernel gives rise to a solution $\hat{R}(z)$ of the Ward equation
\begin{equation}\label{weq}
\overline{\partial} \hat{C}(z) = \hat{R}(z)-1-\Delta \log \hat{R}(z),\qquad \re z<0,
\end{equation}
where
\begin{align*}
\hat{R}(z)=\hat{K}(z,z)=\Psi(z,z),\qquad \hat{C}(z)=\int_{\re w<0}\frac {|\hat{K}(z,w)|^2}{\hat{K}(z,z)}\frac 1 {z-w}\, dA(w).
\end{align*}
(If $\hat{R}(z)>0$ for one $z$, then $\hat{R}(z)>0$ everywhere in the left half-plane, see \cite{AKMW} for details.)

By standard balayage arguments \cite{ACCL2022 1}, we know that $\hat{R}$ blows up at the hard edge (this is not surprising since $K_{n}(z,z)$ grows like $n^2$ in the hard edge regime), i.e.
\begin{equation}\label{bc}
\lim_{\re z\to 0} \hat{R}(z)=+\infty.
\end{equation}
It is also useful to note that when $\re z\to -\infty$ we have, by well-known bulk asymptotics \cite{AmeurBulk},
\begin{equation}\label{bu}
\lim_{\re z\to -\infty}\hat{R}(z)=1.
\end{equation}

The rotational symmetry of the ensemble also easily implies that $\hat{R}$ is vertical translation invariant (see \cite{AKM} for a detailed proof of this)
\begin{equation}\label{symm}
\hat{R}(z+it)=\hat{R}(z),\qquad \forall t\in \R.
\end{equation}
This is equivalent to saying that the analytic function $\Psi$ in \eqref{struct} takes the form $\Psi(z,w)=\Phi(z+\bar{w})$ (see e.g. \cite{AKM}) where $\Phi(z)$ is an analytic function of one variable,
defined for $\re z<0$.

Thus we have
\begin{equation}\label{smock}
\hat{K}(z,w)=G(z,w)\Phi(z+\bar{w}),\qquad \hat{R}(z)=\Phi(z+\bar{z}).
\end{equation}

We now have the following result, which is implicit in the analysis in \cite{AKMW}.

\begin{proposition}
There is a unique analytic function $\Phi(z)$ for $\re z < 0$ which satisfies the Ward equation \eqref{weq} with boundary conditions \eqref{bc} and \eqref{bu}. This function is given by
\begin{equation}\label{soln}
\Phi(z)=\frac{1}{\sqrt{2\pi}} \int_{-\infty}^{+\infty}\frac{e^{-\frac{1}{2} (z-t)^2}}{\frac{1}{2} \mathrm{erfc} \frac{t}{\sqrt{2}}}\, dt,\qquad (\re z<0).
\end{equation}
\end{proposition}

\begin{proof} That \eqref{soln} solves the Ward equation is shown in \cite[Theorem 8]{AKMW}, and \eqref{bc} and \eqref{bu} are likewise clear for the solution \eqref{soln}.
In the other direction, given an arbitrary solution $\Phi$ there exists, by \cite[Theorem 8]{AKMW}, an interval $I\subset \R$ such that $\Phi$ has the structure
\begin{align*}
\Phi(z)=\frac{1}{\sqrt{2\pi}}\int_I \frac{e^{-\frac{1}{2} (z-t)^2}}{\frac{1}{2} \mathrm{erfc} \frac{t}{\sqrt{2}}}\, dt.
\end{align*}
The conditions $\Phi(x)\to 1$ as $x\to -\infty$ and $\Phi(x)\to+\infty$ as $x\to 0$ easily implies that we must have $I=\R$.
\end{proof}

It follows that the limit kernel $\hat{K}(z,w)$ is uniquely given by
\begin{align*}
\hat{K}(z,w)=e^{z\bar{w}-\frac 1 2|z|^2-\frac 1 2|w|^2}\frac{1}{\sqrt{2\pi}}\int_{-\infty}^{+\infty}\frac {e^{-\frac 1 2 (z+\bar{w}-t)^2}}{\frac 1 2 \mathrm{erfc} \frac{t}{\sqrt{2}}}\, dt,\qquad
\re z<0,\, \re w<0.
\end{align*}

Using \eqref{transformation Ktilde Khat}, we then get for $z$ and $w$ in the right half-plane $\re z>0$, $\re w>0$ that
\begin{align}
\tilde{K}(z,w)&=\frac 1 2 e^{\frac 1 2 z\bar{w}-\frac 1 4|z|^2-\frac 14|w|^2}\frac 1 {\sqrt{2\pi}}\int_{-\infty}^{+\infty}\frac {e^{-\frac{1}{2} (\frac {z+\bar{w}}{\sqrt{2}}+t)^2}}{\frac{1}{2} \mathrm{erfc}\frac{t}{\sqrt{2}}}\, dt\nonumber \\
&=e^{\frac 1 2 z\bar{w}-\frac 1 4|z|^2-\frac 14|w|^2}\frac 1 {\sqrt{\pi}}\int_{-\infty}^{+\infty}\frac {e^{-\frac 14(z+\bar{w}+2t)^2}}{\erfc t}\, dt.\label{gen}
\end{align}

If we choose $\mathfrak{s}_1,\mathfrak{s}_2$ real and positive this gives
\begin{equation}\label{exa}
\tilde{K}(\mathfrak{s}_1,\mathfrak{s}_2) = \int_{-\infty}^{+\infty} \frac{e^{-\frac{(t+\mathfrak{s}_{1})^{2} + (t+\mathfrak{s}_{2})^{2}}{2}}}{\sqrt{\pi} \mathrm{erfc}\,t}dt.
\end{equation}

\paragraph{Universality for non-rotation invariant potentials.} Note that the above arguments work not only for rotationally symmetric potentials, but for arbitrary potentials and scaling limits $\tilde{K}$ which are vertical translation invariant.

Indeed, consider a smooth potential $V$ (not necessarily rotation invariant), fix any point $r_{1}$ and define $Q$ as follows:
\begin{align*}
Q(z)=\begin{cases}
V(z),& z \in D, \\
+\infty,& z \notin D,
\end{cases}
\end{align*}
where $D \subset \mbox{int }S[V]$ is such that $\partial D$ is smooth and passes through $r_{1}$. Assume that
$\Delta V(r_{1})>0$. Let $K_{n}$ be the correlation kernel associated to $Q$, and define $\tilde{K}_{n},\hat{K}_{n}$ as in \eqref{scale} and \eqref{K hat}. By \cite{AKMW} we know that each subsequential limit $\hat{K}$ of the rescaled kernels $\hat{K}_n$ satisfy Ward's equation \eqref{weq} as well as the boundary conditions \eqref{bc} and \eqref{bu}. Hence if $\hat{K}$ is vertical translation invariant, i.e., $\hat{K}(z+it,z+it)=\hat{K}(z,z)$ for all $t\in \R$, \cite[Theorem 8]{AKMW} implies that $\hat{K}$ is uniquely given by \eqref{smock},\eqref{soln}.
In other words, the kernel $\tilde{K}$ in \eqref{gen} is universal among translation invariant scaling limits.

\section{Uniform asymptotics of the incomplete gamma function}\label{appendix:incomplete gamma}

To analyze the large $n$ behavior of $K_{n}(z,w)$, we will use the asymptotics of $\gamma(a,z)$ in various regimes of the parameters $a$ and $z$. These asymptotics are available in the literature and are summarized in the following lemmas.
\begin{lemma}\label{lemma:various regime of gamma}(\cite[formula 8.11.2]{NIST}).
Let $a$ be fixed. As $z \to \infty$,
\begin{align*}
\gamma(a,z) = \Gamma(a) + \bigO(z^{a-1}e^{-z}).
\end{align*}
\end{lemma}
\begin{lemma}\label{lemma: uniform}(\cite[Section 11.2.4]{Temme}).
The following hold:
\begin{align}\label{Temme exact formula}
& \frac{\gamma(a,z)}{\Gamma(a)} = \frac{1}{2}\mathrm{erfc}(-\eta \sqrt{a/2}) - R_{a}(\eta), \qquad R_{a}(\eta) := \frac{e^{-\frac{1}{2}a \eta^{2}}}{2\pi i}\int_{-\infty}^{\infty}e^{-\frac{1}{2}a u^{2}}g(u)du,
\end{align}
where $\mathrm{erfc}$ is given by \eqref{def of erfc}, $g(u) := \frac{dt}{du}\frac{1}{\lambda-t}+\frac{1}{u+i \eta}$,
\begin{align}\label{lol8}
& \lambda = \frac{z}{a}, \quad \eta = (\lambda-1)\sqrt{\frac{2 (\lambda-1-\log \lambda)}{(\lambda-1)^{2}}}, \quad  u=-i(t-1)\sqrt{\frac{2(t-1-\log t)}{(t-1)^{2}}},
\end{align}
and the principal branch is used for the roots. In particular, $t \in \mathcal{L}:=\{\frac{\theta}{\sin \theta} e^{i\theta}: -\pi < \theta < \pi\}$ for $u \in \mathbb{R}$ and $\eta \in \R$ for $\lambda >0$. Furthermore, as $a \to + \infty$, uniformly for $z \in [0,\infty)$,
\begin{align}\label{asymp of Ra}
& R_{a}(\eta) \sim \frac{e^{-\frac{1}{2}a \eta^{2}}}{\sqrt{2\pi a}}\sum_{j=0}^{\infty} \frac{c_{j}(\eta)}{a^{j}},
\end{align}
where all coefficients $c_{j}(\eta)$ are bounded functions of $\eta \in \mathbb{R}$ (i.e. bounded for $\lambda \in (0,\infty)$). The coefficients $c_{0}(\eta)$ and $c_{1}(\eta)$ are explicitly given by (see \cite[p. 312]{Temme})
\begin{align}\label{def of c0 and c1}
c_{0}(\eta) = \frac{1}{\lambda-1}-\frac{1}{\eta}, \qquad c_{1}(\eta) = \frac{1}{\eta^{3}}-\frac{1}{(\lambda-1)^{3}}-\frac{1}{(\lambda-1)^{2}}-\frac{1}{12(\lambda-1)}.
\end{align}
\end{lemma}
Recall that $\mathrm{erfc}$ satisfies \cite[7.12.1]{NIST}
\begin{align}\label{large y asymp of erfc}
& \mathrm{erfc}(y) = \frac{e^{-y^{2}}}{\sqrt{\pi}}\bigg( \frac{1}{y}-\frac{1}{2y^{3}}+\frac{3}{4y^{5}}-\frac{15}{8y^{7}} + \bigO(y^{-9}) \bigg), & & \mbox{as } y \to + \infty,
\end{align}
and $\mathrm{erfc}(-y) = 2-\mathrm{erfc}(y)$. By combining Lemma \ref{lemma: uniform} with the above large $y$ asymptotics of $\mathrm{erfc}(y)$, we get the following.
\begin{lemma}\label{lemma: asymp of gamma for lambda bounded away from 1}
Let $\eta$ be as in \eqref{lol8}.
\item[(i)] Let $\delta >0$ be fixed. As $a \to +\infty$, uniformly for $\lambda \geq 1+\delta$,
\begin{align*}
\frac{\gamma(a,\lambda a)}{\Gamma(a)} = 1 + \frac{e^{-\frac{a\eta^{2}}{2}}}{\sqrt{2\pi}} \bigg( \frac{-1}{\lambda-1}\frac{1}{\sqrt{a}}+\frac{1+10\lambda+\lambda^{2}}{12(\lambda-1)^{3}} \frac{1}{a^{3/2}} + \bigO(a^{-5/2}) \bigg).
\end{align*}
\item[(ii)] As $a \to +\infty$, uniformly for $\lambda$ in compact subsets of $(0,1)$,
\begin{align*}
\frac{\gamma(a,\lambda a)}{\Gamma(a)} = \frac{e^{-\frac{a\eta^{2}}{2}}}{\sqrt{2\pi}} \bigg( \frac{-1}{\lambda-1}\frac{1}{\sqrt{a}}+\frac{1+10\lambda+\lambda^{2}}{12(\lambda-1)^{3}} \frac{1}{a^{3/2}} + \bigO(a^{-5/2}) \bigg).
\end{align*}
\end{lemma}

\paragraph{Conflict of interest statement.} The authors have no conflict of interest to disclose.

\paragraph{Data availability statement.} There is no data associated to this work.

\subsection*{Acknowledgement.} CC acknowledges support from the Swedish Research Council, Grant No. 2021-04626. We are grateful to Sung-Soo Byun for useful remarks.


\footnotesize
\begin{thebibliography}{99}
\bibitem{AR2017} K. Adhikari, N.K. Reddy, Hole probabilities for finite and infinite Ginibre ensembles, \textit{Int. Math. Res. Not. IMRN} (2017), no.21, 6694--6730.

\bibitem{AB2012} G. Akemann and Z. Burda, Universal microscopic correlation functions for products of independent Ginibre matrices, \textit{J. Phys. A} \textbf{45} (2012), no.46, 465201, 18 pp.

\bibitem{ACV} G. Akemann, M. Cikovic and M. Venker, Universality at weak and strong non-Hermiticity beyond the elliptic Ginibre ensemble, \textit{Comm. Math. Phys.} \textbf{362} (2018), no. 3, 1111--1141.

\bibitem{ADM} G. Akemann, M. Duits and L.D. Molag, The elliptic Ginibre ensemble: a unifying approach to local and global statistics for higher dimensions, \textit{J. Math. Phys.} \textbf{64} (2023), no. 2, 39 pp.

\bibitem{APS2009} G. Akemann, M.J. Phillips and L. Shifrin, Gap probabilities in non-Hermitian random matrix theory, \textit{J. Math. Phys.} \textbf{50} (2009), no. 6, 063504, 32 pp. 

\bibitem{AV2003} G. Akemann and G. Vernizzi, Characteristic polynomials of complex random matrix models, \textit{Nuclear Phys. B} \textbf{660} (2003), no. 3, 532--556.

\bibitem{A1} Y. Ameur, A localization theorem for the planar Coulomb gas in an external field, \textit{Electron. J. Probab.} \textbf{26} (2021), Paper No. 46, 21 pp.

\bibitem{AmeurBulk} Y. Ameur, Near-boundary asymptotics for correlation kernels, \textit{J. Geom. Anal.} \textbf{23} (2013), no. 1, 73--95.



\bibitem{AB2021} Y. Ameur and S.-S. Byun, Almost-Hermitian random matrices and bandlimited point processes, \textit{Anal. Math. Phys.} \textbf{13} (2023), no.3, Paper No. 52, 57 pp.


%\bibitem{ACC23} Y. Ameur, C. Charlier, J. Cronvall, Free energy and fluctuations in the random normal matrix model with spectral gaps, arXiv: 2312.13904.

\bibitem{ACC2022} Y. Ameur, C. Charlier and J. Cronvall, The two-dimensional Coulomb gas: fluctuations through a spectral gap, arXiv: 2210.13959.

\bibitem{ACCL2022 1} Y. Ameur, C. Charlier, J. Cronvall and J. Lenells, Exponential moments for disk counting statistics at the hard edge of random normal matrices, \textit{J. Spectr. Theory} \textbf{13} (2023), 841--902.

\bibitem{ACCL2022 2} Y. Ameur, C. Charlier, J. Cronvall and J. Lenells, Disk counting statistics near hard edges of random normal matrices: the multi-component regime, \textit{Adv. Math.} \textbf{441} (2024), Paper No. 109549.

\bibitem{ACM2024} Y. Ameur, C. Charlier and P. Moreillon,  Eigenvalues of truncated unitary matrices: disk counting statistics, \textit{Monatsh. Math.} \textbf{204} (2024), no.~2, 197--216.

\bibitem{AC} Y. Ameur and J. Cronvall, Szeg\H{o} type asymptotics for the reproducing kernel in spaces of full-plane weighted polynomials, \textit{Comm. Math. Phys.} \textbf{398} (2023), 1291--1348.

\bibitem{AKM} Y. Ameur, N.-G. Kang and N. Makarov, Rescaling Ward identities in the random normal matrix model, \textit{Constr. Approx.} \textbf{50} (2019), 63--127.

\bibitem{AKMW} Y. Ameur, N.-G., Kang, N. Makarov and A. Wennman, Scaling limits of random normal matrix processes at singular boundary points, \textit{J. Funct. Anal.} \textbf{278} (2020), 108340.

\bibitem{AKS20} Y. Ameur, N.-G. Kang and S.-M. Seo, On boundary confinements for the Coulomb gas, \textit{Anal. Math. Phys.} \textbf{10} (2020), no.4, Paper No. 68, 42 pp.

\bibitem{AKS2018} Y. Ameur, N-G. Kang and S-M. Seo, The random normal matrix model: insertion of a point charge, \textit{Potential Anal.} \textbf{58} (2023), no. 2, 331--372.

\bibitem{AHM Duke} Y. Ameur, H. Hedenmalm and N. Makarov, Fluctuations of eigenvalues of random normal matrices, \textit{Duke Math. J.} \textbf{159} (2011), no. 1, 31--81.

%\YA{\bibitem{AT} Y. Ameur and E. Troedsson, Remarks on the one-point density of Hele-Shaw $\beta$-ensembles, arxiv: 2402.13882.}

%\bibitem{BBLM2015} F. Balogh, M. Bertola, S.-Y. Lee and K.T.-R. McLaughlin, Strong asymptotics of the orthogonal polynomials with respect to a measure supported on the plane, \textit{Comm. Pure Appl. Math.} \textbf{68} (2015), no. 1, 112--172.

%\bibitem{BBS2008} R. Berman, B. Berndtsson and J. Sj\"{o}strand, Asymptotics of Bergman kernels, \textit{Ark. Mat.} \textbf{46} (2008).

\bibitem{Berezin} S. Berezin, Functional limit theorems for constrained Mittag--Leffler ensemble in hard edge scaling, arXiv:2308.12658.

\bibitem{Byun} S.-S. Byun, Planar equilibrium measure problem in the quadratic fields with a point charge, \textit{Comput. Methods Funct. Theory} \textbf{24} (2024), no.2, 303--332.

\bibitem{BC2022} S.-S. Byun and C. Charlier, On the characteristic polynomial of the eigenvalue moduli of random normal matrices, arXiv:2205.04298 (to appear in \textit{Constr. Approx.}).

\bibitem{BF2022} S.-S. Byun and P.J. Forrester, Progress on the study of the Ginibre ensembles I: GinUE, arXiv:2211.16223.

%\bibitem{BF2023} S.-S. Byun and P.J. Forrester, Progress on the study of the Ginibre ensembles II: GinOE and GinSE, arXiv:2301.05022.

\bibitem{BLY2021} S.-S. Byun, S.-Y. Lee and M. Yang, Lemniscate ensembles with spectral singularity, arXiv:2107.07221.

\bibitem{SP2024} S.-S. Byun and S. Park, Large gap probabilities of complex and symplectic spherical ensembles with point charges, arXiv:2405.00386.

\bibitem{BS2021} S.-S. Byun and S.-M. Seo, Random normal matrices in the almost-circular regime, \textit{Bernoulli} \textbf{29} (2023), no. 2, 1615--1637.

\bibitem{BY2022} S.-S. Byun and M. Yang, Determinantal Coulomb gas ensembles with a class of discrete rotational symmetric potentials, \textit{SIAM J. Math. Anal.} \textbf{55} (2023), no.6, 6867--6897.

\bibitem{CZ1998} L.-L. Chau and O. Zaboronsky, On the structure of correlation functions in the normal matrix model, \textit{Comm. Math. Phys.} \textbf{196} (1998), no. 1, 203--247.

\bibitem{CharlierFH} C. Charlier, Asymptotics of determinants with a rotation-invariant weight and discontinuities along circles, \textit{Adv. Math.} \textbf{408} (2022), Paper No. 108600, 36 pp.

\bibitem{CharlierAnnuli} C. Charlier, Large gap asymptotics on annuli in the random normal matrix model, \textit{Math. Ann.} \textbf{388} (2024), no. 4, 3529--3587.

\bibitem{CBalayage} C. Charlier, Hole probabilities and balayage of measures for planar Coulomb gases, arXiv: 2311.15285.

\bibitem{CD2019} C. Charlier and A. Doeraene, The generating function for the Bessel point process and a system of coupled Painlev\'{e} V equations, \textit{Random Matrices Theory Appl.} \textbf{8} (2019), no. 3, 1950008, 31 pp.

\bibitem{CLBalayage} C. Charlier and J. Lenells, Balayage of measures: behavior near a corner, arXiv:2403.02964.

\bibitem{CPR} Ph. Choquard, B. Piller and R. Rentsch, On the dielectric susceptibility of classical Coulomb systems. II, \textit{J. Stat. Phys.} \textbf{46} (1986) 599-633.

\bibitem{CFLV2} F.D. Cunden, P. Facchi, M. Ligab\`{o} and P. Vivo, Third-order phase transition: random matrices and screened Coulomb gas with hard walls, \textit{J. Stat. Phys.} \textbf{175} (2019), no. 6, 1262--1297.

\bibitem{CMV} F.D. Cunden, F. Mezzadri and P. Vivo, Large deviations of radial statistics in the two-dimensional one-component plasma, \textit{J. Stat. Phys.} \textbf{164} (2016), no. 5, 1062--1081.

\bibitem{CFLV1} F.D. Cunden, P. Facchi, M. Ligab\`{o} and P. Vivo, Universality of the third-order phase transition in the constrained Coulomb gas, \textit{J. Stat. Mech. Theory Exp.} (2017), no. 5, 053303, 18 pp.

\bibitem{C} J. Cronvall, to appear.

%\bibitem{DeanoSimm} A. Dea\~{n}o and N. Simm, Characteristic polynomials of complex random matrices and Painlev\'{e} transcendents, \textit{Int. Math. Res. Not. IMRN} \textbf{2022} (2022), no. 1, 210--264.

\bibitem{Deift} P. Deift, Orthogonal polynomials and random matrices: a Riemann-Hilbert approach, Courant Lecture Notes in Mathematics, 3.

\bibitem{ForresterHoleProba} P.J. Forrester, Some statistical properties of the eigenvalues of complex random matrices, \textit{Phys. Lett. A} \textbf{169} (1992), no. 1-2, 21--24.

\bibitem{Fo} P.J. Forrester, Log-gases and Random Matrices (LMS-34), \textit{Princeton University Press}, Princeton 2010.

\bibitem{FH} P.J. Forrester and G. Honner, Exact statistical properties of the zeros of complex random polynomials, \textit{J. Phys. A.} \textbf{41} (1999), 375003.

\bibitem{FJ} P.J. Forrester and B. Jancovici, Two-dimensional one-component plasma
in a quadrupolar field, \textit{Internat. J. Modern Phys. A} \textbf{11} (1996), no. 5, 941--949.

\bibitem{FKS1997} Y.V. Fyodorov, B.A. Khoruzhenko and H.-J. Sommers, Almost Hermitian random matrices: crossover from Wigner-Dyson to Ginibre eigenvalue statistics, \textit{Phys. Rev. Lett.} \textbf{79} (1997), no. 4, 557--560.

\bibitem{FS1997} Y.V. Fyodorov and H.-J. Sommers, Statistics of resonance poles, phase shifts and time delays in quantum chaotic scattering: random matrix approach for systems with broken time-reversal invariance, Quantum problems in condensed matter physics \textit{J. Math. Phys.} \textbf{38} (1997), no.4, 1918--1981.

\bibitem{HW} H. Hedenmalm and A. Wennman, Planar orthogonal polynomials and boundary universality in the random normal matrix model, \textit{Acta Math.} \textbf{227} (2021), 309-406.

\bibitem{J} B. Jancovici, Classical Coulomb systems near a plane wall. I, \textit{J. Stat. Phys.} \textbf{28} (1982), 43--65.

\bibitem{J2} B. Jancovici, Classical Coulomb systems near a plane wall. II, \textit{J. Stat. Phys.} \textbf{29} (1982), 263--280.

\bibitem{JV} K. Johansson and F. Viklund, Coulomb gas and Grunsky operator on a Jordan domain with corners, arXiv:2309.00308.

\bibitem{K2005} E. Kanzieper, \textit{Exact replica treatment of non-Hermitean complex random matrices}, Frontiers in Field Theory, edited by O. Kovras, Ch. 3, pp. 23 -- 51 (Nova Science Publishers, NY 2005).

\bibitem{KLM2023} T. Kr\"{u}ger, S.-Y. Lee and M. Yang, Local Statistics in Normal Matrix Models with Merging Singularity, arXiv:2306.12263.

\bibitem{LMS2018} B. Lacroix-A-Chez-Toine, S.N. Majumdar and Gr\'{e}gory Schehr, Rotating trapped fermions in two dimensions and the complex Ginibre ensemble: Exact results for the entanglement entropy and number variance, \textit{Phys. Rev. A} \textbf{99} (2019), 021602.

\bibitem{LM} S.-Y. Lee and N. Makarov, Topology of quadrature domains, \textit{J. Am. Math. Soc.} \textbf{29} (2016), 333--369.

\bibitem{LR} S.-Y. Lee and R. Riser, Fine asymptotic behavior of random normal matrices: ellipse case, \textit{J. Math. Phys.} \textbf{57} (2016), no.2, 023302, 29 pp.

\bibitem{Mehta} M.L. Mehta, Random matrices. \textit{Pure and Applied Mathematics} (Amsterdam), Vol. 142, 3rd ed., Elsevier/Academic Press, Amsterdam, 2004.

\bibitem{NAKP2020} T. Nagao, G. Akemann, M. Kieburg and I. Parra, Families of two-dimensional Coulomb gases on an ellipse: correlation functions and universality, \textit{J. Phys. A} \textbf{53} (2020), no. 7, 075201, 36 pp.

\bibitem{NIST} F.W.J. Olver, A.B. Olde Daalhuis, D.W. Lozier, B.I. Schneider, R.F. Boisvert, C.W. Clark, B.R. Miller and B.V. Saunders, NIST Digital Library of Mathematical Functions. http://dlmf.nist.gov/, Release 1.0.13 of 2016-09-16.

\bibitem{RB} M.L. Rosinberg and L. Blum, The ideally polarizable interface: A solvable model and general sum rules, \textit{J. Chem. Phys.} \textbf{81} (1984), 3700--3714.

\bibitem{SaTo} E. B. Saff and V. Totik, {\em Logarithmic Potentials with External Fields},  Grundlehren der Mathematischen Wissenschaften, Springer-Verlag, Berlin, 1997.

\bibitem{Seo0} S.-M. Seo, Edge scaling limit of the spectral radius for random normal matrix ensembles at hard edge, \textit{J. Stat. Phys.} \textbf{181} (2020), no. 5, 1473--1489.

\bibitem{Seo2021} S.-M. Seo, Edge behavior of two-dimensional Coulomb gases near a hard wall, \textit{Ann. Henri Poincar\'{e}} \textbf{23} (2021), 2247--2275.

\bibitem{Sm} E.R. Smith, Effects of surface charge on the two-dimensional one-component plasma: I. Single double
layer structure, \textit{J. Phys. A: Math. Gen.} \textbf{15} (1982), 1271--1281.

\bibitem{TV2015} T. Tao and V. Vu, Random matrices: universality of local spectral statistics of non-Hermitian matrices, \textit{Ann. Probab.} \textbf{43} (2015), no.2, 782--874.

\bibitem{Temme} N.M. Temme, \textit{Special functions: An introduction to the classical functions of mathematical physics}, John Wiley \& Sons (1996).

\bibitem{TW} C.A. Tracy and H. Widom, Level spacing distributions and the Bessel kernel, \textit{Comm. Math. Phys.}
\textbf{161} (1994), 289--309.

%\YA{\bibitem{ZW} A. Zabrodin, A. and P. Wiegmann, Large $N$ expansion for the 2D Dyson gas, \textit{J. Phys. A} \textbf{39} (2006), no.28, 8933--8964.}

\bibitem{Z} A. Zabrodin, Random matrices and Laplacian growth, \textit{The Oxford handbook of random matrix theory}, 802-823.
Oxford University Press, Oxford, 2011.

\bibitem{ZS2000} K. \.{Z}yczkowski and H.-J. Sommers, Truncations of random unitary matrices, \textit{J. Phys. A} \textbf{33} (2000), no. 10, 2045--2057.
\end{thebibliography}
\end{document}